\documentclass[a4paper,11pt]{article}
\usepackage[T1]{fontenc}
\usepackage[latin1]{inputenc}
\usepackage[english]{babel}
\usepackage{amsfonts,amssymb,amsmath,amsthm,bm}
\usepackage{latexsym}
\usepackage{eurosym}
\usepackage[usenames,dvipsnames]{color}

\usepackage{graphicx,subfigure}
\usepackage{dsfont}
\usepackage{geometry}
\usepackage{epsfig,epstopdf}
\usepackage{hyperref}
\hypersetup{
    colorlinks,
    citecolor=blue,
    filecolor=blue,
    linkcolor=blue,
    urlcolor=blue
}
\usepackage[normalem]{ ulem }
\usepackage{cancel}
\usepackage{float,caption}

\DeclareGraphicsExtensions{.png,.jpg,.bmp,.eps,.pdf}
\graphicspath{{./fig/}}

\newtheorem{definition}{\textbf{Definition}}
\newtheorem{remark}{\textbf{Remark}}
\newtheorem{prop}{\textbf{Proposition}}
{\everymath{\displaystyle\everymath{}}\array}
{\endarray}

\newcommand{\alias}[2]{
\providecommand{#1}{}
\renewcommand{#1}{#2}
}

\alias{\P}{\mathbb{P}}
\alias{\N}{\mathcal{N}}
\alias{\L}{\mathcal{L}}
\alias{\Z}{\mathbb{Z}}
\alias{\Q}{\mathbb{Q}}
\alias{\R}{\mathbb{R}}
\alias{\C}{\mathcal{C}}
\alias{\T}{\mathbb{T}}
\alias{\E}{\mathbb{E}}
\alias{\H}{\mathcal{H}}
\alias{\1}{\mathbf{1}}
\alias{\B}{\mathcal{B}}
\alias{\M}{\mathcal{M}}
\alias{\G}{\mathcal{G}}
\alias{\V}{\mathbb{V}}
\alias{\Y}{Y_{\bullet}}
\newcommand{\x}{x^{\pm}}
\newcommand{\xs}{x^{\pm\ast}}
\newcommand{\Dd}{\widehat{\omega}}
\newcommand{\Dq}{\widehat{v}}
\newcommand{\w}{\omega}

\newcommand{\nc}{\newcommand}
\nc{\bl}{\begin{align*}}\nc{\el}{\end{align*}}
\nc{\ba}{\begin{array}}\nc{\ea}{\end{array}}
\nc{\be}{\begin{equation}} \nc{\ee}{\end{equation}}
\nc{\bea}{\begin{eqnarray}} \nc{\eea}{\end{eqnarray}}
\nc{\bean}{\begin{eqnarray*}} \nc{\eean}{\end{eqnarray*}}
\nc{\bel}{\begin{align}} \nc{\eel}{\end{align}}
\nc{\beln}{\begin{align*}} \nc{\eeln}{\end{align*}}
\nc{\bit}{\begin{itemize}}\nc{\eit}{\end{itemize}}
\nc{\bu}{\bullet} \nc{\nn}{\nonumber}
\nc{\cA}{{\mathcal A}} \nc{\cB}{{\mathcal B}} \nc{\cC}{{\mathcal
C}} \nc{\cD}{{\mathcal D}} \nc{\bbD}{\mathbb{D}}
\nc{\cG}{{\mathcal G}} \nc{\cF}{{\mathcal F}} \nc{\cS}{{\mathcal
S}} \nc{\cU}{{\mathcal U}} \nc{\cH}{{\mathcal H}}
\nc{\cK}{{\mathcal K}} \nc{\cM}{{\mathcal M}} \nc{\cO}{{\mathcal
O}} \nc{\cP}{{\mathcal P}} \nc{\bbE}{\mathbb{E}}
\nc{\bbEQ}{\mathbb{E}_{\mathbb{Q}}} \nc{\eps}{\varepsilon}
\nc{\bbEP}{\mathbb{E}_{\mathbb{P}}}\nc{\bbL}{\mathbb{L}}
\nc{\bbP}{\mathbb{P}} \nc{\bbQ}{\mathbb{Q}} \nc{\del}{\partial}
\nc{\Om}{\Omega} \nc{\om}{\omega} \nc{\bbR}{\mathbb{R}}
\nc{\bbC}{\mathbb{C}} \nc{\bfr}{\begin{flushright}}
\nc{\efr}{\end{flushright}} \nc{\dXt}{\Delta X_{t}}
\nc{\dXs}{\Delta X_{s}} \nc{\bs}{\blacksquare} \nc{\dX}{\Delta X}
\nc{\dY}{\Delta Y}
\nc{\dnkx}{\left(X(T^{n}_{k})-X(T^{n}_{k-1})\right)}
\nc{\esssup}{\mathrm{ess}\mbox{ }\mathrm{sup}}
\nc{\essinf}{\mathrm{ess}\mbox{ } \mathrm{inf}}
\nc{\dhats}{\widehat{\delta_s}}

\nc{\chf}{\mbox{$\mathbf1$}}
\nc{\ind}{\mathds{1}}

\begin{document}
\author{Ren\'e A\"id  \quad  Luciano Campi \quad Liangchen Li \quad Mike Ludkovski    }
\title{An Impulse--Regime Switching Game Model \\ of Vertical Competition}
\maketitle
\begin{abstract}
We study a new kind of non-zero-sum stochastic differential game with mixed impulse/switching controls, motivated by strategic competition in commodity markets. A representative upstream firm produces a commodity that is used by a representative downstream firm to produce a final consumption good. Both firms can influence the price of the commodity. By shutting down or increasing generation capacities, the upstream firm influences the price with impulses. By switching (or not) to a substitute, the downstream firm influences the drift of the commodity price process. We study the resulting impulse--regime switching game between the two firms, focusing on explicit threshold-type equilibria. Remarkably, this class of games naturally gives rise to multiple Nash equilibria, which we obtain via a verification based approach. We exhibit three types of equilibria depending on the ultimate  number of switches by the downstream firm (zero, one or an infinite number of switches). We illustrate the diversification effect provided by vertical integration in the specific case of the crude oil market. Our analysis shows that the diversification gains strongly depend on the pass-through from the crude price to the gasoline price.
\end{abstract}
\tableofcontents

\section{Introduction}

Since Hotelling's (1931) \cite{H31} seminal study of commodity prices, considerable efforts have been undertaken to understand the dynamics of the equilibrium price of commodities and in particular, its long--run properties. The cyclical nature of price dynamics is driven by the substitution effect, whereby consumers will switch to a different commodity if prices rise too high. In a deterministic setting the switching time to the substitute is simple to analyze, but with the stochastic economic cycle consumers face a huge challenge in determining when is the appropriate moment to switch. The succession of booms and busts of commodity prices complicates the switching timing. In the long--run, production capacities adapt to demand and make the price oscillate around a long--term equilibrium.  Indeed, the long--run behaviour of commodity prices exhibits super--cycle patterns. The econometric studies in Leon and Soto (1997) \cite{LS97}, Erten and Ocampo (2012) \cite{EO12},  Jacks (2013) \cite{J13}  and more recently Stuemer (2018) \cite{S18}, all find the presence of super--cycles of several decades in the price of commodities. This phenomenon makes one wonder whether it is even necessary for the consumers to ever switch and whether it is not preferable to just wait for the prices to crash again.

In this paper we design a dynamic model of competition between production and consumption of a commodity used as an intermediate good, allowing to draw conclusions on the long--run dynamics of the commodity price.  In our model, two factors drive the price of the commodity: on the one hand, short--term but persistent shocks of demand and/or production, and on the other hand, strategic decisions of the (representative) upstream production firm and of the (representative) downstream consumer firm. The upstream producer extracts the commodity at cost $c_p$ and sells it for a price $X$. The downstream industry buys the commodity and converts it into a final good that has a price $P$, non--decreasing in $X$. This framework covers a wide range of industries. One might think for example,  of the agricultural sector where soy enters as an input for the food industry to produce a large range of consumer goods. In the aluminum industry, upstream smelters produce aluminum to be used by the automotive and transportation industries. In the oil industry, the crude is extracted by production firms, then transformed into gasoline and kerosene by downstream refineries, and then consumed in the retail market. For the sake of simplicity, we identify the downstream firm that transforms the commodity with the final consumer and this downstream firm's profit with the consumer's surplus.

We focus on the role of the commodity price $X$ that intrinsically creates competition between the representative agents of producers and consumers. In a nutshell, producers prefer high price  $X$, while consumers prefer low price $X$. This competition is dynamic and manifests itself through strategic price effects actuated by the two industries. Therefore, $X$ is (partially) jointly controlled by the producers/consumers, leading to game--theoretic impacts.

On the upstream production side, the producer needs the commodity price $X$ to be high enough to make a profit margin. We suppose that the dynamics of investment and disinvestment in upstream production is driven by  \emph{production capacity shocks} that cause \emph{jumps} in the price $X$. This assumption is consistent with the theory of real options that predicts the existence of threshold prices triggering the  decision of entry and the exit from the market (see MacDonald and Siegel (1986) \cite{MS86} and Dixit and Pindyck (1994) \cite{DP94}). It is also consistent with the observations of quick swings in investment and disinvestement in production, see e.g.~the boom and bust of commodity prices in 2008--10.

On the downstream consumer side, consumers induce a long--term effect on the commodity price only if they switch to a substitute, and they switch to a substitute  only if they anticipate that $X$ will remain high enough for a long time.  The downstream side faces slower dynamics because it involves the transformation of many local installations using the commodity. To have an example in mind, one may think of the thousands of adjustments required to change heating systems in buildings, or of the slow effect of the energy saving programs launched by OECD  after the 1970s oil shock.  Thus, in our model, the downstream market for the final good can be in {\em contraction} or {\em expansion} regime. The contraction regime corresponds to a decreasing demand for the primary commodity, i.e.~the market is abandoning the use of the commodity for a substitute, while the expansion mode corresponds to an increasing demand for the commodity. Depending on the state of the downstream retail market, the \emph{drift}  of the commodity price takes either a constant positive value in the expansion mode or a negative value in the contraction mode. Because such consumer shifts are slow and expensive, the state is persistent (i.e.~piecewise constant in time) and changing the state of the final good market incurs heavy switching costs. This toggling of the price trend can be interpreted as endogenous \emph{regime-switching}, a common way of modeling commodity prices through the business cycle. Beyond the impact of producers and consumers decisions, the commodity price is subject to exogenous short--term stochastic shocks, captured through a Brownian motion driving risk factor.

Our aim is to construct and characterize the dynamic equilibrium in the commodity market due to this vertical competition. Our major contribution is to provide an endogenous, game-theoretic basis for two key stylized features of commodity markets: (i) super--cycles that manifest as long--term mean-reversion; (ii) fundamental impact of supply and demand that maintains the price in a range of values rather than a single equilibrium value. Furthermore, our model allows for three types of equilibria depending on the number of demand switches undertaken by the consumer at equilibrium: zero, one, or an infinite number of switches.  All equilibria exhibit the latter qualitative properties. Besides, the higher the consumer's switching cost, the more she is compelled to endure an unfavorable range of prices.

Along the way, we also make mathematical contributions to the literature on non-zero-sum stochastic games (see Martyr and Moriarty (2017) \cite{MM17}, Atard (2018) \cite{A18}, De~Angelis et al. (2018) \cite{DFM18}, A\"id et al. (2020) \cite{ABCCV20}). To our knowledge ours is the first paper that: (i) considers a mixed impulse-control/switching-control stochastic game; (ii) explicitly constructs impulse-switching threshold-type equilibria in non-zero-sum games;  provides new verification theorems regarding best-response strategies for (iii) an impulsing agent in a regime-switching setting and (iv) switching agent with an impulsed state process. While our solution is non exhaustive in the sense that we a priori focus on a special class of equilibria (leaving open the question of existence of other equilibrium families), it is highly tractable. Namely, we are able to provide closed-form description of the dynamic equilibrium, offering precise quantitative insights regarding the producer and consumer roles and their equilibrium behavior.

To emphasize the latter point, beyond several synthetic examples that illustrate and visualize our model features, we also present a detailed case-study of the diversification effect provided by vertical integration in the crude oil market circa 2019, viewed as a competition between crude oil producers and oil refiners that convert crude into gasoline and other consumer goods.  Indeed, the industrial organization of upstream and downstream segments is an important concern both for the anti-trust regulators and for the firms themselves. It is reflected in the extensive economic literature on the subject and in its persistent presence within the political debate (see Lafontaine and Slade (2007) \cite{LS07} for a review of the topic). From a firm's perspective, vertical integration brings multiple  virtues, including  the potential to reduce the long--term exposure to commodity price fluctuations. See for example Helfat and Teece (1987) \cite{HT87} for an empirical estimation of the hedge procured by vertical integration in the oil business. In our case study, we consider the generic type of equilibrium and a small downstream firm asking herself whether she has an interest in getting more vertically integrated. We show that the gains from integration are directly linked to the pass-through parameter that links the crude oil price to the retail gasoline price. The higher this pass-through, the higher production activity dominates the retail activity both in terms of expected rate of profit and the standard deviation of rate of profit.

The rest of the paper is organized as follows. Section \ref{sec:model} sets up the competitive producer-consumer commodity market. Section~\ref{sec:eqm} then constructs the respective threshold-type impulse-switching equilibria by considering the producer and consumer best-response strategies. Section~\ref{sec:eqm-types} illustrates and discusses the different types of emergent equilibria using toy examples. Section~\ref{sec:oil-study} presents the above vertical integration case study and Section~\ref{sec:conclude} concludes. All the proofs, as well as additional comparative statics, are delegated to Section~\ref{sec:proofs}.

\section{The model}\label{sec:model}

\subsection{Description}

We use $(X_t)$ to denote the (pre-equilibrium) commodity price, modeled as a continuous-time stochastic process. The two players are denoted as $p$roducer and $c$onsumer. In what follows sub-index $p$ (resp.~$c$) in the notation will always refer to the producer (resp.~consumer).  The market involves the original raw commodity that is being produced and the goods market (e.g.~gasoline).
The producer extracts the commodity at cost $c_p$ and sells it for price $x$. The consumer buys it for price $x$, converts it into a final good, and sells it for price $P$.

\paragraph{Profit rates:} The price $x$ of the commodity influences the \emph{volume} of trade, captured by the demand function $D_p (x)$. A similar phenomenon plays out in the final-good market: the goods price $P$ leads to sales volume $D_c (P)$. Since the consumer is in effect the intermediary between the commodity and the goods market, she will pass some of her input price shocks to the output price $P \equiv P(x)$.

We ignore the players' fixed costs because they can be considered to be integrated in the investment costs, and concentrate on the variable costs and revenues that are driven by the respective input/output prices.
 Based on the above discussion, the instantaneous profit rate of the producer is
\begin{align}
\pi_p (x) := (x - c_p) D_p (x).
 \end{align}
Let  $c_c$ be the processing/conversion cost from input commodity to final good and $\alpha$ be the respective conversion factor, so that one unit of commodity becomes $\alpha$ units of the final good (e.g.~barrels of crude oil, converted into barrels of gasoline). Then
 the instantaneous profit rate of the consumer is
\begin{align}\pi_c (x) :=  D_c (P) P - \frac{D_c (P)}{\alpha} (x + c_c).\end{align}
 We note that while the consumer has market power, he is not the only user of the commodity (e.g.~crude oil is also used by the petrochemical industry), so there is no direct link between production volume and consumption volume. Thus, while there is a physical link between the consumer input volume $D_c(P)/\alpha$ and her output volume $D_c(P)$, there is no direct link between $D_c(P)/\alpha$ and aggregate commodity demand $D_p(x)$.

We shall consider linear inverse demand $$D_p (x) = d_0 - d_1 x.$$
If we further assume that $P(x) = p_0 + p_1 x$ (the price of the final good is linearly proportional to the commodity price), and $D_c (P) = d'_0 - d'_1 P$ (final good demand is linearly decreasing in its price $P$), the profit rate of the consumer becomes:
\begin{align}\notag
\pi_c (x) & =  D_c  (P(x)) \cdot\left( P(x) - \frac{(x+c_c)}{\alpha}\right)  \\ \notag
	& =   \left( d'_0-d'_1p_0 \right)\left(p_0-\frac{c_c}{\alpha}\right)+ \left( \left( d'_0-d'_1p_0 \right)\left(p_1-\frac{1}{\alpha}\right) -d'_1p_1\left(p_0-\frac{c_c}{\alpha}\right)\right) x +d'_1p_1\left(\frac{1}{\alpha}-p_1\right)x^2 \\
	& =: \gamma_0 + \gamma_1 x + \gamma_2 x^2. \label{cons-profit}
	\end{align}
	
The consumer profit is concave in the commodity price $x$, $\gamma_2 < 0$, if and only if  the pass--through coefficient $p_1$ is higher than the conversion factor $1/\alpha$. It means that the final good price increases faster than the need of the downstream industry to produce one more good, which is a sound economic condition for having a sustainable downstream industry. To sum up, the profit rates of both producer and consumer are concave and quadratic in $x$.

\paragraph{Market conditions:} We model the commodity price process $(X_t)$ as a controlled It\^o diffusion of the form
\begin{align}\label{dyn-X} dX_t = \mu_t dt + \sigma dW_t - dN_t.\end{align}
The Brownian motion $(W_t)$ captures exogenous price shocks due to random demand or production fluctuations, or pertinent economic shocks for the industry. In this sense, the model is agnostic in the reasons why the commodity price fluctuates around its mean trend. The point process $N_t := \sum_{i\ge 1} \xi_i \mathbf 1_{\{ \tau_i \le t \} }$ captures the producer interventions at times $(\tau_i)_{i \ge 1}$ and impulses $(\xi_i)_{i \ge 1}$. A positive impulse is triggered by an investment phase, and has a negative impact on the price. A negative impulse is induced by a disinvestment phase and has a positive impact on the price.

The drift process $(\mu_t)$ represents the state of the retail market for the final good. It is either in expansion or in contraction state. When in expansion, demand is growing faster than the available production capacity, hence prices tend to rise: $\mu_t = \mu_+ > 0$. When in contraction, the demand is shrinking faster than the production capacity, thus the price tends to decrease, and thus $\mu_t = \mu_- < 0$. This modeling corresponds to an imperfect adjustment of the market as in a {\em sticky price} model in macroeconomics. The drift is fully controlled by the consumer, $$\mu_t = \mu_+ \sum_{i=0}^\infty \mathbf 1_{\{ \sigma_{2i} \le t < \sigma_{2i+1} \} } + \mu_-  \sum_{i=1}^\infty \mathbf 1_{\{ \sigma_{2i-1} \le t < \sigma_{2i} \} },\qquad t \ge 0,$$ where $\sigma_i$ is the $i$-th switching instance taken by the consumer in the case $\mu_{0-}=\mu_+$ (with the convention $\sigma_0 =0$, so that $\sigma_1$ is the first switching time) and analogously when $\mu_{0-}=\mu_-$ by interchanging odd and even switching times. Thus, both players influence $(X_t)$, although their actions are of distinct types, namely impulse control $(N_t)$ by the producer and switching-drift control $(\mu_t)$ by the consumer. The resulting controlled price dynamics are denoted as $X^{(\mu,N)}$.

The quadratic nature of the upstream and downstream profit rate functions  $\pi_p(\cdot)$ and $\pi_c(\cdot)$ implies that each player has their own natural habitat given by the intervals $(x_c^1, x_c^2)$ and $(x_p^1,x_p^2)$ for commodity price levels with:
\begin{align}
  x^1_p & := \min\Big\{c_p,\frac{d_0}{d_1}\Big\} ,   \qquad && x^2_p :=  \max\Big\{c_p,\frac{d_0}{d_1}\Big\} , \\
  x^1_c & := \min\Big\{ \frac{p_0-c_c/\alpha}{1/\alpha-p_1},\frac{d'_0-d'_1p_0}{p_1d'_1}\Big\} ,  \qquad && x^2_c := \max\Big\{\frac{p_0-c_c/\alpha}{1/\alpha-p_1},\frac{d'_0-d'_1p_0}{p_1d'_1}\Big\}.
\end{align}

Players make a positive profit only if the price stays in the interval $(x^1_i,x^2_i)$, $i \in \{c,p\}$. The concavity of the profit functions implies that players have \emph{preferred} commodity levels $\bar{X}_p, \bar{X}_c$ that maximize their profit rates, namely:
\begin{align}
  \bar{X}_p := \frac{d_0 + c_p d_1}{2 d_1}, \qquad \bar{X}_c := -\frac{\gamma_1}{2 \gamma_2}.
\end{align}
Typically, we expect that $\bar{X}_c < \bar{X}_p$, so that the preferred commodity price of the consumer is lower than that of the producer. The stochastic fluctuations coming from $(W_t)$ can generate three different market conditions:
\begin{align*}
  X_t < \bar{X}_c &  \quad\text{I: abnormally low prices}; \\
  \bar{X}_c \le X_t \le \bar{X}_p &   \quad\text{II: vertical competition}; \\
\bar{X}_p  < X_t   &  \quad\text{III: abnormally high prices}.
\end{align*}

In the first and last cases, both players have the same preferences to raise or decrease $X_t$; in the intermediate case, they compete against each other. Because both players can in principle push $(X_t)$ in either direction, the market organization is influenced by their relative gain of doing so, as well as their action costs. In cases I and III, the players are in waiting mode because of the second--mover advantage, hoping that the other will act first which allows the second to benefit from the price effect without paying the cost of (dis-)investment or of switching. In case II they are in preemption mode, with the player who moves first being able to increase her profits at the expense of the other. These dynamic shifts between waiting and preemption is an important feature of vertical competition.

\paragraph{Objective functions and admissible strategies:} The objective functionals of the players consist of integrated profit rates $\pi_{\cdot}(x)$, discounted at constant rate $\beta >0$ and subtracting the control costs that are paid at respective intervention epochs. We take the investment cost function of the producer to be some convex function $K_p : \mathbb R \to \mathbb R$, and of the consumer as $H: \{\mu_-, \mu_+\} \to \mathbb R_+$. We denote the latter as $H(\mu_-) = h_-, H(\mu_+) = h_+$. Depending on the initial drift $\mu_0$ being positive/negative the producer's objective function is given by:
\begin{align}
J_p ^\pm (x; N,\mu) :=  \E\Big[ \int_0^\infty e^{- \beta \, t} \big( X_t - c_p) D_p(X_t) dt - \sum_i e^{-\beta \, \tau_i} K_p (\xi_{i}) \Big| \, \mu_0 = \mu^\pm, X_0 = x\Big],\end{align}
and, similarly, the representative consumer's objective function is:
\begin{align}
J_c ^\pm (x; N,\mu) := \E\Big[ \int_0^\infty e^{- \beta \, t} \big( \gamma_0 + \gamma_1 X_t + \gamma_2 X_t ^2)  dt - \sum_j e^{-\beta \, \sigma_j} H(\mu_{\sigma_j}) \Big| \, \mu_0 = \mu^\pm, X_0 = x \Big].\end{align}

In order for the state variable dynamics and players' expected payoffs to be well-defined we give the following definition of admissible strategies. To this end, let $(\Omega, \mathcal F, (\mathcal F_t)_{t \ge 0}, \mathbb P)$ be a probability space with a filtration satisfying the usual conditions and supporting an $(\mathcal F_t)_{t \ge 0}$-Brownian motion $(W_t)$. 

\begin{definition}[Admissible strategies] \label{def:adm}
We say that $(\tau_i, \xi_i)_{i \ge 1}$ is an admissible strategy for the producer if the following properties hold:\begin{enumerate}
\item $(\tau_i)_{i \ge 1}$ is a sequence of $[0,\infty]$-valued stopping times such that $0 \le \tau_1< \tau_2 < \cdots $ and $\lim_{i \to \infty} \tau_i =\infty$ a.s., with the convention that $\tau_i = \infty$ for some $i \ge 1$ implies $\tau_k = \infty$ for all $k \ge i$;
\item $(\xi_i)_{i \ge 1}$ is a sequence of real-valued $\mathcal F_{\tau_i}$-measurable random variables;
\item the sequence $(\tau_i , \xi_i)_{i \ge 1}$ satisfies $\sum_{i \ge 1} e^{-\beta \tau_i } \xi_i \in L^2 (\mathbb P)$.
\end{enumerate}
Similarly, we say that the sequence $(\sigma_j)_{j \ge 1} $ is an admissible strategy for the consumer if \begin{enumerate}
\item[4.] each $\sigma_j$ is a $[0,\infty]$-valued stopping time, $0 \le \sigma_1 < \sigma_2 < \cdots $, with the convention that $\sigma_j = \infty$ for some $j \ge 1$ implies $\sigma_k = \infty$ for all $k \ge j$;
\item[5.] $\sum_{j \ge 1} e^{-\beta \sigma_j} \in L^2 (\mathbb P)$.
\end{enumerate}
The set of all producer's (resp.~consumer's) admissible strategies is denoted by $\mathcal A_p$ (resp.~$\mathcal A_c$). 
\end{definition}

\begin{remark}
{\rm Observe that the property 1 above implies that the producer intervention times do not accumulate in finite time, so that for all $t >0$ the process $N_t = \sum_{i\ge 1} \xi_i \mathbf 1_{\{ \tau_i \le t \} }$, $t \ge 0$, is well-defined, adapted and finite-valued. Moreover, the integrability condition in 5 gives that $\sigma_j \to \infty$ (as $j \to \infty$), i.e.~the switching times of the consumer do not accumulate in finite time either, so that the dynamics of the controlled state variable \eqref{dyn-X} is well-defined too. Regarding the expected profits of the players, they are both finite due to integrability properties in 3 and 5 above.}
\end{remark}

\begin{remark}\label{rmk:priority}
{\rm According to the definition of admissibility above, neither player can intervene more than once at a time. However, simultaneous interventions coming from both of them are not excluded. As discussed, the dynamics of the intervention in upstream production is much faster than the switching of the consumption regime for final good. Thus, in case both players try to act simultaneously,  we assume that the producer has priority. 
This avoids  unnecessary technicalities and allows for a consistent modeling of the vertical competition.
}
\end{remark}

\subsection{Equilibrium}
Using this notion of admissible strategies, we give the definition of Nash equilibrium.

\begin{definition}[Nash equilibrium]
 A Nash equilibrium is any pair $((\xi_i, \tau_i)_{i \ge 1}, (\sigma_j)_{j \ge 1}) \in \mathcal A_p \times \mathcal A_c$ satisfying the following property:
\[ J_p ^\pm (x; N',\mu) \le J_p ^\pm (x; N,\mu), \qquad J_c ^\pm (x; N,\mu ') \le J_c ^\pm (x; N,\mu), \qquad \forall x \in \mathbb{R},\]
for any other pair of strategies $((\xi_i ', \tau'_i )_{i \ge 1},(\sigma'_j)_{j \ge 1}) \in \mathcal A_p \times \mathcal A_c$, where in the payoffs $J^- _r (x; \cdot)$, $r \in \{c,p\}$, above we have $N'_t = \sum_{i \ge 1} \xi'_i \mathbf 1_{\{ \tau'_i \le t \} }$ and $\mu'_t = \mu_+ \sum_{i=0}^\infty 1_{\{ \sigma'_{2i} \le t < \sigma'_{2i+1} \} } + \mu_-  \sum_{i=1}^\infty \mathbf 1_{\{ \sigma'_{2i-1} \le t < \sigma'_{2i} \} }$ for $t\ge 0$, $\mu' _{0-}=\mu_+$ and the convention $\sigma^\prime _0 =0$ (analogously in the other case $\mu'_{0-}=\mu_-$ by interchanging odd and even switching times).
\end{definition}

In line with the envisioned Markovian structure and in order to maximize tractability, we concentrate on a specific class of dynamic equilibria. Namely, we aim to
construct threshold-type Feedback Nash Equilibria which are of the form
\begin{align}\label{eq:imp-strategies}
\tau_0 = 0, \quad \tau_i = \inf \{ t > \tau_{i-1} : X_t \in \Gamma_p (t-) \}, \quad i \ge 1, \qquad \xi_i = \delta( X_{\tau_i}, \mu_{\tau_i-}),\end{align}
and \begin{align}\label{eq:switching-times}
\sigma_0 =0, \quad \sigma_j = \inf \{ t > \sigma_{j-1} : X_t \in \Gamma_c (t) \}, \quad j \ge 1,\end{align}
where \[ \Gamma_r (t) = \Gamma^+ _r \mathbf 1_{\{\mu_t = \mu_+\}} + \Gamma^- _r  \mathbf 1_{\{\mu_t = \mu_-\}}, \quad r\in \{c,p\}, \]
for some measurable function $\delta : \mathbb R \to \mathbb R$  and some suitable Borel sets $\Gamma_p ^\pm ,\Gamma_c ^\pm \subset \mathbb R$. Thus, \eqref{eq:imp-strategies}-\eqref{eq:switching-times} imply that players act based solely on the current price $(X_t)$ and demand regime $(\mu_t)$, ruling out history-dependent strategies, and moreover the strategies are characterized through fixed action regions $\Gamma_p^\pm, \Gamma_c^\pm$ and impulse maps $\delta(\cdot)$. We will denote by $\tau'_k$ the aggregated intervention times coming jointly from the two players. The fact that in \eqref{eq:imp-strategies} producer's intervention times $\tau_i$ are defined via $\Gamma_p (t-)$ translates the assumption that in case of simultaneous interventions, the producer plays first and so her thresholds naturally depend on the drift $\mu_{t-}$ just before her and consumer's actions (compare to Remark~\ref{rmk:priority}).

The action regions $\Gamma_p ^\pm ,\Gamma_c ^\pm$  are expected to be as follows. The impulse intervention region of the upstream production $\Gamma^\pm _p =(x^\pm _\ell, x^\pm _h)$ is two-sided: the producer will act whenever $X_t$ reaches $x^\pm _h$ from below or drops to $x^\pm _\ell$ from above. Note that these thresholds $x_\ell^\pm, x_h^\pm$ are $\mu$-dependent. On the consumption side, when $\mu_t = \mu_+$ (expansion regime),  the consumer will switch to $\mu_-$ if $X_t$ gets too high: $\Gamma^+_c = (y_h, \infty)$. Similarly when $\mu_t = \mu_-$ (contraction regime), she will switch to $\mu_+$ if $X_t$ gets too low $\Gamma^-_c = (-\infty, y_\ell)$. Finally, when the producer intervenes, he will bring $X_t$ to her impulse level $x^{\pm\ast}_{r}$ so that the impulse amount is $\xi^{\pm}_{r} = x^{\pm} _{r} -x^{\pm\ast}_{r}$.
The natural ordering we expect is the producer impulses towards $\bar{X}_p$
\begin{align}
x^{\pm}_\ell < x^{\pm\ast}_\ell  \quad \textrm{and} \quad  x^{\pm\ast}_h < x^{\pm} _h ,
\end{align}  and the consumer switches towards $\bar{X}_c$,
\begin{align}
y_\ell < \bar{X}_c < y_h ,\end{align}
so that when acting both players try to move $X$ towards their preferred levels. However, the precise ordering between  the impulse thresholds $x^{\pm}_r$ and the switching thresholds $y$'s is not clear a priori and will emerge as part of the overall equilibrium construction.

\subsection{Illustration of Competitive Dynamics}

To further understand the market evolution under competition of the producer and consumer, we focus on the case where both players are active. The producer's strategy is summarized via a $2\times 4$ matrix  $\mathcal{C}_p$ which lists the thresholds $\x_\ell,\x_h$ and the target levels $\xs_\ell,\xs_h$. Thus, the no-intervention regions are $[\x_\ell, \x_h]$ and impulse amounts are $\x_h-\xs_h, \xs_\ell-\x_\ell$:
\begin{align}
\mathcal{C}_p=\begin{bmatrix}
x^+_\ell, &  x^{+\ast}_\ell, & x^{+\ast}_h, & x^{+}_h\\
x^-_\ell, &  x^{-\ast}_\ell, & x^{-\ast}_h, & x^{-}_h\\
\end{bmatrix}.
\end{align}
The consumer has two switching thresholds $y_\ell, y_h$; in a typical setup we expect them to satisfy the following ordering
\begin{align}\label{eq:OrderInXY}
x^-_\ell<y_\ell<y_h<x_h^+.
\end{align}

Note that in the expansion regime (drift $\mu_+$), we assume that $y_h < x^+_h$. Therefore, coming from below, 
$X^\ast_t$ hits $y_h$ first, causing the consumer to switch into the contraction regime with drift $\mu_-$. As a result, the impulse threshold $x^+_h$ is not effective, i.e.~it will never get triggered along an equilibrium path of $(X_t)$. Similar argument implies that $x^-_\ell$ is \textit{not effective} either if $x^-_\ell<y_\ell$. In the left panel of Fig.~\ref{fig:illuBR} we illustrate such threshold-based vertical competition among the two players.
\begin{figure}[ht]
	\centering
	\includegraphics[height=2.2in,trim=0.4in 0.7in 0.4in 0.4in]{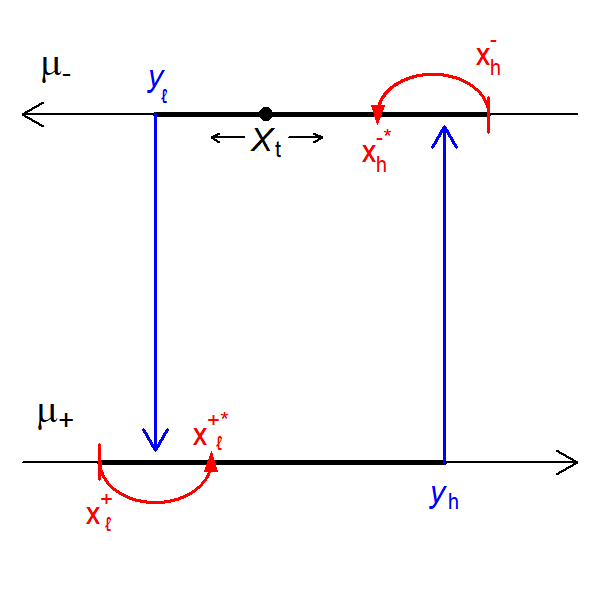}\hspace{10mm}
	\includegraphics[width=0.4\textwidth]{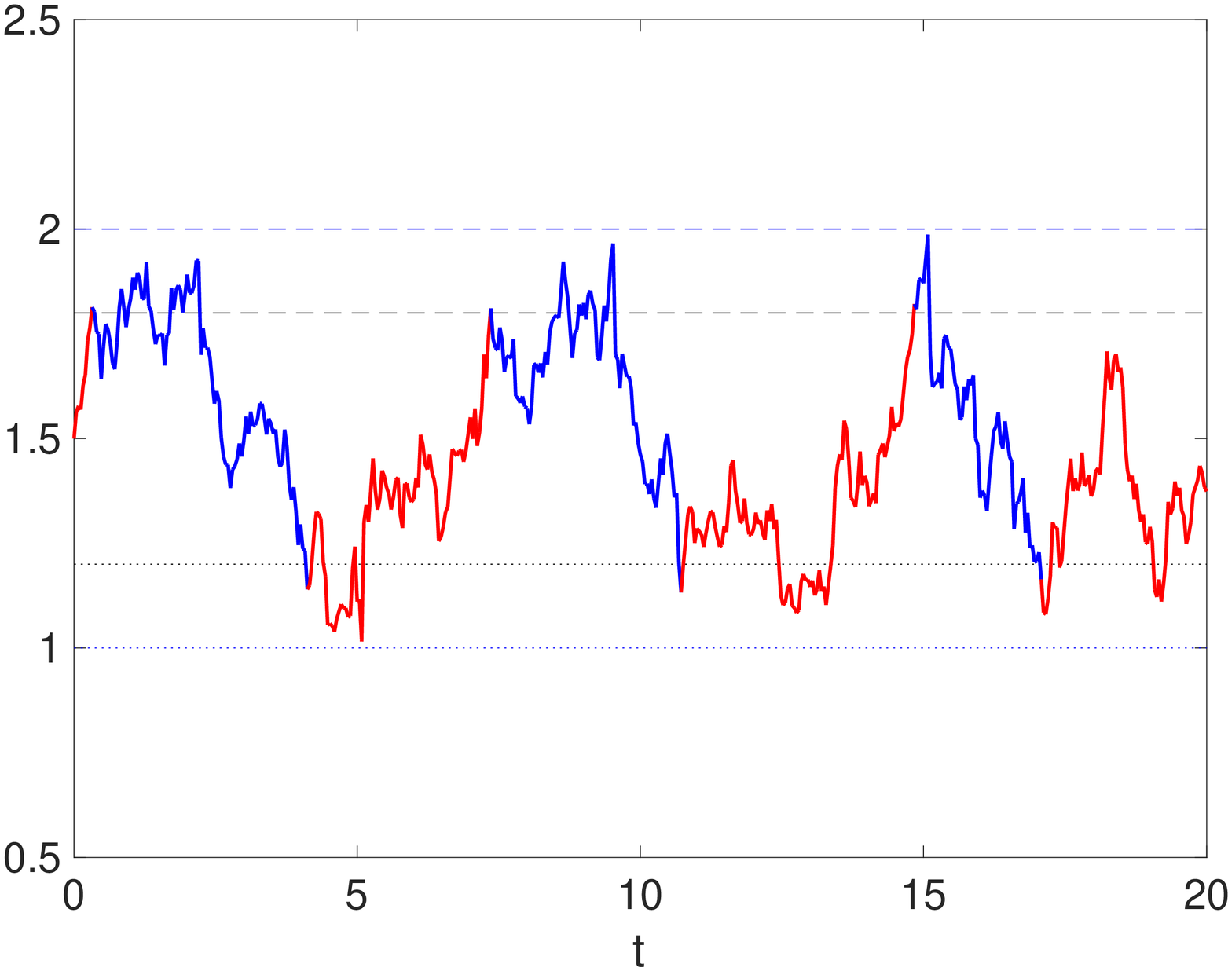}
	\caption{\emph{Left panel:} Dynamic competition between producer and consumer. The blue arrows represent drift-switching controls exercised by the consumer at levels $y_{\ell}$ and $y_h$, while the red curved arrows represent impulse controls exercised by the producer at levels $x_{\ell}^+, x^-_h$ that instantaneously push $X_t$ to $x_\ell^{+\ast}$ and $x^{-\ast}_h$ respectively. \emph{Right}: A sample path of the controlled commodity price $(X^\ast_t)$ under competitive equilibrium. Observe that $X^\ast_t \in [1,2]$ for all $t$.}\label{fig:illuBR}
\end{figure}

To illustrate competitive dynamics, the right panel of Fig.~\ref{fig:illuBR} shows a sample trajectory of $(X^\ast_t)$ (the superscript emphasizing the fact that we are now looking at equilibrium) with producer and consumer strategies
$$\mathcal{C}_p=\begin{bmatrix}
1.0, &1.3,& 1.7,& 2.0\\
1.0,& 1.3, & 1.7, & 2.0
\end{bmatrix}, \quad \quad (y_\ell, y_h)=(1.2, 1.8).
$$

According to the above discussion, the effective thresholds are $(x^+_\ell, y_h)$ when $\mu_t=\mu_+$, or $(y_\ell, x^-_h)$ when $\mu_t=\mu_-$. In other words, in the expansion regime, $(X_t)$ will be between $[1.0, 1.8]$  and in the contraction regime it will be between $[1.2, 2.0]$. In Fig.~\ref{fig:illuBR} (Right), we start in the contraction regime with  $X_0=1.5$ and $\mu_0=\mu_-$. On this trajectory, $(X^\ast_t)$ moves down until it touches the consumer's threshold $y_\ell$, where the consumer switches to a positive drift to draw the price up. Nevertheless, the price keeps decreasing and hits $x^+_\ell=1.0$, whereby the producer intervenes and pushes it to $x^{+\ast}_\ell=1.3$. Prices then continue to rise up to $y_h = 1.8$ at which point the consumer switches again and starts pushing them back down (supposedly she wishes to keep them somewhere around 1.5). This cyclic behavior continues ad infinitum, yielding a stationary distribution for the pair $(X^\ast_t, \mu^\ast_t)$. Note that the consumer uses her switching control to keep $X^\ast_t$ from going too high or too low, essentially cycling between $y_\ell$ and $y_h$. Indeed, starting at $X^\ast_t = y_\ell$, the consumer switches to expansion which causes prices to trend up; once they hit $y_h$ the consumer switches to contraction, causing prices to trend down. As a result, $\mu_t$ alternates between $\mu_+, \mu_-$ generating a mean-reverting behavior. Throughout, the producer acts as a ``back-up'', explicitly forcing prices from becoming extreme (namely from falling in the expansion regime, or rising in the contraction regime). These additional interventions by the producer make the domain of $(X^\ast_t)$ bounded.

It is also possible that, say, $x^+_h < y_h$ so that in the expansion regime the producer will act first both when $(X^\ast_t)$ falls (impulse threshold $x^+_\ell$) and when $(X^\ast_t)$ rises ($x^+_h$), making the consumer \emph{inactive}. In that case it is plain to see that the drift $\mu_t \equiv \mu_+$ will stay positive forever; $(X_t)$ will be forced to a bounded domain but will not have mean-reverting dynamics since the drift is constant. Instead, it will experience repeated impulses downward to counteract the upward trend due to ongoing consumption growth.

\section{Best--response functions} \label{sec:eqm}
To obtain a threshold-type Feedback Nash Equilibrium we view it as a fixed point of the producer and consumer \emph{best--response maps}. Therefore, our overall strategy is to (i) characterize threshold--type switching strategies for the consumer given a pre-specified, threshold--type behavior by the producer; (ii) characterize threshold--type impulse strategies for the producer who faces a pre-specified regime--switching behavior of $(X_t)$; (iii) employ t\^atonnement, i.e.~iteratively apply the best--response controls alternating between the two players to construct an \textit{interior, non-preemptive} equilibrium satisfying the ordering \eqref{eq:OrderInXY}.

To analyze best--response strategies, we utilize stochastic control theory, rephrasing the related dynamic optimization objectives through \emph{variational inequalities} (VI) for the jump--diffusion dynamics \eqref{dyn-X}. The competitor thresholds then act as boundary conditions in the VIs. To establish the desired equilibrium we need to verify that the best response is also of threshold-type and solves the expected systems of equations.
We note that all three pieces above are new and we have not been able to find precise analogues of the needed verification theorems in the extant literature. Nevertheless, they do build upon similar single--agent control formulations, so the overall technique is conceptually clear.

\subsection{Consumer Best--Response}
Fixing impulse thresholds $x^\pm_r$ ($r=h,l$), the consumer faces a two--state switching control problem on the bounded domain $(x^\pm_\ell, x^\pm_h)$. Namely, given a producer's impulse strategy $(\tau_i ,\xi_i)_{i \ge 1}$ with $\tau = \inf \{ t : X_t \notin [x^{\pm}_\ell, x^{\pm}_h ] \}$,
we expect the following stochastic representation for her value functions $w^\pm(x)$ with $x \in [x^{\pm}_\ell, x^{\pm}_h ]$
\begin{align}\label{eq:bellman}
w^\pm(x)  &= \sup_{ \sigma  \in\mathcal{T } }\E_{x,\pm}\Bigg[\int_0^{{\tau}\wedge \sigma}e^{-\beta t} \pi_c (X_t) dt  +e^{-\beta \underline{\tau}}\mathds{1}_{\{\tau < \sigma \}}
\Big(w^\pm(X_{\tau}-\xi)\Big) +e^{-\beta \underline{\tau}}\mathds{1}_{\{\tau > \sigma \}}\Big(w^\mp(X_{\sigma}) - h_\pm \Big)\Bigg],
\end{align}
where  $\E_{x,\pm}$ denotes expectation with respect to $\mu_t \in \{\mu_-, \mu_+\}$ and $h_{\pm}$ are the fixed intervention costs of the consumer.  The above is a system of two coupled equations, which locally resembles an optimal stopping problem with running payoff $\pi_c(\cdot)$, reward $w^\mp(\cdot)$ (last term), and stop--loss payoff (middle term) $w^\mp(\cdot)$ due to the producer impulse at $\tau$.
This is \emph{almost} the formulation as considered in \cite{ALL17} except with two modifications:
\begin{itemize}
  \item The domain is bounded on both sides (previously there was a one--sided stop--loss region).

  \item The boundary condition $w^+(x_\ell) = w^+(x^{+\ast}_\ell)$ is autonomous but nonlocal. 
  Therefore, the two stopping--type VIs for the consumer are coupled only through the free boundaries, not through the stop--loss thresholds as in \cite{ALL17}.
\end{itemize}

Now, given a producer strategy $\mathcal C_p$, if the consumer's response is such that $y_\ell < x^-_\ell$ and $x^+_h < y_h$, the consumer will be stuck forever in the initial regime because the price touches $ x^-_\ell$ before $y_\ell $ in the contraction regime and $x^+_h$ before $y_h$ in the expansion regime. In this case, the price will oscillate between $x^-_\ell$ and $x^-_h$ if the initial market is in the contraction regime, and between $x^+_\ell$ and $x^+_h$  in the expansion regime.

In the case where the consumer's response satisfies $y_\ell < x^-_\ell$ and $y_h < x^+_h$, depending on the initial state, the consumer will switch once to the expansion regime or will be stuck in the initial expansion regime. If the initial regime is $\mu_+$, the price will touch $y_h$, the regime will switch to contraction, the price will never touch $y_\ell$ and will oscillate between $x^- _\ell$ and $x^-_h$.  If the initial state is already $\mu_+$, no switch of regime will ever occur. The same reasoning applies for the symmetric case where $x^-_\ell < y_\ell$ and $x^+_h < y_h$.

Finally, if the consumer's response satisfies  $x^-_\ell < y_\ell$ and $y_h < x^+_h$, then whatever the initial regime, the state $(\mu_t)$ will switch many times between the two regimes.

The best--response of the consumer consists in picking the best response amongst the three possible ones above.
Thus, we distinguish three cases:
\begin{description}
  \item[(a) No-Switch:]  The consumer is completely inactive and simply collects her payoff based on the strategy $(\x_{\ell,h})$.

  \item[(b) Single-Switch:]  The consumer always prefers one regime to the other. Then she is inactive (like in case (a) above) in the preferred regime and faces an optimal stopping (since there is only a single switch to consider) problem in the other regime.

  \item[(c) Multiple-Switch:] The consumer switches back and forth between both regimes: the continuation region is $(y_\ell,y_h)$.
\end{description}

Proposition~\ref{prop:cons-inactive} provides the value function of the consumer in case (a). The system \eqref{eq:single-switch} characterizes the game payoff in case (b), and  Proposition~\ref{prop:cons-best-response} provides the value function of the consumer in case (c).

\subsubsection{No--switch} \label{ssec:cbr-zero}

Regardless of the consumer strategy, in the continuation region, a direct application of the Feynman--Kac formula  on \eqref{eq:bellman} shows that her value function solves  the following ordinary differential equation (ODE)
\begin{equation}\label{inaction-PDE} - \beta w + \mu_\pm w_x  + \frac{1}{2} \sigma^2 w_{xx}  + \pi_c(x) = 0.\end{equation}
Solving this inhomogeneous second-order ODE, we obtain
$w^{\pm}(x)=\Dd^{\pm}(x)+u^{\pm}(x)$, where letting $\theta_2 ^\pm <0< \theta_1^\pm$ be the two real roots of the quadratic equation $-\beta + \mu_\pm z + \frac{1}{2}\sigma^2 z^2 =0$, \begin{itemize}
\item $u^{\pm}(x)=\lambda^{\pm}_{1}e^{\theta^{\pm}_1 x}+\lambda^{\pm}_{2}e^{\theta^{\pm}_2 x}$ solves the homogeneous ODE $-\beta u + \mu_\pm u_x + \frac{1}{2}\sigma^2 u_{xx} =0$ and $\lambda^\pm_{i,0}$, $i=1,2$ are to be determined from appropriate boundary conditions;

\item $\Dd^{\pm}(x)$ is a particular solution to \eqref{inaction-PDE}, given by $$\Dd^\pm (x) = E x^2 + F_\pm x + G_\pm \qquad \text{ where }$$
\begin{align}
E & = \frac{\gamma_2}{\beta}, \quad F_\pm = \frac{1}{\beta}\Big( \gamma_1 +  2 \mu_\pm \frac{\gamma_2}{\beta}\Big), \quad G_\pm = \frac{1}{\beta} \big(\gamma_0 + \sigma^2 \frac{\gamma_2}{\beta} + \mu_\pm F_\pm \big).
\end{align}
\end{itemize}

When the consumer is inactive (denoted by $w^\pm_0$), the continuation region is $[x_\ell^\pm, x_h^\pm]$ with the boundary conditions at the impulse levels
\begin{align}\label{eq:bd_noswitch}
w^{\pm}_0(\x_r)=w^{\pm}_0(\xs_r), \qquad r\in\{\ell, h\}.
\end{align}
From \eqref{eq:bd_noswitch} the respective coefficients $\lambda^{\pm}_{1,0}, \lambda^{\pm}_{2,0}$ are solved from the  following uncoupled linear system:
\begin{align}
\lambda^{\pm}_{1,0}\cdot\big[e^{\theta^{\pm}_1 \x_\ell}-e^{\theta^{\pm}_1 \xs_\ell}\big]+\lambda^{\pm}_{2,0}\cdot\big[e^{\theta^{\pm}_2 \x_\ell}-e^{\theta^{\pm}_2 \xs_\ell}\big]=\Dd^\pm (\xs_\ell)-\Dd^\pm (\x_\ell), \label{inaction-eq1}\\
\lambda^{\pm}_{1,0}\cdot\big[e^{\theta^{\pm}_1 \x_h}-e^{\theta^{\pm}_1 \xs_h}\big]+\lambda^{\pm}_{2,0}\cdot\big[e^{\theta^{\pm}_2 \x_h}-e^{\theta^{\pm}_2 \xs_h}\big]=\Dd^\pm (\xs_h)-\Dd^\pm (\x_h).\label{inaction-eq2}
\end{align}
For $x > x^{\pm}_h$ we take $w^\pm_0(x) = w^\pm_0(x^{\pm*}_h)$ and similarly in the contraction regime,  we take $w^\pm_0(x) = w^\pm_0(x^{\pm*}_\ell)$ for $x<x^{\pm}_\ell$.

\begin{prop}\label{prop:cons-inactive}
Let $(\lambda^\pm _{1,0}, \lambda^\pm _{2,0}) \in \mathbb R^4$ be the solution to the system \eqref{inaction-eq1}-\eqref{inaction-eq2}. Then the functions $w^\pm _0 (x)$, $x \in [x_\ell ^\pm, x_h ^\pm]$, are the value functions for an inactive consumer, i.e.~$w_0 ^\pm (x) = J_c ^\pm (x; N, \mu^\pm)$, where $N$ is the producer impulse strategy associated with the thresholds $(x_\ell ^\pm , x_\ell ^{\pm \ast}; x_h ^\pm , x_h ^{\pm \ast})$ with $x_\ell ^\pm < x_h ^\pm$.
\end{prop}

The role of $w^\pm_0(\cdot)$ is important for judging the other two cases, and moreover for deciding whether the  best--response ought to be of threshold--type.

\subsubsection{Single--switch}\label{ssec:cbr-one}

We next consider the situation where the payoff in the expansion regime is higher than the contraction one for any price $x$, so that the consumer is never incentivized to switch to the contraction regime. We then expect the consumer's corresponding best--response to be either a single--switch strategy (to the preferred regime) or no--switch (if already there).
Economically, this corresponds to $y_h > x_h^{+}$ so that as the price rises, the producer impulses $(X_t)$ down, and the consumer is not intervening to decrease her demand. As a result, the consumer never switches (except perhaps the first time from negative to positive drift) and $\lim_{t\to \infty} \mu_t = \mu_+$. This can be observed when demand switching is very expensive, so that the producer has full market power and is able to keep prices consistently low. The consumer is forced to be in the expansion regime forever and she is not able to influence $(X_t)$.

Suppose that the consumer prefers expansion regime ($\mu_t=\mu_+$) and adopts threshold-type strategies. Given $\mathcal{C}_p$, her strategy is summarized by
\begin{align*}
y_\ell > x^-_\ell, \qquad y_h=+\infty,
\end{align*}
and the resulting contraction--regime value function $w^-$ should be a solution to the variational inequality

\begin{align}\label{eq:qvi-cons-one-sided}
\sup\big\{-\beta w^- +\mu_- w^-_x+\frac{1}{2}\sigma^2 w^-_{xx}+ \pi_c ;\, w^+_0-h_--w^-\big\}=0,
\end{align}
where $w^+_0$ is from Proposition~\ref{prop:cons-inactive} and the continuation region is $[y_\ell, x_h^-]$. This is a standard optimal stopping problem. Note that while the above equation for $w^-$ depends on $w^+_0$, the equation for $w^+_0$ is autonomous---the system of equations becomes decoupled because the two regimes of $(\mu_t)$ no longer communicate.

To solve \eqref{eq:qvi-cons-one-sided} we posit that her best--response is of the form
	\begin{align}
	w^-(x)&=\begin{cases}
	w^+_0(x)-h_-, & x\leq y_\ell,\\
	\Dd^- (x)+\lambda^- _1e^{\theta^- _1x}+\lambda^- _2e^{\theta^- _2x}, &  y_\ell <  x < x^-_h,\\
	w^-(x^{-\ast}_h), & x^-_h \leq x,
	\end{cases}
	\end{align}

with the smooth pasting and boundary conditions:
\begin{align}\label{eq:single-switch}
\begin{cases}
\Dd^-(y_\ell)+\lambda^-_1e^{\theta^-_1y_\ell}+\lambda^-_2e^{\theta^-_2y_\ell}=\Dd^+(y_\ell)+\lambda^+_{1,0}e^{\theta^+_1y_\ell}+\lambda^+_{2,0}e^{\theta^+_2y_\ell}-h_-, & (\mathcal{C}^0\text{ at }y_\ell)\\
\Dd^-(x^-_h)+\lambda^-_1e^{\theta^-_1x^-_h}+\lambda^-_2e^{\theta^-_2x^-_h}=\Dd^-(x^{-\ast}_h)+\lambda^-_1e^{\theta^-_1x^{-\ast}_h}+\lambda^-_2e^{\theta^-_2x^{-\ast}_h}, & (\mathcal{C}^0\text{ at }x^-_h)\\
\Dd^-_x(y_\ell)+\lambda^-_1\theta^-_1e^{\theta^-_1y_\ell}+\lambda^-_2\theta^-_2e^{\theta^-_2y_\ell}=\Dd^+_x(y_\ell)+\lambda^+_{1,0}\theta^+_1e^{\theta^+_1y_\ell}+\lambda^+_{2,0}\theta^+_2e^{\theta^+_2y_\ell}. & (\mathcal{C}^1\text{ at }y_\ell)
\end{cases}
\end{align}

The system \eqref{eq:single-switch} is to be solved for the three unknowns $y_\ell, \lambda^-_{1}, \lambda^-_2$, while $\lambda^+_{1,0},\lambda^+_{2,0}$ are the coefficients of the consumer's payoff associated to the no-switch strategy in the $\mu_+$ regime, see previous subsection. We can re-write it as
first solving for $\lambda^{-}_{1,2}$ from the linear system
\begin{align}
\begin{bmatrix}
e^{\theta^-_1y_\ell} & e^{\theta^-_2y_\ell} \\
e^{\theta^-_1 x^-_h}-e^{\theta^-_1 x^{-\ast}_h} & e^{\theta^-_2 x^-_h}-e^{\theta^-_2 x^{-\ast}_h}
\end{bmatrix}
\cdot
\begin{bmatrix}
\lambda^-_1\\
\lambda^-_2
\end{bmatrix}
=\begin{bmatrix}
w^+_0(y_\ell)-\Dd^-(y_\ell)-h_-\\
\Dd^-(x^{-\ast}_h)-\Dd^-(x^{-}_h)
\end{bmatrix}
\end{align}
and then determining $y_\ell$ from the smooth pasting $\cC^1$-regularity \begin{align}
w^-_x(y_\ell)=w^+_{0,x}(y_\ell).
\end{align}
The case of a single--switch from expansion to contraction regime can be treated analogously in a symmetric way.

\subsubsection{Double--switch}
Finally, we consider the main case where the consumer adopts \textit{threshold--type} switches, i.e.~the ordering in \eqref{eq:OrderInXY}  holds. Given $\mathcal{C}_p$, the $w^\pm$ are then supposed to be a solution to the coupled variational inequalities

\begin{align} \label{eq:QVI-cons1}
\sup\big\{-\beta w^+ + \mu_+w^+_x+\frac{1}{2}\sigma^2 w^+_{xx}+\pi_c;\, \max\{w^--h_+, w^+ \}-w^+\big\}&=0,\\
\sup\big\{-\beta w^- + \mu_-w^-_x+\frac{1}{2}\sigma^2 w^-_{xx}+\pi_c;\, \max\{w^+-h_-, w^-\}-w^-\big\}&=0,\label{eq:QVI-cons2}
\end{align}
where we expect continuation regions of the form $(x_\ell^+, y_h)$ and $(y_\ell, x_h ^-)$.
To set up a verification argument for the consumer's best--response  we make the ansatz
\begin{subequations}\label{eq:Con_DS_payoff}
\begin{align}
w^+(x)&=\begin{cases}
w^+(x^{+\ast}_\ell), & x\leq x^+_\ell,\\
\Dd^+(x)+\lambda^+_1e^{\theta^+_1x}+\lambda^+_2e^{\theta^+_2x}, & x^+ _\ell < x < y_h,\\
w^-(x)-h_+, & x \geq y_h,
\end{cases}\\
w^-(x)&=\begin{cases}
w^+(x)-h_-, & x \leq y_\ell,\\
\Dd^-(x)+\lambda^-_1e^{\theta^-_1x}+\lambda^-_2e^{\theta^-_2x}, & y_\ell < x < x^-_h ,\\
w^-(x^{-\ast}_h), & x\geq x^-_h .
\end{cases}
\end{align}
\end{subequations}
This yields 6 equations:
\begin{align}\label{smooth-pasting}
\begin{cases}
\Dd^+(y_\ell)+\lambda^+_1e^{\theta^+_1y_\ell}+\lambda^+_2e^{\theta^+_2y_\ell}  -h_- =\Dd^-(y_\ell)+\lambda^-_1e^{\theta^-_1y_\ell}+\lambda^-_2e^{\theta^-_2y_\ell}, &(\mathcal{C}^0\text{ at } y_\ell)\\
\Dd^+(x^+_\ell)+\lambda^+_1e^{\theta^+_1x^+_\ell}+\lambda^+_2e^{\theta^+_2x^+_\ell}=\Dd^+(x^{+\ast}_\ell)+\lambda^+_1e^{\theta^+_1x^{+\ast}_\ell}+\lambda^+_2e^{\theta^+_2x^{+\ast}_\ell}, & (\mathcal{C}^0\text{ at }x^+_\ell)\\
\Dd^-(y_h)+\lambda^-_1e^{\theta^-_1y_h}+\lambda^-_2e^{\theta^-_2y_h}  -h_+=\Dd^+(y_h)+\lambda^+_1e^{\theta^+_1y_h}+\lambda^+_2e^{\theta^+_2y_h} , &(\mathcal{C}^0\text{ at } y_h)\\
\Dd^-(x^{-}_h)+\lambda^-_1e^{\theta^-_1x^{-}_h}+\lambda^-_2e^{\theta^-_2x^{-}_h}=\Dd^-(x^{-\ast}_h)+\lambda^-_1e^{\theta^-_1x^{-\ast}_h}+\lambda^-_2e^{\theta^-_2x^{-\ast}_h}, & (\mathcal{C}^0\text{ at }x^-_h)\\
\Dd^+_x(y_\ell)+\lambda^+_1\theta^+_1e^{\theta^+_1y_\ell}+\lambda^+_2\theta^+_2e^{\theta^+_2y_\ell}=\Dd^-_x(y_\ell)+\lambda^-_1\theta^-_1e^{\theta^-_1y_\ell}+\lambda^-_2\theta^-_2e^{\theta^-_2y_\ell}, &(\mathcal{C}^1\text{ at } y_\ell)\\
\Dd^-_x(y_h)+\lambda^-_1\theta^-_1e^{\theta^-_1y_h}+\lambda^-_2\theta^-_2e^{\theta^-_2y_h}=\Dd^+_x(y_h)+\lambda^+_1\theta^+_1e^{\theta^+_1y_h}+\lambda^+_2\theta^+_2e^{\theta^+_2y_h}. &(\mathcal{C}^1\text{ at } y_h)\\
\end{cases}
\end{align}
The six equations can be split into a linear system for the four coefficients $\lambda^\pm_{1,2}$'s
\begin{align}\label{eq:linearSysBR_consumer}
\begin{bmatrix}
e^{\theta^+_1y_\ell} & e^{\theta^+_2y_\ell} &  -e^{\theta^-_1y_\ell} & -e^{\theta^-_2y_\ell} \\
e^{\theta^+_1 x^+_\ell}-e^{\theta^+_1 x^{+\ast}_\ell} & e^{\theta^+_2 x^+_\ell}-e^{\theta^+_2 x^{+\ast}_\ell} & 0 & 0 \\
-e^{\theta^+_1y_h} & -e^{\theta^+_2y_h} & e^{\theta^-_1y_h} & e^{\theta^-_2y_h} \\
0 & 0 & e^{\theta^-_1 x^-_h}-e^{\theta^-_1 x^{-\ast}_h} & e^{\theta^-_2 x^-_h}-e^{\theta^-_2 x^{-\ast}_h}
\end{bmatrix}
\cdot
\begin{bmatrix}
\lambda^+_1\\
\lambda^+_2\\
\lambda^-_1\\
\lambda^-_2
\end{bmatrix}
=\begin{bmatrix}
\Dd^-(y_\ell)-\Dd^+(y_\ell)-h_+\\
\Dd^+(x^{+\ast}_\ell)-\Dd^+(x^{+}_\ell)\\
\Dd^+(y_h)-\Dd^-(y_h)-h_-\\
\Dd^-(x^{-\ast}_h)-\Dd^-(x^{-}_h)
\end{bmatrix}
\end{align}
and the smooth-pasting conditions determining the two switching thresholds $y_{\ell,h}$ (viewed as free boundaries)
\begin{align}\label{eq:smooth-consumer}
w^+_x(y_r)=w^-_x(y_r), \qquad r\in\{\ell, h\}.
\end{align}
\begin{prop}\label{prop:cons-best-response} Let the $6$-tuple $(\lambda^\pm _1, \lambda^\pm _2, y_h, y_\ell)$ be a solution to the system \eqref{eq:linearSysBR_consumer}-\eqref{eq:smooth-consumer} such that the order in \eqref{eq:OrderInXY} is fulfilled. Then, the functions defined in \eqref{eq:Con_DS_payoff} give the best--response payoffs of consumer, and a best--response strategy is given by $(\hat \sigma_i)_{i \ge 1}$, where 
\[ \hat \sigma_{0} =0, \quad \hat \sigma_i = \inf \left\{ t > \hat \sigma_{i-1} : X_t \in \Gamma_c (t) \right\}, \quad i \ge 1,\]
with $\Gamma_c ^+ = [y_\ell ,+\infty)$ and $\Gamma_c ^- = (-\infty, y_h]$.
\end{prop}

Figure~\ref{fig:consBR012A} illustrates the shapes of the consumer's value function in the different case of best--response. For the strategy given, we have a dominant function in the contraction regime ($w^-_0$) and a dominant function in the expansion regime ($w_0^+$).

\begin{figure}[ht]
	\centering
	\subfigure[][ $\mathcal{C}_p= \begin{bmatrix} 1.2, & 2.5, & 2.6, & 4.2  \\1.5, & 3.0,& 2.8, & 4.5  \end{bmatrix}$]{
  \includegraphics[width=0.35\textwidth]{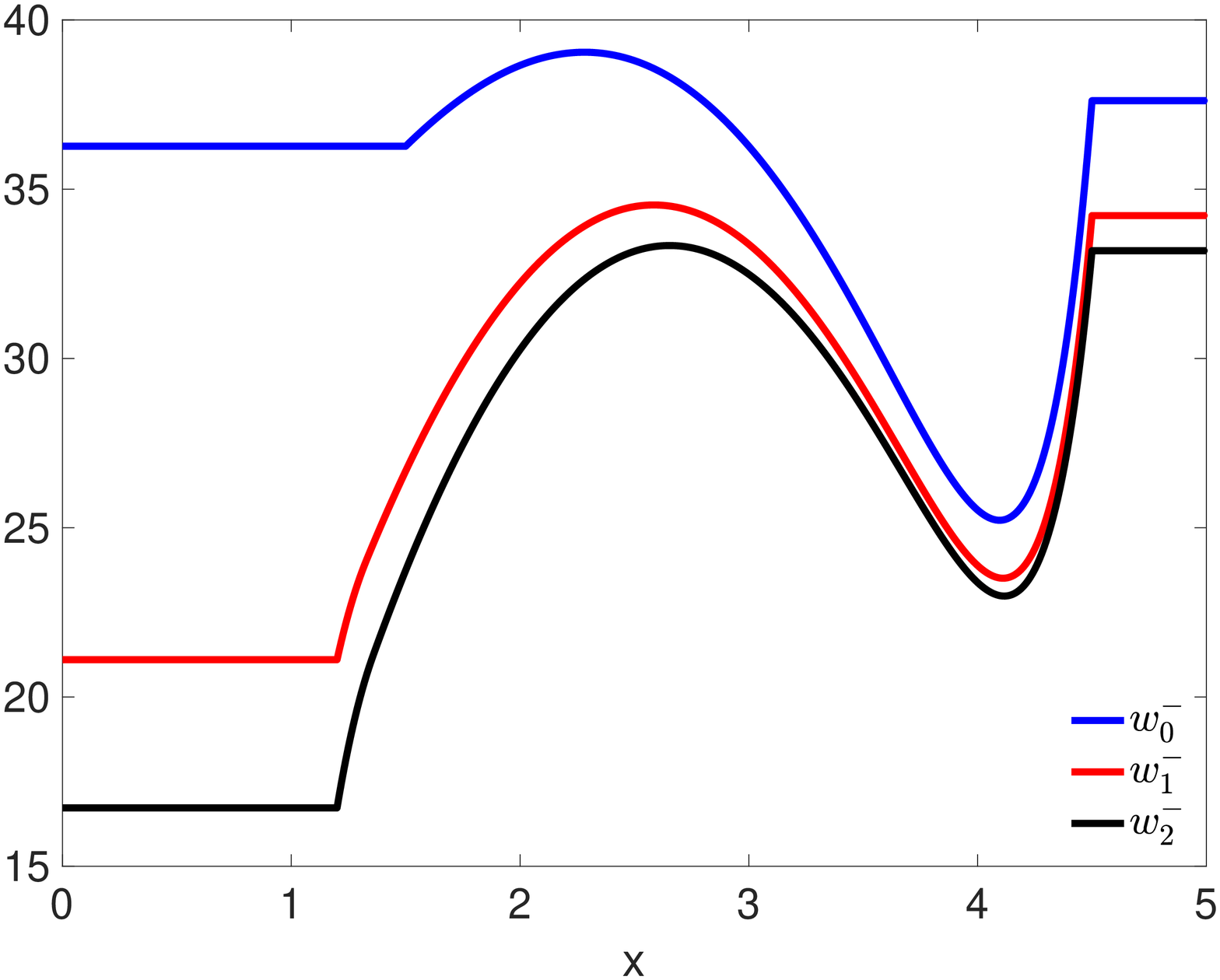}  \includegraphics[width=0.35\textwidth]{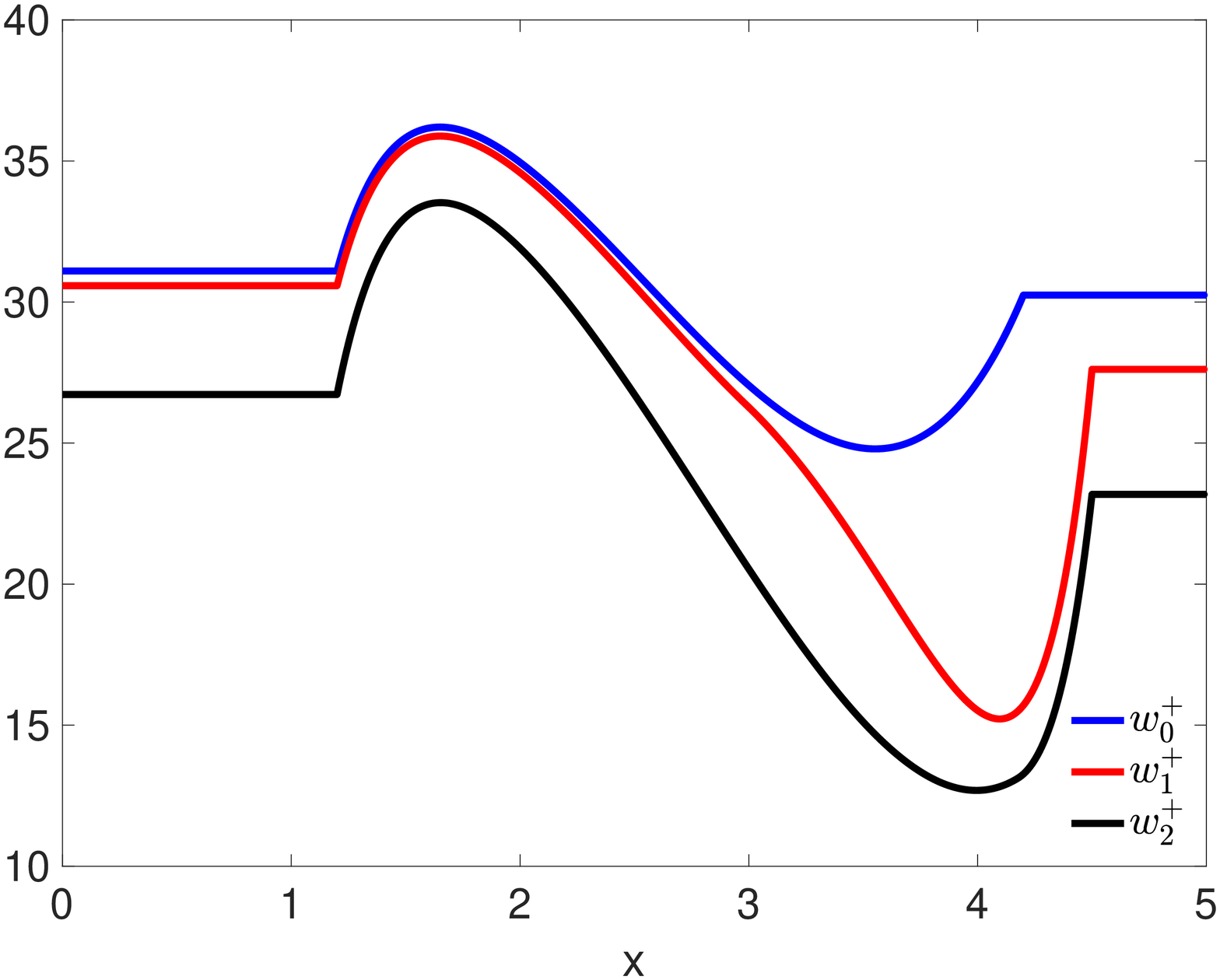}}
	\caption{Value functions $w^\pm_0$, $w^\pm_1$ and $w^\pm_2$ of the consumer given the producer's strategy (a). }
	\label{fig:consBR012A}
\end{figure}

\noindent {\bf Remark}: For comparison purposes, it is also useful to know the continuation region of the consumer when she \emph{alone} controls the market price $(X_t)$. As usual, this region is $(-\infty,y_h)$ in the expansion regime and $(y_\ell, +\infty)$ in the contraction regime, with the natural ordering $y_\ell < y_h$. The value functions $w^\pm$ satisfy:
\begin{align} \label{eq:QVI-cons-alone}
\sup\big\{ -\beta w^+ + \mu_+ w^+_x +\frac{1}{2}\sigma^2 w^+_{xx} + \pi_c;\, w^- - h_+ \big\}&=0,\\
\sup\big\{ -\beta w^- + \mu_- w^-_x    + \frac{1}{2}\sigma^2 w^-_{xx} + \pi_c; \,  w^+ - h_-  \big\}&=0. \label{eq:QVI-cons-alone-2}
\end{align}

To set up a verification argument for the consumer's best--response  we make the ansatz
\begin{subequations}\label{eq:w-cons-alone}
\begin{align}
w^+(x) & = \begin{cases}	
		w^-(y_h) - h_+, 	& x \geq y_h, \\
		 \Dd^+(x) + \lambda_{1,0}^+ e^{\theta_1^+ x} + \lambda_{2,0}^+ e^{\theta_2^+ x}, 	& x < y_h,
		\end{cases}	\\
w^-(x) & = \begin{cases}
		 \Dd^-(x) + \lambda_{1,0}^- e^{\theta_1^- x} + \lambda_{2,0}^- e^{\theta_2^- x},	 		& x > y_\ell,\\
		w^+(y_\ell) -h_-,		 	& x \leq  y_\ell .
		\end{cases}
\end{align}
\end{subequations}
Furthermore, in the expansion regime, to keep $w^+(x)$ bounded as $x \to -\infty$ we must have $\lambda_{2,0}^+ =0$ because $\theta_2^+<0$. In the contraction regime, a similar argument gives $\lambda_{1,0}^- = 0$. 
We are left with the four unknowns $y_\ell, y_h$ and $\lambda_{1,0}^p$ and $\lambda_{2,0}^-$ determined from the following smooth pasting conditions:
\begin{align} \label{eq:cons-alone-paste}
\begin{cases}
\Dd^-(y_\ell) + \lambda^-_{2,0} e^{\theta^-_2 y_\ell}  =
 \Dd^+(y_\ell)   + \lambda_{1,0}^+ e^{\theta_1^+ y_\ell} - h_-,
&(\mathcal{C}^0\text{ at } y_\ell)\\
 \Dd^+(y_h)  + \lambda^+_{1,0} e^{\theta^+_1 y_h}
=
 \Dd^-(y_h) + \lambda_{2,0}^- e^{\theta_2^- y_h}   - h_+,
&(\mathcal{C}^0\text{ at } y_h)	\\
\Dd^-_{x}(y_\ell) + \lambda^-_{2,0} \theta^-_2 e^{\theta^-_2 y_\ell}
=
\Dd^+_{x}(y_\ell)   + \lambda^+_{1,0} \theta^+_1 e^{\theta^+_1 y_\ell},
&(\mathcal{C}^1\text{ at } y_\ell)\\
\Dd^+_{x}(y_h)   + \lambda^+_{1,0} \theta^+_1 e^{\theta^+_1 y_h}
=
\Dd^-_{x}(y_h) + \lambda^-_{2,0} \theta^-_2 e^{\theta^-_2 y_h}.
&(\mathcal{C}^1\text{ at } y_h)\\
\end{cases}
\end{align}$\hfill\Box$

\subsection{Producer Best Response}
We now consider the best--response of the producer, given the consumer's switching strategy denoted by $\mathcal{C}_c:=[y_\ell, y_h]$. Once again, we face three Cases:
\begin{enumerate}
  \item The producer is a monopolist, i.e.~the consumer is completely inactive;
  \item The consumer adopts a single--switch strategy;
  \item The consumer adopts a double--switch strategy.
\end{enumerate}

\subsubsection{Producer as Sole Optimizer}
To begin with, we determine the monopoly-like strategy of the producer assuming the consumer adopts a no--switch strategy. In that case $\mu_t$ is constant throughout and the functions $v^{\pm}$ of the producer satisfy the variational inequality (VI):
\begin{align}
\sup \Big\{ - \beta v^\pm + \mu_{\pm} v^\pm_x + \frac{1}{2} \sigma^2 v^\pm_{xx}  + \pi_p  \; ,
                                            \sup_{\xi} \big\{ v^\pm(\cdot + \xi) - v^\pm(\cdot) - K_p (\xi) \big\} \Big\} & = 0.
\end{align}
Note that the two VIs for $v^+$ and $v^-$ are autonomous, hence uncoupled from each other.
In the continuation region,  the general solution of the ODE
$$ - \beta v + \mu_{\pm} v_x + \frac{1}{2} \sigma^2 v_{xx} +  \pi_p(x) = 0$$
is of the form $v^\pm (x)$ $=$ $\widehat v^\pm (x) + u^\pm (x)$, where $u^\pm = \nu^\pm_1e^{\theta^\pm_1x}+\nu^\pm_2e^{\theta^\pm_2x}$, with $\theta^\pm _1, \theta_2 ^\pm$ as before, satisfies the homogenous ODE $- \beta u +\mu_\pm u_x + \frac{1}{2} \sigma^2 u_{xx} = 0$, and $\widehat v^\pm (x)$ is a particular solution given by
\begin{equation} \label{prod-val} \widehat v^\pm (x) = A x^2 + B_\pm x + C_\pm , \end{equation}
where the coefficients $A, B_\pm , C_\pm$ are identified as:
\begin{align*}
A  & = -\frac{d_1}{\beta}, \quad B_\pm = \frac{1}{\beta}\Big( d_0 - \frac{2 \, \mu_\pm \, d_1}{\beta} + c_p \, d_1\Big), \quad C_\pm = \frac{1}{\beta} \big(\mu_\pm B_\pm + A \sigma^2 - c_p d_0 \big).
\end{align*}

Assuming the producer adopts threshold--type impulse strategies defined by $\xi^\ast(x)$ in the intervention region, her expected payoff is of the form:
\begin{align}\label{eq:values-monopoly-prod}
v^\pm(x)&=\begin{cases}
v^\pm(\xs_h)-K_p(\xi^\ast(x)) & x\geq \x_h,\\ 
\Dq^\pm(x)+\nu^\pm_1e^{\theta^\pm_1x}+\nu^\pm_2e^{\theta^\pm_2x}, & \x_\ell<x<\x_h,\\
v^\pm(\xs_\ell)-K_p( \xi^\ast(x))  & x\leq \x_\ell. 
\end{cases}
\end{align}

When applying the optimal impulse $\xi^{\pm \ast}(x)$ at the threshold $x_{r}^\pm, r = \ell,h$, the producer brings $X_t$ back to the price level $x^{\pm\ast}_{r}$ $:=$ $x^\pm _{r} - \xi^{\pm \ast}(x^\pm _{r})$.
For optimality, the  respective impulse amounts satisfy the first order conditions
\begin{align}\label{prod-foc}
v_x^\pm(\xs_h) &= -\partial_\xi K_p(\xi^\ast(x^\pm_h)), \qquad 
v_x^\pm(\xs_\ell) =-\partial_\xi K_p(\xi^\ast(x^\pm_\ell)). 
\end{align}
We reinterpret the above as the equation to be satisfied by $\xi^\ast(x_r^{\pm})$ which are treated temporarily as unknowns and plugged into further equations.
To ensure that the value function is continuous at $x^\pm _{r}$ we further need
	\begin{align}
	v^\pm (x^\pm _{r}) & = v^\pm (x^{\pm \ast}_{r}) - K_p (\xi^{\pm \ast} _r). \label{prod-cont}
	\end{align}
Finally,  making the hypothesis that the value function is differentiable at the borders of the intervention region, we have:
\begin{subequations}\label{prod-diff}
	\begin{align}
	v^\pm _x(x^\pm _{\ell}) & = v^\pm _x(x^{\pm \ast}_{\ell})   -  \partial_\xi K_p(\xi^\ast(x^\pm_\ell)) , \label{prod-diff1} \\
	v^\pm _x(x^\pm _{h})   & =  v^\pm _x(x^{\pm \ast}_{h})  - \partial_\xi K_p(\xi^\ast(x^\pm_h)).\label{prod-diff2}
	\end{align}
\end{subequations}

We consider two cases of impulse costs:  (i) constant $K_p(\xi) =\kappa_0$ and (ii) linear $K_p(\xi) = \kappa_0 + \kappa_1 |\xi|$. In case (i), because the impulse cost is independent of the intervention amount there will be an optimal impulse level $\xs_r$ so that for any $x$ in the intervention region the strategy is to impulse back to $\xs_r$ which is the same at the two thresholds. In case (ii), $\partial_\xi K_p = \pm\kappa_1$ and
all the smooth pasting and boundary conditions can be gathered in the following system:
\begin{align}\label{eq:prodMonoSys}
\begin{cases}
\Dq^\pm(\x_h)+\nu^\pm_1e^{\theta^\pm_1\x_h}+\nu^\pm_2e^{\theta^\pm_2\x_h}=\Dq^\pm(\xs_h)+\nu^\pm_1e^{\theta^\pm_1\xs_h}+\nu^\pm_2e^{\theta^\pm_2\xs_h}-\kappa_0-\kappa_1(\x_h-\xs_h), & (\mathcal{C}^0\text{ at }\x_h)\\
\Dq^\pm(\x_\ell)+\nu^\pm_1e^{\theta^\pm_1\x_\ell}+\nu^\pm_2e^{\theta^\pm_2\x_\ell}=\Dq^\pm(\xs_\ell)+\nu^\pm_1e^{\theta^\pm_1\xs_\ell}+\nu^\pm_2e^{\theta^\pm_2\xs_\ell}-\kappa_0-\kappa_1(\xs_\ell-\x_\ell), & (\mathcal{C}^0\text{ at }\x_\ell)\\
\Dq^\pm_x(\xs_h)+\nu^\pm_1\theta^\pm_1e^{\theta^\pm_1\xs_h}+\nu^\pm_2\theta^\pm_2e^{\theta^\pm_2\xs_h}=-\kappa_1& (\mathcal{C}^1\text{ at }\xs_h)\\
\Dq^\pm_x(\xs_\ell)+\nu^\pm_1\theta^\pm_1e^{\theta^\pm_1\xs_\ell}+\nu^\pm_2\theta^\pm_2e^{\theta^\pm_2\xs_\ell}=\kappa_1, & (\mathcal{C}^1\text{ at }\xs_\ell)\\
\Dq^\pm_x(\x_h)+\nu^\pm_1\theta^\pm_1e^{\theta^\pm_1\x_h}+\nu^\pm_2\theta^\pm_2e^{\theta^\pm_2\x_h}=\Dq^\pm_x(\xs_h)+\nu^\pm_1\theta^\pm_1e^{\theta^\pm_1\xs_h}+\nu^\pm_2\theta^\pm_2e^{\theta^\pm_2\xs_h}-\kappa_1, & (\mathcal{C}^1\text{ at }\x_h)\\
\Dq^\pm_x(\x_\ell)+\nu^\pm_1\theta^\pm_1e^{\theta^\pm_1\x_\ell}+\nu^\pm_2\theta^\pm_2e^{\theta^\pm_2\x_\ell}=\Dq^\pm_x(\xs_\ell)+\nu^\pm_1\theta^\pm_1e^{\theta^\pm_1\xs_\ell}+\nu^\pm_2\theta^\pm_2e^{\theta^\pm_2\xs_\ell}+\kappa_1. & (\mathcal{C}^1\text{ at }\x_\ell)
\end{cases}
\end{align}
Note that there are two uncoupled linear systems for $v^+$ and $v^-$. The $\cC^0$ conditions are from \eqref{prod-cont}, the first two $\cC^1$ conditions are from \eqref{prod-foc} which determines the optimal impulse destination, and the last two $\cC^1$ conditions are from \eqref{prod-diff}. 

By a standard verification argument, one can show that if both systems above admit solutions $\nu^\pm _{1,2}$ and $x^\pm _{\ell,h}$, where the latter satisfy the order condition $x_\ell ^\pm < x_h ^\pm$, then the functions $v^\pm (x)$ as in \eqref{eq:values-monopoly-prod} are the value functions of the producer and his optimal strategies are given by the thresholds $x_{\ell,h}^\pm$ and impulse amounts $\xi^{\ast}(x_{\ell,h}^{\pm \ast})$. This can be done by following exactly the arguments in, e.g., \cite{CLP10} (see also their Remark 2.1), which are very standard in the literature of impulse control problems. Therefore, details are omitted.

\subsubsection{Non-preemptive Response}
Suppose the following ordering, which is similar to \eqref{eq:OrderInXY}, holds:
\begin{equation}\label{order-best-resp-prod} x_\ell ^\pm < y_l < y_h < x_h ^\pm .\end{equation}
We then expect  $v^\pm$ to solve the VIs
\begin{align}\label{eq:QVI-prod-nonpree}
\begin{cases}
\sup\big\{-\beta v^+ +\mu_+ v^+_x+\frac{1}{2}\sigma^2v^+_{xx} + \pi_p\; ;\; \sup_\xi (v^+(\cdot-\xi)-v^+-K_p(\xi))\big\}=0,\\
\sup\big\{-\beta v^- + \mu_- v^-_x+\frac{1}{2}\sigma^2v^-_{xx} + \pi_p\; ;\; \sup_\xi (v^-(\cdot-\xi)-v^--K_p(\xi))\big\}=0.\\
\end{cases}
\end{align}
To obtain the producer best--response it suffices to identify the two active impulse thresholds  $x^+_\ell,x^-_h$ and the respective target levels $x^{+\ast}_\ell, x^{-\ast}_h$. The other two boundary conditions take place at the consumer thresholds $y_\ell, y_h$, so that the strategy (see~\eqref{eq:DS_OneSide} below) is
$\mathcal{C}_p=\begin{bmatrix}
x^+_\ell,& x^{+\ast}_\ell, & -,&+\infty\\
-\infty, & - ,& x^{-\ast}_h,& x^{-}_h
\end{bmatrix}.$
The game coupling shows up in the additional \emph{boundary} condition that when the consumer switches, the producer's value is unaffected:
\begin{equation}\label{cons-on-prod} v^+ ( y) = v^- ( y), \qquad y \in (-\infty, y_\ell] \cup [y_h, +\infty).\end{equation}
Accordingly, our ansatz is
\begin{subequations}\label{eq:DS_OneSide}
\begin{align}
v^-(x)&=\begin{cases}
v^-(x^{-\ast}_h)-K_p( \xi^\ast( x)), & x\geq x^-_h,\\
\Dq^-(x)+\nu^-_1e^{\theta^-_1x}+\nu^-_2e^{\theta^-_2x}, & y_\ell<x<x^-_h,\\
v^+(x), & x\leq y_\ell,
\end{cases}\\
v^+(x)&=\begin{cases}
v^-(x), & x\geq y_h,\\
\Dq^+(x)+\nu^+_1e^{\theta^+_1x}+\nu^+_2e^{\theta^+_2x}, & x^+_\ell<x<y_h,\\
v^+(x^{+\ast}_\ell)-K_p( \xi^\ast(x)), & x\leq x^+_\ell.
\end{cases}
\end{align}
\end{subequations}
To simplify the presentation, let us concentrate on the proportional impulse costs $K_p(\xi) = \kappa_0 + \kappa_1 |\xi|$. We have the smooth pasting $\mathcal{C}^1$ and boundary conditions: 
\begin{align}\label{eq:prod-sys}
\begin{cases}
\Dq^+(y_\ell)+\nu^+_1e^{\theta^+_1y_\ell}+\nu^+_2e^{\theta^+_2y_\ell}=\Dq^-(y_\ell)+\nu^-_1e^{\theta^-_1y_\ell}+\nu^-_2e^{\theta^-_2y_\ell}, & (\mathcal{C}^0\text{ at }y_\ell)\\
\Dq^-(y_h)+\nu^-_1e^{\theta^-_1y_h}+\nu^-_2e^{\theta^-_2y_h}=\Dq^+(y_h)+\nu^+_1e^{\theta^+_1y_h}+\nu^+_2e^{\theta^+_2y_h}, & (\mathcal{C}^0\text{ at }y_h)\\
\Dq^+(x^+_\ell)+\nu^+_1e^{\theta^+_1x^+_\ell}+\nu^+_2e^{\theta^+_2x^+_\ell}=\Dq^+(x^{+\ast}_\ell)+\nu^+_1e^{\theta^+_1x^{+\ast}_\ell}+\nu^+_2e^{\theta^+_2x^{+\ast}_\ell}-K_p( \xi^\ast( x^+_\ell)), & (\mathcal{C}^0\text{ at }x^+_\ell)\\
\Dq^-(x^{-}_h)+\nu^-_1e^{\theta^-_1x^{-}_h}+\nu^-_2e^{\theta^-_2x^{-}_h}=\Dq^-(x^{-\ast}_h)+\nu^-_1e^{\theta^-_1x^{-\ast}_h}+\nu^-_2e^{\theta^-_2x^{-\ast}_h}-K_p(\xi^\ast( x^-_h)), &(\mathcal{C}^0\text{ at }x^-_h)\\
\Dq^+_x(x^+_\ell)+\nu^+_1\theta^+_1e^{\theta^+_1x^+_\ell}+\nu^+_2\theta^+_2e^{\theta^+_2x^+_\ell}=\Dq^+_x(x^{+\ast}_\ell)+\nu^+_1\theta^+_1e^{\theta^+_1x^{+\ast}_\ell}+\nu^+_2\theta^+_2e^{\theta^+_2x^{+\ast}_\ell} {-\kappa_1}, & (\mathcal{C}^1\text{ at }x^+_\ell)\\
\Dq^-_x(x^{-}_h)+\nu^-_1\theta^-_1e^{\theta^-_1x^{-}_h}+\nu^-_2\theta^-_2e^{\theta^-_2x^{-}_h}=\Dq^-_x(x^{-\ast}_h)+\nu^-_1\theta^-_1e^{\theta^-_1x^{-\ast}_h}+\nu^-_2\theta^-_2e^{\theta^-_2x^{-\ast}_h} {+\kappa_1}. &(\mathcal{C}^1\text{ at }x^-_h)\\
\Dq^+_x(x_\ell^{+\ast})+\nu^+_1\theta^+_1e^{\theta^+_1x_\ell^{+\ast}}+\nu^+_2\theta^+_2e^{\theta^+_2 x_\ell^{+\ast}} = - \kappa_1& {(\mathcal{C}^1\text{ at } x_\ell^{+\ast}) }\\
\Dq^-_x(x_h^{-\ast})+\nu^-_1\theta^-_1e^{\theta^-_1 x_h^{-\ast}}+\nu^-_2\theta^-_2e^{\theta^-_2 x_h^{-\ast}} = \kappa_1, & {(\mathcal{C}^1\text{ at } x_h^{-\ast})}
\end{cases}
\end{align}
Unlike the single--agent setting \eqref{eq:prodMonoSys}, the equations~\eqref{eq:prod-sys} are coupled. The coefficients $\nu^{\pm}_{1,2}$ are the solution to the linear system
\begin{align}
\begin{bmatrix}
e^{\theta^+_1y_\ell} & e^{\theta^+_2y_\ell} &  -e^{\theta^-_1y_\ell} & -e^{\theta^-_2y_\ell} \\
e^{\theta^+_1 x^+_\ell}-e^{\theta^+_1 x^{+\ast}_\ell} & e^{\theta^+_2 x^+_\ell}-e^{\theta^+_2 x^{+\ast}_\ell} & 0 & 0 \\
-e^{\theta^+_1y_h} & -e^{\theta^+_2y_h} & e^{\theta^-_1y_h} & e^{\theta^-_2y_h} \\
0 & 0 & e^{\theta^-_1 x^-_h}-e^{\theta^-_1 x^{-\ast}_h} & e^{\theta^-_2 x^-_h}-e^{\theta^-_2 x^{-\ast}_h}
\end{bmatrix}
\cdot
\begin{bmatrix}
\nu^+_1\\
\nu^+_2\\
\nu^-_1\\
\nu^-_2
\end{bmatrix}
=\begin{bmatrix}
\Dq^-(y_\ell)-\Dq^+(y_\ell)\\
\Dq^+(x^{+\ast}_\ell)-\Dq^+(x^{+}_\ell)-K_p\\
\Dq^+(y_h)-\Dq^-(y_h)\\
\Dq^-(x^{-\ast}_h)-\Dq^-(x^{-}_h)-K_p
\end{bmatrix}
\end{align}
and the thresholds $x_h^+, x_\ell^-$ are determined by the $\cC^1$ smooth--pasting (recall that $x^{-\ast}_h = x^{-}_h -\xi^\ast(x^-_h)$, $x^{+\ast}_\ell = x^{+}_\ell -\xi^\ast(x^+_\ell)$):
\begin{align}
\begin{cases}
v^-_x(x^-_h)=v^-_x(x^{-\ast}_h),\\
v^+_x(x^+_\ell)=v^+_x(x^{+\ast}_\ell),
\end{cases}
\end{align}
and the first order conditions (FOCs) giving the optimal impulses:
\begin{align}
v_x^-(x_h^{-\ast}) = -\partial_\xi K_p(\xi^\ast(x^-_h)) \qquad v_x^+(x_\ell^{+\ast}) &= - \partial_\xi K_p(\xi^\ast(x^+_\ell)).
\end{align}

\begin{prop}\label{prop:prod-best-response1} Let the $8$-tuple $(\nu^\pm _1, \nu^\pm _2, x^+_h, x^- _\ell, x^{+*}_h, x^{-*} _\ell)$ be a solution to the system \eqref{eq:prod-sys}, such that the order in \eqref{order-best-resp-prod} is fulfilled and $x^+ _\ell < x_\ell ^{+*}, x_h ^{-*} < x_h ^-$. Let $v^\pm$ be defined in \eqref{eq:DS_OneSide} and assume
\begin{equation}\label{eq:SOC}
v_{xx} ^+(x_\ell ^{+*})  < 0, 
 \qquad v_{xx} ^- (x_h ^{-*}) < 0. 
\end{equation}
Then the functions $v^\pm$ are the best--response payoffs of the producer, and a best--response strategy is given by
\begin{align}
\tau^* _0 = 0, & \quad \tau^* _i = \inf \left\{ t > \tau^* _{i-1} : X^\ast _t \in \Gamma_p (t-) \right\}, \\
 \xi^*_i (x_\ell ^+) & = x_\ell ^{+*} - x_\ell ^+ , \qquad \xi^*_i (x_h ^-) = x_h ^{-} - x_h ^{-*} , \qquad i \ge 1,\end{align}
with $\Gamma_p (t) = \Gamma^+ _p \mathbf 1_{\{\mu_t = \mu_+\}} + \Gamma^- _p  \mathbf 1_{\{\mu_t = \mu_-\}}$, where $\Gamma_p ^+ = (-\infty, x^+_\ell ]$ and $\Gamma_p ^- = [x^- _h , +\infty)$, while $(X^*_t)$ follows the dynamics corresponding to the consumer's strategy $(\sigma_i)_{i \ge 1}$ and the producer's impulse strategy $(\tau^* _i , \xi_i ^*)_{i \ge 1}$.  \end{prop}

\subsubsection{Preemptive Response} \label{ssec:pbr-pre}
It is possible that  the static discounted future profit of the producer satisfies, say, $v^+(x) \geq v^-(x)$ for any $x$, so that he always prefers expansion regime to contraction regime.

In that case, the consumer switching at $y_h$ from expansion to contraction hurts the producer and
one possible strategy for him is to \textit{preempt} in order to prevent the consumer from switching the drift to $\mu_-$. This situation could be viewed as  looking for best $x^+_h < y_h$, given $y_h$. In the latter case the constrained solution could be $x^+_h = y_h-$, whereby the system  \eqref{eq:prodMonoSys} does not hold and the best--response is to impulse $(X_t)$ right before it hits $y_h$, $x^+_h = y_h-$. This strategy is not well-defined (i.e.~the supremum is not achieved on the open interval $(x_\ell^+, y_h)$), but the resulting preemptive best--response value in the $\mu_+$ regime can be obtained by using the ansatz (where we slightly abuse the notation to write $x_h^{+\ast} = y_h - \xi^\ast(y_h)$ for the target impulse level at $y_h$)
	\begin{align}\label{eq:prod-payoff-2}
	v^+(x)&=\begin{cases}
	v^+(x^{+\ast}_h)-K_p(\xi^\ast(x)), & x\geq y_h,\\
	\Dq^+(x)+\nu^+_1e^{\theta^+_1x}+\nu^+_2e^{\theta^+_2x}, & x^+_\ell < x < y_h,\\
	v^+(x^{+\ast}_\ell)-K_p(\xi^\ast(x)), & x\leq x^+_\ell,
	\end{cases}
	\end{align}
and the boundary conditions for determining the target impulse levels
\begin{align}
v_x^+(x^{+\ast}_h) &= - \kappa_1, \quad  v_x^+(x_\ell^{+\ast}) = +\kappa_1.
\end{align}
Note that we now have 5 unknowns, $\nu^+_{1,2}, x^+_\ell, x^{+\ast}_{\ell}, x^{+\ast}_h$ rather than six as we ``fixed'' $x^+_h = y_h$.  This yields the following system
\begin{align}\label{eq:prod-sys-2}
\begin{cases}
\Dq^+(y_h)+\nu^+_1e^{\theta^+_1y_h}+\nu^+_2e^{\theta^+_2y_h}=\Dq^+(x^{+\ast}_h)+\nu^+_1e^{\theta^+_1(x^{+\ast}_h)}+\nu^+_2e^{\theta^+_2(x^{+\ast}_h)}-K_p(\xi^*(y_h)) & (\mathcal{C}^0\text{at }y_h)\\
\Dq^+(x^{+}_\ell)+\nu^+_1e^{\theta^+_1x^{+}_\ell}+ \nu^+_2 e^{\theta^+_2x^{+}_\ell}=\Dq^+(x^{+\ast}_\ell)+\nu^+_1e^{\theta^+_1x^{+\ast}_\ell}+\nu^+_2e^{\theta^+_2x^{+\ast}_\ell}-K_p(\xi^*(x_\ell ^+)) &(\mathcal{C}^0\text{ at }x^+_\ell)\\
\Dq^+_x(x^{+}_\ell)+\nu^+_1\theta^+_1e^{\theta^+_1x^{+}_\ell}+\nu^+_2\theta^+_2e^{\theta^+_2x^{+}_\ell}=\Dq^+_x(x^{+\ast}_\ell)+\nu^+_1\theta^+_1e^{\theta^+_1x^{+\ast}_\ell}+\nu^+_2\theta^+_2e^{\theta^+_2x^{+\ast}_\ell} +\kappa_1 &(\mathcal{C}^1\text{ at }x^+_\ell) \\
\widehat v_x^+(x^{+\ast}_h) + \nu^+_1\theta^+_1e^{\theta^+_1 x_h^{+\ast} }+\nu^+_2\theta^+_2e^{\theta^+_2  x_h^{+\ast} } = - \kappa_1 & (\mathcal{C}^1\text{ at }x^{+\ast}_h) \\
\Dq^+_x(x^{+\ast}_\ell)+\nu^+_1\theta^+_1e^{\theta^+_1x^{+\ast}_\ell}+\nu^+_2\theta^+_2e^{\theta^+_2x^{+\ast}_\ell} = \kappa_1. & (\mathcal{C}^1\text{ at }x^{+\ast}_\ell)\\
\end{cases}
\end{align}

Preemption in the contraction regime writes in a symmetric way.

In general, we need to manually verify whether $x_h^+ > y_h$ (the ``normal'' case) or $x^+ _h = y_h$ (the preemptive case) whenever we consider the producer best--response. The two situations lead to different boundary conditions at the upper threshold, and hence cannot be directly compared. 
Considering the optimization problem for $x_h^+$, we expect his value function to increase in $x_h^+$ on $(x_\ell^+, y_h)$ and experience a positive jump at $y_h$, i.e.~conditional on someone acting,  the producer prefers the consumer's switch to applying his impulse. However, if this is not the case, the consumer action hurts the producer and assuming the impulse costs are low, the best-response is $x_h^+ = y_h$. This corner solution arises due to the underlying discontinuity: on $(x_\ell^+, y_h)$ the producer compares the value of waiting to the value of doing an optimal impulse, but at $y_h$ he compares the value of switching to that of doing an optimal impulse. So it could be that ``waiting'' $>$ impulsing $>$ switching at $y_h$, leading to pre-emptive impulse to prevent the worst (for the producer) outcome.

\begin{prop} \label{prop:prod-best-response2} Assume $\mu_{0-}=\mu_+$. Let the $5$-tuple $(\nu^+_1, \nu^+ _2, x^+_\ell , x^{+*}_\ell, x^{+*} _h)$ be a solution to the system \eqref{eq:prod-sys-2}, such that the order in \eqref{order-best-resp-prod} is fulfilled and $x^+ _\ell < x_\ell ^{+*}, x_h ^{+*} < y_h$. Let $v^+$ be defined as in \eqref{eq:prod-payoff-2} and assume \begin{equation}\label{eq:SOC-2}
v_{xx} ^+(x_\ell ^{+*}) <0,
\qquad v_{xx} ^+ (x_h ^{+*}) <0. 
\end{equation}
Then, the function $v^+$ is the best--response payoff of the producer in the expansion regime, and a best--response strategy is given by
\begin{align}
\tau^* _0 = 0, & \quad \tau^* _i = \inf \{ t > \tau^*_{i-1} : X_t \in \Gamma_p (t-) \}, \\
 \xi^*_i (x_\ell ^+) & = x_\ell ^{+*} - x_\ell ^+ , \quad \xi^*_i (y_h) = y_h - x_h ^{+*} , \quad i \ge 1,\end{align}
with $\Gamma_p (t) = \Gamma^+ _p (t) = (-\infty, x^+_\ell ] \cup [y_h , +\infty)$, while $X^\ast$ follows the dynamics corresponding to the producer's impulse strategy $(\tau^* _i , \xi_i ^*)_{i \ge 1}$.
\end{prop}

\begin{figure}[th!]
	\centering
	\includegraphics[width=0.35\textwidth]{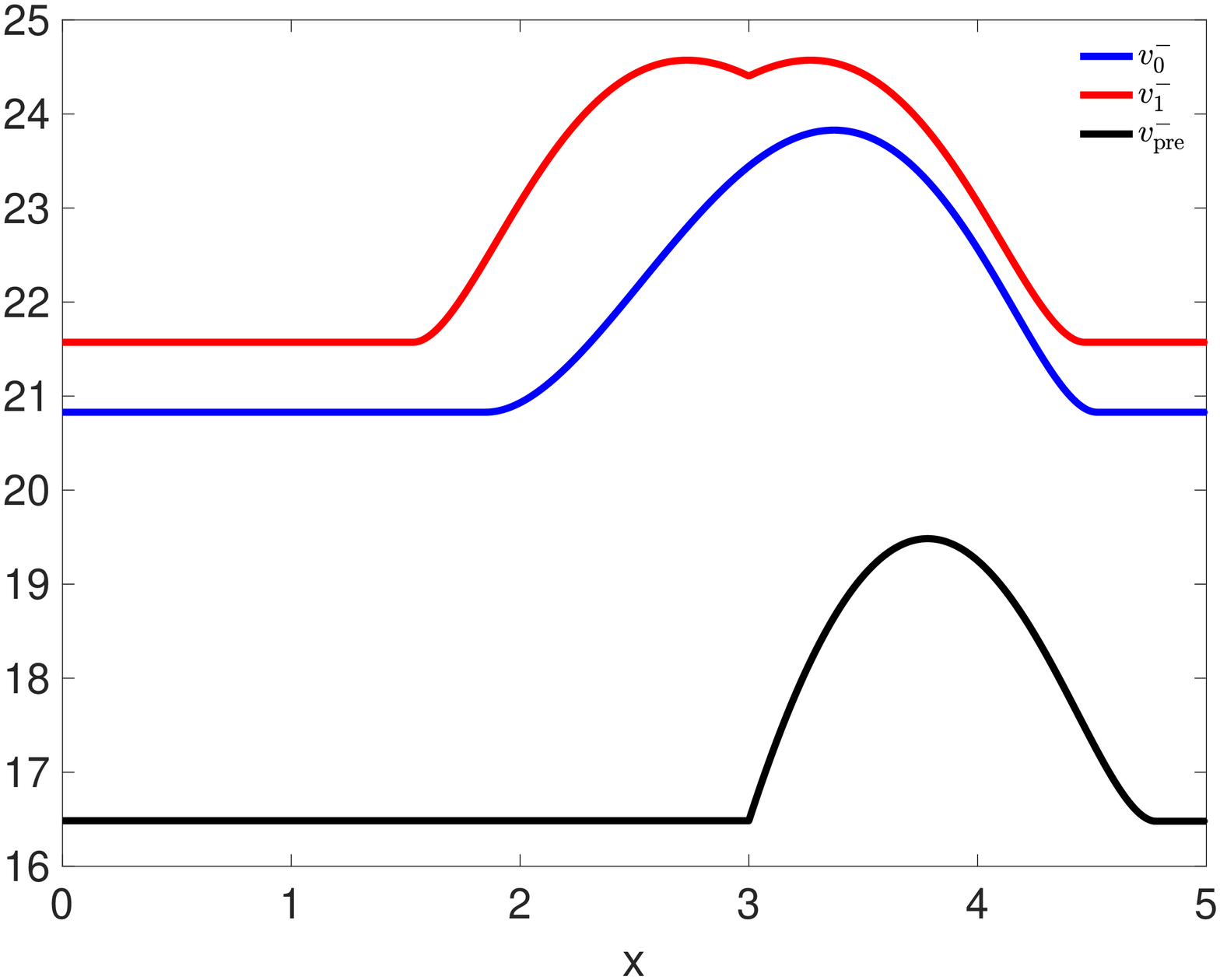}
	\includegraphics[width=0.35\textwidth]{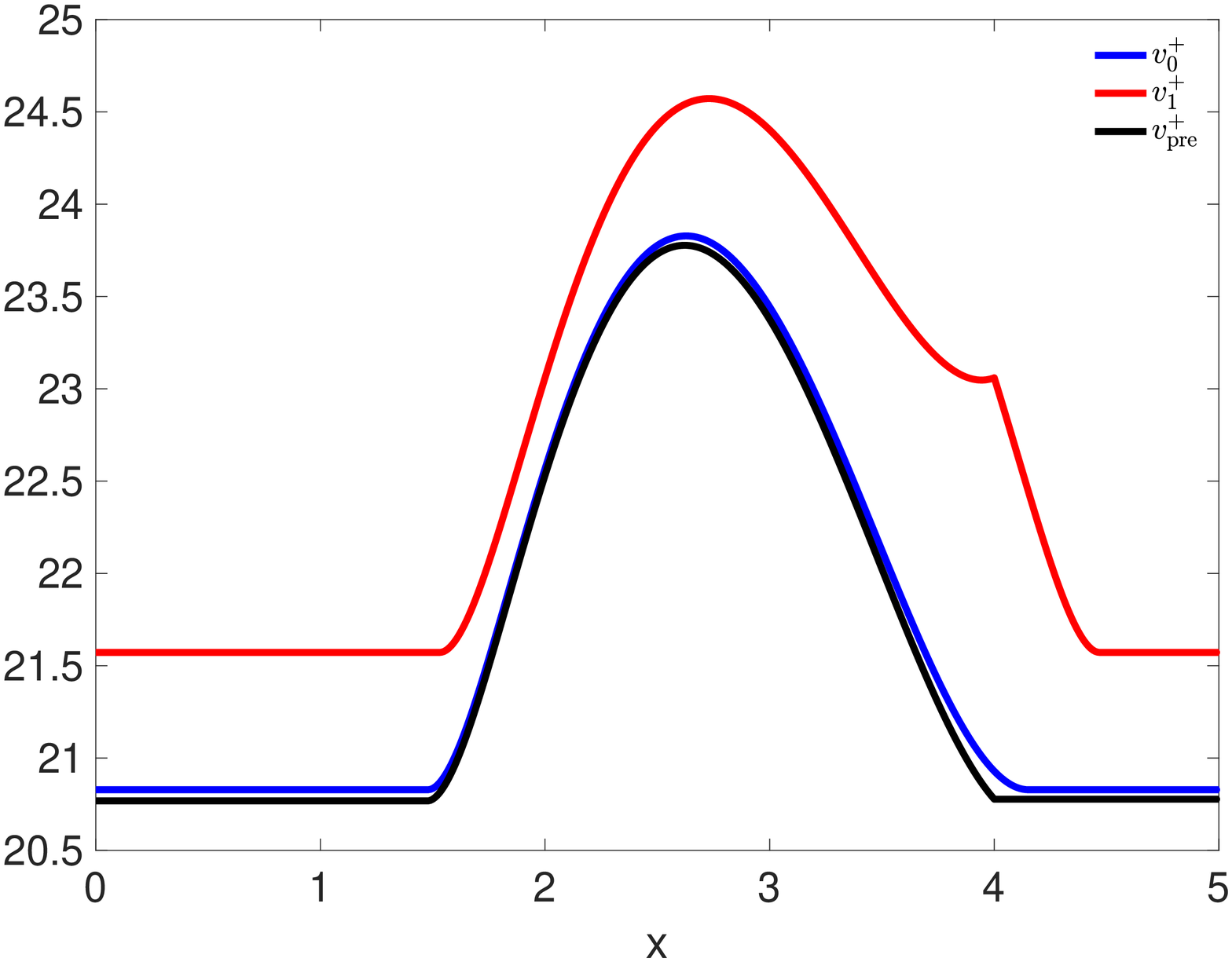}
	\caption{Value functions $v^\pm_0$, $v^\pm_1$, and $v^\pm_{\rm pre}$ of the producer given the consumer's strategy $\mathcal{C}_c=  (3,  4).$}
	\label{fig:prodBR01pre}
\end{figure}

Figure~\ref{fig:prodBR01pre} illustrates the shapes of the producer's value function in the different cases of best--response. For the given consumer strategy, we have a dominant function in the contraction regime ($v^-_1$) and a dominant function in the expansion regime ($v_1^+$).

\section{Equilibria}\label{sec:eqm-types}

The best--response functions defined in Section~\ref{sec:eqm} lead to three types of potential market equilibria, depending on the equilibrium behaviour of the consumer and characterized by the relative positions of the consumerand producer thresholds:
\begin{itemize}
\item Type~I -- generic:  $ y_\ell \leq x^-_h $ and $ x^+_\ell \leq y_h$.
\item Type~II --  transitory: $ -\infty = y_\ell \le x^-_\ell  $ and $ x^+_\ell \leq y_h$; or $y_\ell \leq x^-_h $ and $y_h = +\infty$.
\item Type~III -- preemptive: $ y_\ell \leq x^-_h $ and $ x^+_\ell = y_h$; or $ y_\ell = x^-_h $ and $ x^+_\ell \leq y_h$.
\end{itemize}

In equilibrium Type~I, the consumer switches back and forth forever between the two expansion and contraction regimes. The optimal policy of the consumer is given by the threshold $y_\ell$ in the contraction regime and $y_h$ in the expansion regime, while the optimal policy of the producer is formed by the pair $(x_\ell^+,x_\ell^{+,\ast})$ in the expansion regime and symmetrically by the pair $(x_h^{+,\ast,}x_h^{-})$ in the contraction regime. We anticipate this to be the most common equilibrium type; it was precisely described and illustrated in Section~\ref{ssec:eqt1}.

In equilibrium Type~II, the consumer and the producer both prefer a given regime and thus, the consumer switches at most once when the market is initialized in the opposite regime. Afterwards, only the producer acts to maintain the price between $(x_\ell,x_h)$. Consider the case of a single switch  from expansion to contraction; the consumer's optimal policy consists then in only one threshold, $y_h$. The optimal policy of the producer is more complicated: in the expansion regime, it consists of the pair $(x_\ell^+,x_\ell^{+,\ast})$ and in the contraction regime, it consists of a quadruplet $(x^-_\ell,x^{-,\ast}_\ell,x^{-,\ast}_h,x^-_h)$. The same reasoning applies in the other single switch case. This equilibrium is described in Section~\ref{ssec:eqt2}.

The last type of equilibrium, named Type~III, resembles the preceding one in the sense that at most one switch can be observed. But it differs because here the consumer is stuck forever in a state she wishes to leave. In that case described in Section~\ref{ssec:eqt3}, only the producer acts. Starting in the expansion regime, for instance, the consumer would like to switch to the contraction regime when the price reaches a threshold $y_h$. But the producer, who prefers perpetual expansion, {\em preempts} the switch by acting at the threshold $y_h^-$, just before the action of the consumer.

Threshold--type equilibria offer analytical tractability to describe the long--run market behavior. The latter can be summarized by the stationary distribution of the commodity price $(X^\ast_t)$ and the consumer regimes $(\mu^\ast_t)$ as induced by the equilibrium strategies $(N^\ast_t, \mu^\ast_t)$. To quantify these effects, we define an auxiliary discrete--time jump chain $(M^\ast_n)_{n=0}^\infty$ which takes values in the state space
\begin{align}
E:= \left\{S_{+}, S_{-}, I^-_\ell, I^-_h, I^+_\ell, I^+_h\right\}.
\end{align}
The chain $M^\ast$ keeps track of the sequential actions of the players, where $S_\pm$ represents the switches of the consumer (``$S_{+}$'' stands for the switch $\mu_-\rightarrow\mu_+$ and ``$S_{-}$'' for  $\mu_+\rightarrow\mu_-$) and $I^\pm_r$ the impulses (up/down at the two impulse boundaries) of the producer. Thus, $M^\ast$ summarizes the sequence of market interventions stored within $\tau_i, \sigma_i$ stopping times.  Note that states $M^\ast_n \in \{S_+, I^+_{\ell h}\}$ imply a positive drift $\mu_+$ of $X^\ast$, while the rest imply a negative drift $\mu_-$. Moreover,
if the consumer adopts a double--switch strategy and the producer adopts a non-preemptive strategy as discussed in Section \ref{ssec:eqt1}, then the thresholds $x^+_h$ and $x^-_\ell$ will be hit at most once by $X^\ast$ and therefore the corresponding states $I^+_h$ and $I^-_\ell$ of $M^\ast_n$ are transient.

Because the dynamics of $X^\ast$ between interventions are always Brownian motion (BM) with drift, the transition probabilities of $M^\ast$ can be described in terms of hitting probabilities of a BM. This offers closed-form expressions for the the transition probability matrix $\mathbf{P}$ of $M^\ast$, and its invariant distribution denoted by $\vec{\Pi}$. Moreover, the sojourn times of $M^\ast$ correspond to $(X^\ast_t)$ hitting the various thresholds (in terms of the original continuous-time ``$t$'') and are similarly linked to BM first passage times. Combining the above ideas, we can then derive a complete description of $(\mu^\ast_t)$, namely the long-run proportion of time that the commodity demand is in expansion/contraction regimes and the respective expected switching time, see \eqref{eq:time-to-switch}.

In the following section otherwise stated, we use the parameter values in Table \ref{tab:params-2}, such that $\pi_i(x) = a_i (x-x^1_i)(x^2_i-x)$, $i\in\{c,p\}$. This yields consumer and producer preferred price levels of $\bar{X}_c= 3, \bar{X}_p=4$. The same set of parameters yields an equilibrium of each type, showing the non-uniqueness of equilibria in this model.

\begin{figure}[h]
\centering
\begin{minipage}{0.5\textwidth}
	\centering
	\includegraphics[width=0.9\textwidth]{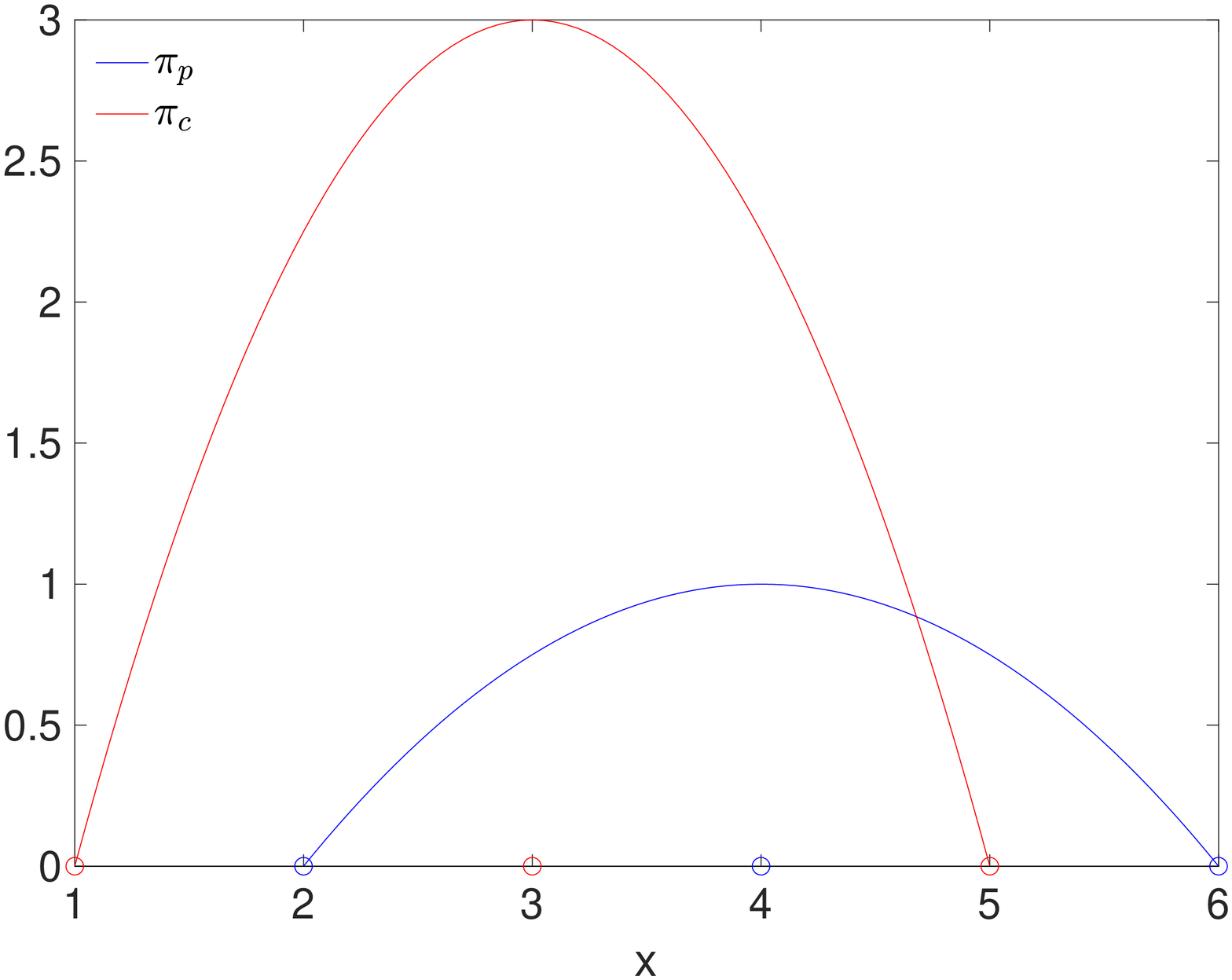}
	\caption{{\small Producer's and consumer's profit rates as a function of $x$.}} \label{fig:profit}
\end{minipage}\hspace{10mm}
\begin{minipage}{0.4\textwidth}
\centering\vspace{19mm}
\captionsetup{type=table}
        \begin{tabular}{l r | l r | l r}
        \hline\hline
        \multicolumn{2}{c}{Market} & \multicolumn{2}{c}{ Consumer} & \multicolumn{2}{c}{Producer}  \\ \hline
        $\beta$ & $0.1$  & $x^1_c$ & 1 & $x^1_p$ & 2 \\
         $\sigma$    & $0.25$ & $x^2_c$ & 5 & $x^2_p$ & 6 \\
        $\mu_+$   & $0.1$  & $a_c$ & 0.75 & $a_p$ & 0.25 \\
         $\mu_-$   & $-0.1$ & $h_\pm$   & $10$ &$\kappa_0$ & 3 \\
        &  & & & $\kappa_1$ & 0
         \\ \hline\hline
        \end{tabular}
	\caption{ {\small Model parameters for Section \ref{ssec:eqt1}.}}
        \label{tab:params-2}
\end{minipage}
\end{figure}

\subsection{Type I -- Generic}\label{ssec:eqt1}

We look for an \textit{interior, non-preemptive} equilibrium satisfying the ordering \eqref{eq:OrderInXY}, i.e.~a pair of consumer and producer strategies of the form $(y_\ell^\ast,y_h^\ast)$ and $(x^+_\ell,x^{+\ast}_\ell,x^{-\ast}_h,x^-_h)$. To construct this equilibrium, we employ t\^atonnement, i.e.~iteratively apply the best--response controls alternating between the two players. This corresponds to the interpretation of Nash equilibrium as a \emph{fixed point of best--response maps} $BR$. The equilibrium is obtained using two different fixed--point algorithms. Given strategies $\mathcal C_p^0$ and $\mathcal C_c^0$, we have either an asynchronous or synchronous algorithm, namely
\begin{align*}
&\mathcal C_p^{k+1}  = BR(\mathcal C_c^k), \quad && \mathcal C_p^{k+1}  = BR(\mathcal C_c^k), \\
&\mathcal C_c^{k+1}  = BR(\mathcal C_p^{k+1}), \quad && \mathcal C_c^{k+1}  = BR(\mathcal C_p^{k}), \\
&\text{asynchronous}  && \text{synchronous}.
\end{align*}

The resulting equilibrium found using both algorithms is the same and is
\begin{align}\label{eq:cq-ex1}
\mathcal{C}^{\rm{I},\ast}_p=\begin{bmatrix}
2.0, &  3.6, & -, & +\infty\\
-\infty, &  -, & 4.5, & 6.1 \\
\end{bmatrix}, \qquad
\mathcal{C}^{\rm{I},\ast}_c=[2.2, 4.4].
\end{align}

\begin{figure}[t!]
	\centering
	\includegraphics[width=0.3\textwidth]{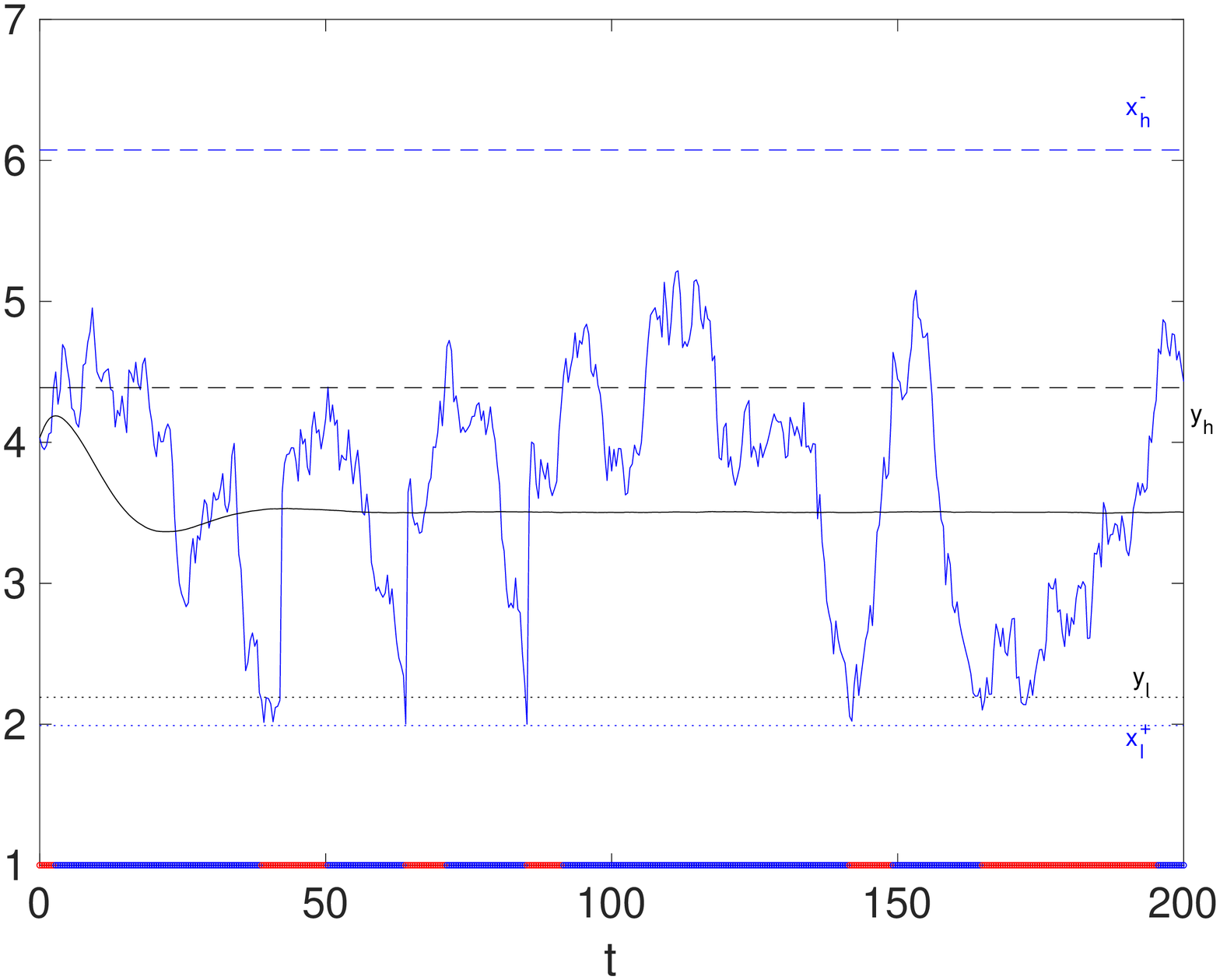}
	\includegraphics[width=0.3\textwidth]{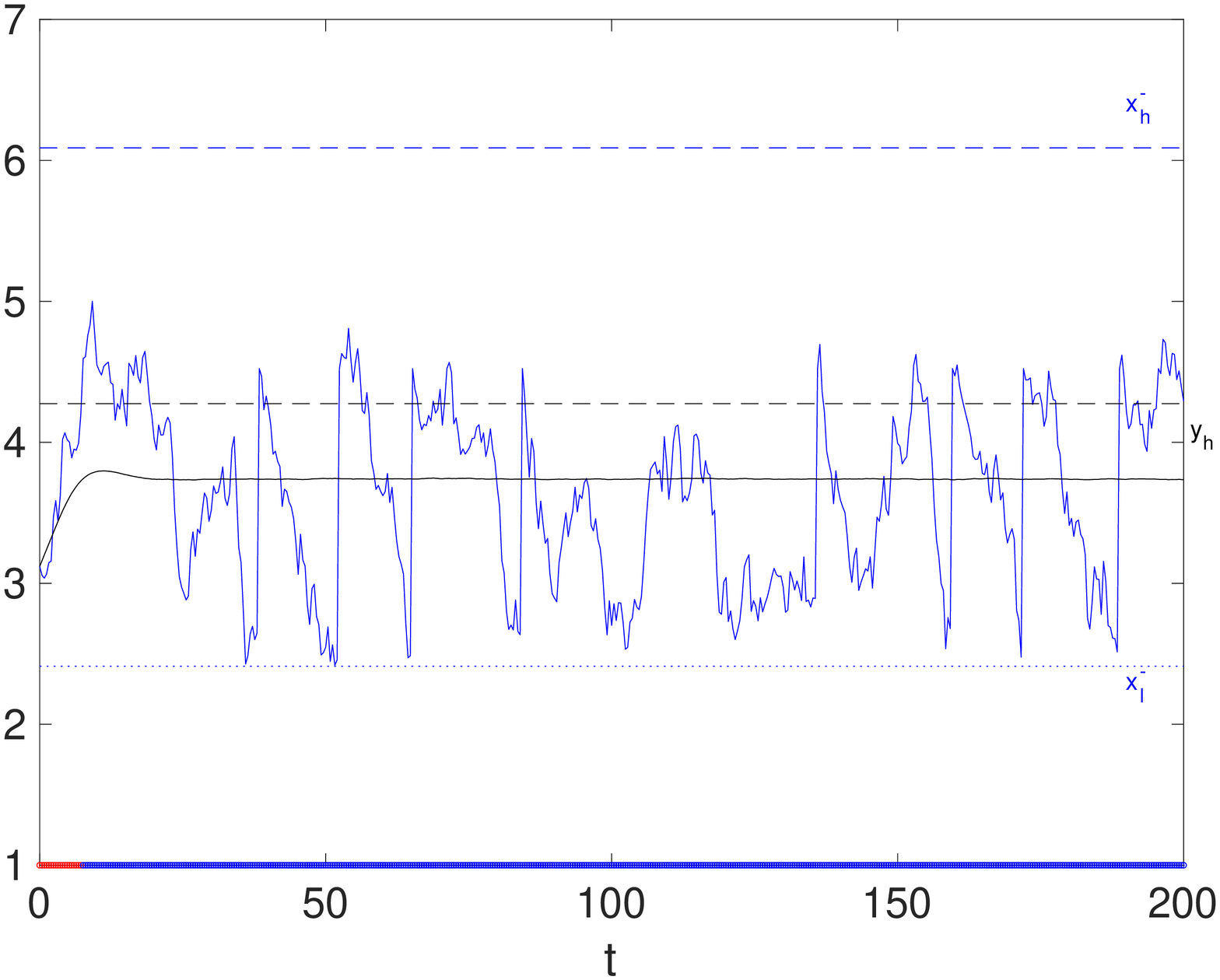}
	\includegraphics[width=0.3\textwidth]{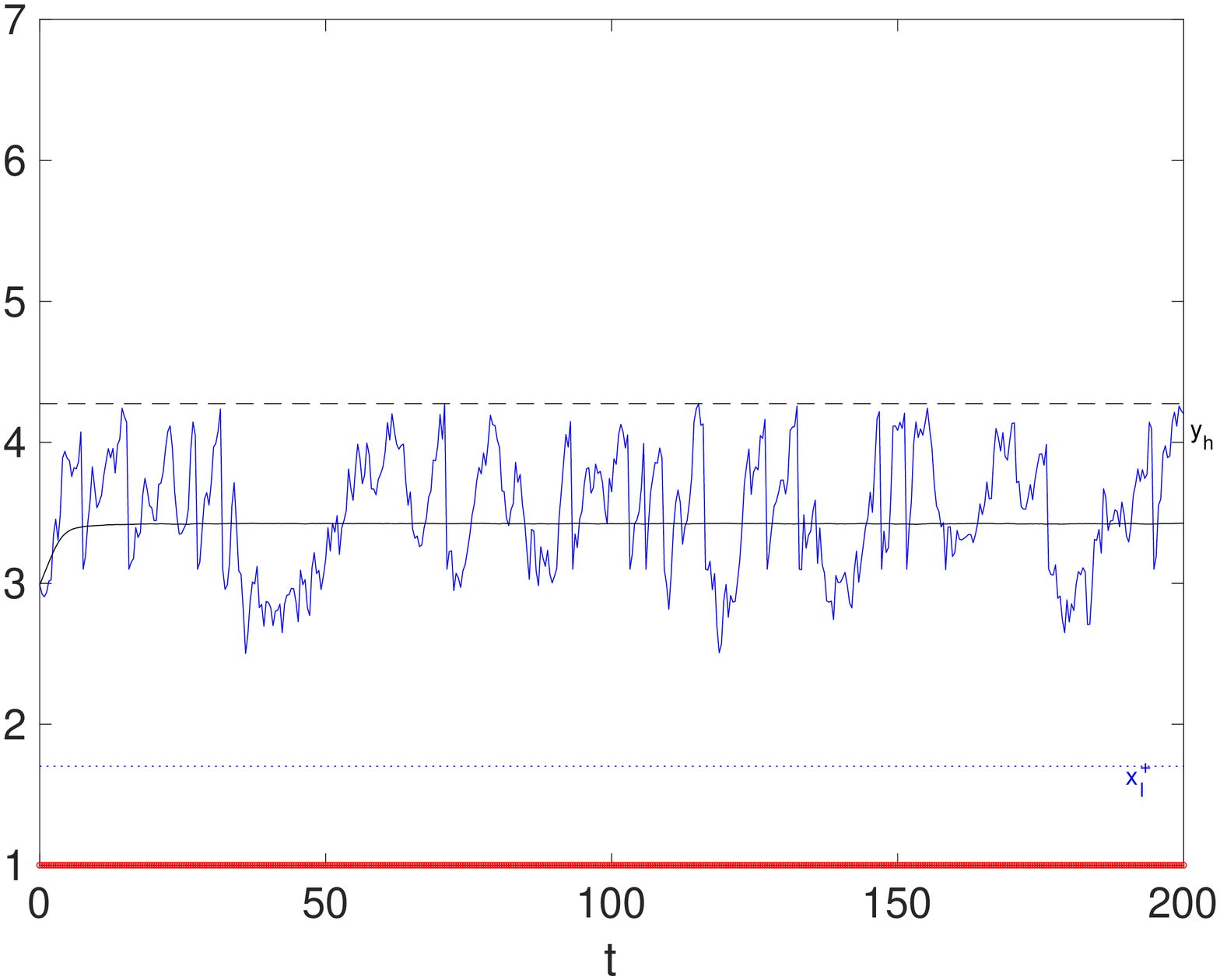}
	\caption{A sample path of the controlled market price $(X^\ast_t)$ under a Type I equilibrium (Left), Type~II (Center) and Type~III (Right) together with $\E\big[ X^\ast_t \big]$ (black solid curve). The colors along the $x$-axis indicate the corresponding $\mu^\ast_t$. We start with $\mu^\ast_0 = \mu_+$.}
	\label{fig:SamplePathT}
\end{figure}

The dynamic equilibrium of the commodity price  $X^\ast$ is illustrated in Figure~\ref{fig:SamplePathT} (Left). The market starts in the expansion regime, $\mu^\ast_0 = \mu_+$. We observe that $x^{-\ast}_h$ is close to $y_h^{\ast}$, implying that once the price has reached the switching level $y_h^\ast$, it is likely to touch soon thereafter the threshold $x^{-\ast}_h$, making the price drop to $x^{-\ast}_h$. The producer ``backs up'' this mean-reversion by impulsing down if prices rise too much and impulsing down if they drop too much. Otherwise, she lets the consumer be in charge via switching control that benefits him as well.

At equilibrium, the price $X^\ast$ fluctuates in a range of values where neither the producer nor the consumer have negative profit rate. If alone in the market, the optimal monopolistic strategies of the producer and the consumer are
\begin{align}\label{eq:cq-alone}
\mathcal{C}^m_p:=\begin{bmatrix}
1.9, &  3.5, & 3.5, & 5.6 \\
2.4, &  4.5, & 4.5, & 6.1 \\
\end{bmatrix}, \qquad
\mathcal{C}^{m}_c:=[1.7, 4.3].
\end{align}
We see that the equilibrium strategy of the producer $\mathcal C^{\rm{I},\ast}_p$ is quite close to what he would have done if alone in the market. On his side, the consumer-induced equilibrium price range is wider than he would prefer ($2.6$ against $2.2$ if alone). In equilibrium, it is as if the producer lets the consumer do the job of bringing back the price to his preferred level $\bar{X}_p$. The producer intervenes only if $X^\ast_t$ drops too low or gets too high, after the regime switching has occurred. But, in the long--run, the average price $\lim_{t\to\infty}\E[X^\ast_t]$ is close to $3.5$, which is the mid--value between $\bar X_p$ and $\bar X_c$.

The players' equilibrium strategy profile yields a stationary distribution for the pair $(X^\ast_t, \mu^\ast_t)$. The macro market $\mu^\ast$ switches between the expansion and the contraction regimes back and forth, while the jointly controlled price $(X^\ast_t)$ is bounded in the range $[x^+_\ell, x^-_h]$ and fluctuates in a \textit{mean-reverting} pattern due to alternating signs of its drift. These stylized features can be broadly traced in the world commodity markets which undergo cyclical Expansion/Contraction patterns.

\noindent {\bf Dynamics of $(X^\ast_t)$ in the equilibrium}: The dynamics of the commodity price $(X^\ast_t)$ are less tractable due to the impulses applied by the producer. Let $\phi^\ast(\cdot)$ denote the long-run (i.e.~stationary) distribution of $(X^\ast_t)$. In Figure~\ref{fig:EstDynamicX}, we show $\phi^\ast$  obtained from an empirical density based on a long trajectory of $(X^\ast_t)$, relying on Monte Carlo simulations and the ergodicity of the recurrent, bounded process $(X^\ast_t)$. For additional interpretability, we also plot the invariant distributions $\phi^\ast_\pm$ conditional on $\mu_t = \mu_\pm$.
\begin{figure}[h!]
	\centering
	\subfigure[]{\label{fig:EstDynamic_X} 
		\includegraphics[width=0.35\textwidth]{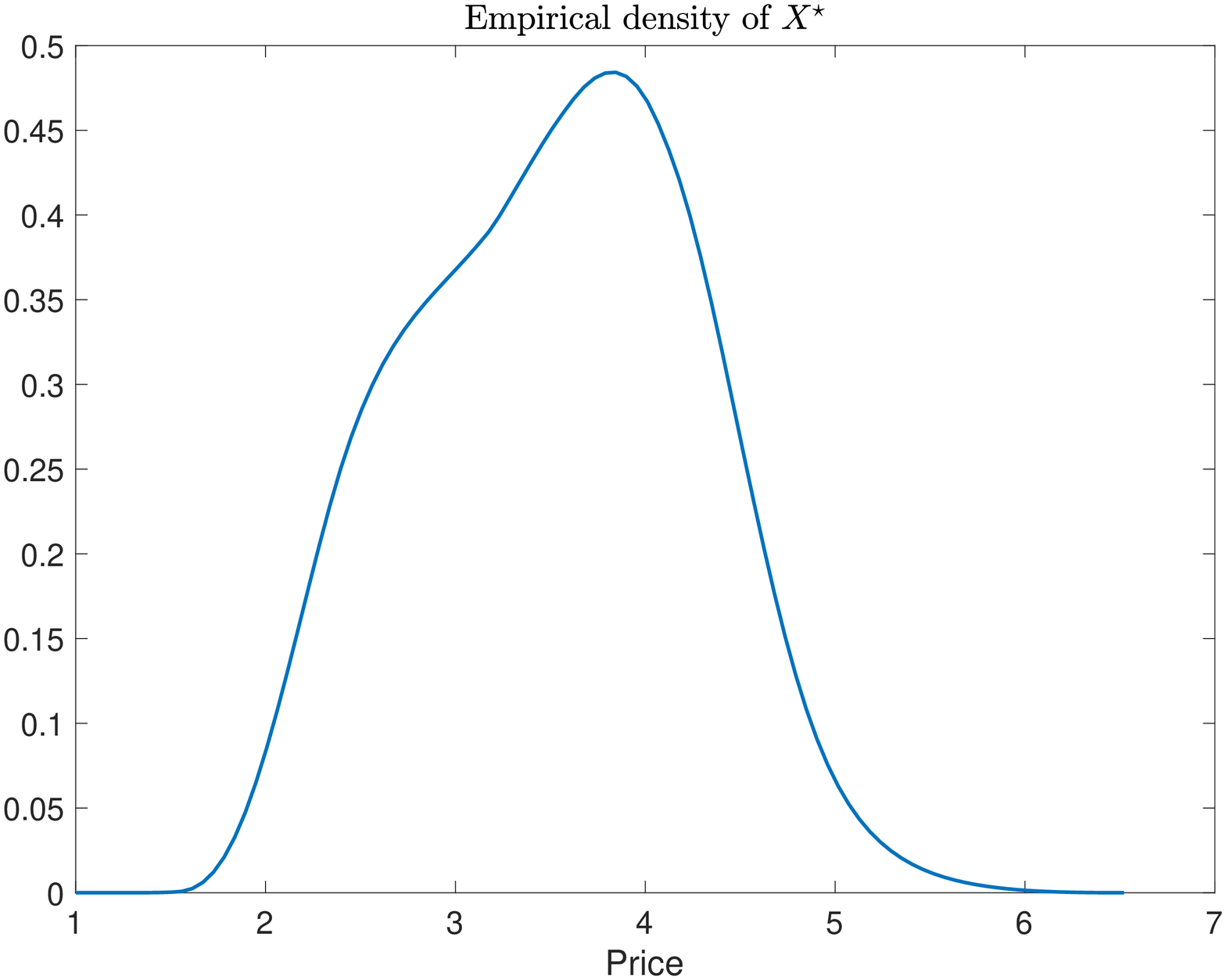}}
	\subfigure[]{
		\label{fig:EstDynamicMu_X}
		\includegraphics[width=0.35\textwidth]{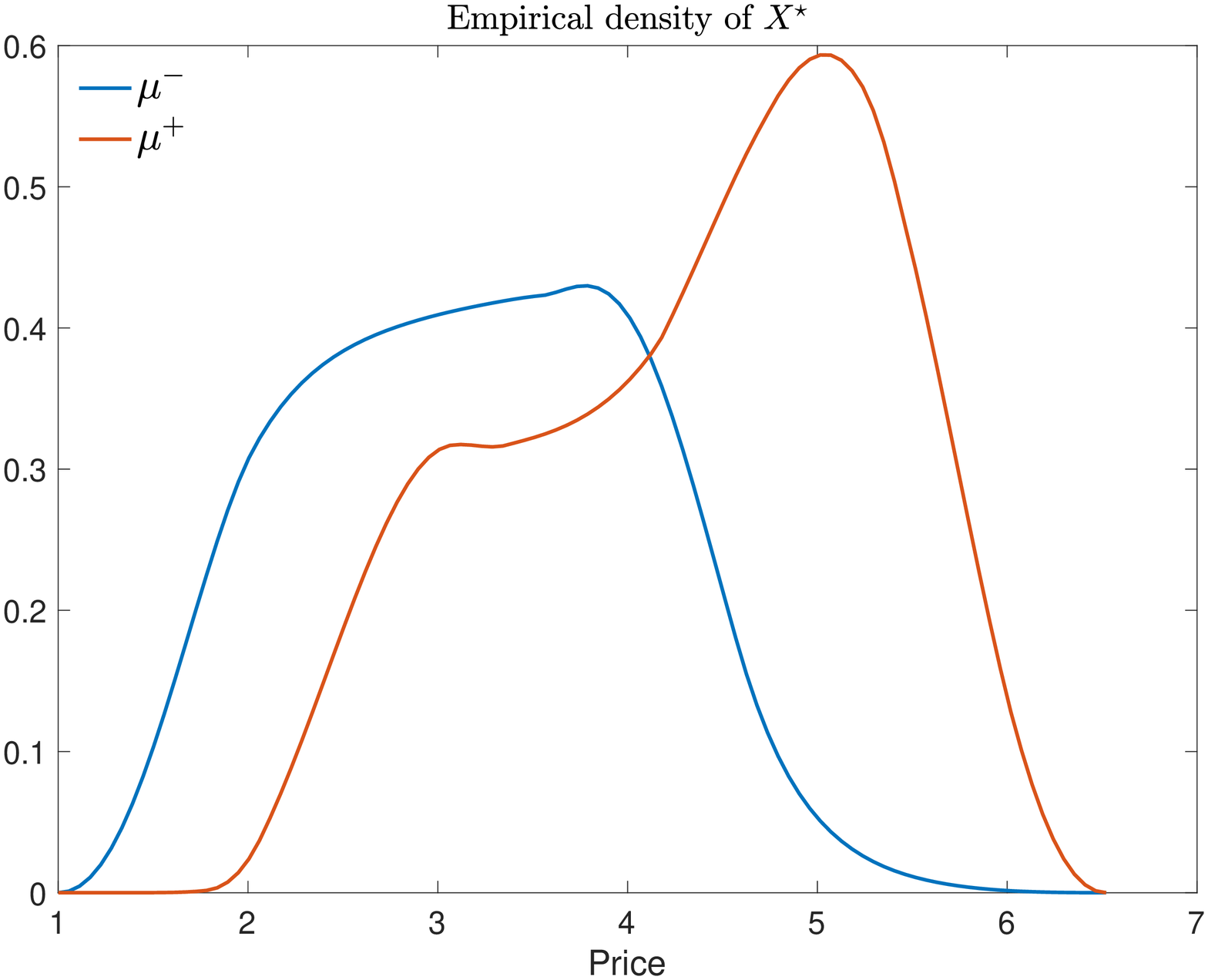}}
	\caption{Estimated long-run densities $\phi^\ast$ of $X^\ast$ in a double--switch and one--sided impulse equilibrium. Panel (a): overall kernel smoothed $\phi^\ast(\cdot)$. (b): long--run distributions $\phi^\ast_\pm$ of $X^\ast$ conditional on $\mu(t) =\mu_+$ (resp.~$\mu(t) = \mu_-$). Recall that the domain of $X^\ast_t$ depends on the current regime, so the support of the two densities differs.}
	\label{fig:EstDynamicX} 
\end{figure}

\subsection{Type II --  Transitory}
\label{ssec:eqt2}

In another type of equilibrium the consumer switches only in one regime, with the other being absorbing. For this reason, we name it {\em transitory}. To fix ideas, suppose that the consumer only switches from expansion regime to contraction regime. In that case, the producer effectively acts like a profit maximizing monopoly in the contraction regime with two--sided impulses; in the expansion regime she will apply a one--sided impulse as in the equilibrium type I.

To solve for such equilibrium, we first compute the producer strategy in the contraction regime which is a decoupled VI as in  \eqref{eq:values-monopoly-prod}  leading to the 6 equations in \eqref{eq:prodMonoSys} but only for $v^-$, $x_r^-, x_r^{-*}$, $r \in \{ \ell, h \}$. This solution induces the corresponding no--switch solution $\w_0^-$ of the consumer as in \eqref{inaction-eq1}-\eqref{inaction-eq2}. Both $v^-, \w_0^-$ are then fixed and act as source terms to solve for the equilibrium in the expansion regime. For the latter, we need to compute $v^+(\cdot)$ and the associated thresholds $x_\ell^+, x_\ell^{+*}$ (only one threshold), as well as $\w^+(x)$ and the switching threshold $y_h$ (note that there is no $y_\ell$). The boundary conditions are $v^+(y_h) = v^-(y_h), \w^+(y_h) = \w^-(y_h) - h_0$.

This reasoning leads to the following algorithm. If $x^+_\ell$ and $x^{+,\ast}_\ell$ are fixed, we compute the best--response of the consumer by solving the variational problem for the consumer value function $\w^+$ such that:
\begin{align*}
w^+(x) & = \begin{cases}
	 	 w_0^-(x) - h_0, & y_h \leq x,\\
		\Dd^+(x) + \lambda^+_1 e^{\theta^+_1x} + \lambda^+_2 e^{\theta^+_2x}, & x^+_\ell  < x < y_h,\\
	 	w^+(x^{+,\ast}_\ell), & x \leq x^+_\ell.
	\end{cases}
\end{align*}
This is exactly the best--response in the single--switch case with the solution given by the system~\eqref{eq:single-switch} and which provides the consumer's threshold $y_h$. Now, if we consider that $y_h$ is fixed, we can compute the best--response of the producer by solving a VI for the value function $v^+$ that satisfies
\begin{align*}
v^+(x) & = \begin{cases}
	 	 v^-(x), & y_h \leq x,\\
		\Dq^+(x) + \nu^+_1e^{\theta^+_1x} + \nu^+_2e^{\theta^+_2x}, & x^+_\ell  < x < y_h,\\
	 	v^+(x^{+\ast}_\ell) - \kappa_0 - \kappa_1 (x^{+,\ast}_\ell - x), & x \leq x^+_\ell.
	\end{cases}
\end{align*}
The boundary conditions giving the four unknowns $(x^+_\ell,x^{+,\ast}_\ell, \nu^+_1,\nu^+_2)$ are:
\begin{align}
\begin{cases}
v^-(y_h)=\Dq^+(y_h)+\nu^+_1e^{\theta^+_1y_h}+\nu^+_2e^{\theta^+_2y_h}, & (\mathcal{C}^0\text{ at }y_h)\\
\Dq^+(x^+_\ell)+\nu^+_1e^{\theta^+_1x^+_\ell}+\nu^+_2e^{\theta^+_2x^+_\ell}=\Dq^+(x^{+\ast}_\ell)+\nu^+_1e^{\theta^+_1x^{+\ast}_\ell}+\nu^+_2e^{\theta^+_2x^{+\ast}_\ell} \!\!- \kappa_0 - \kappa_1 ( x^{+,\ast}_\ell - x^+_\ell), & (\mathcal{C}^0\text{ at }x^+_\ell)\\
\Dq^+_x(x^+_\ell)+\nu^+_1\theta^+_1e^{\theta^+_1x^+_\ell}+\nu^+_2\theta^+_2e^{\theta^+_2x^+_\ell}=\Dq^+_x(x^{+\ast}_\ell)+\nu^+_1\theta^+_1e^{\theta^+_1x^{+\ast}_\ell}+\nu^+_2\theta^+_2e^{\theta^+_2x^{+\ast}_\ell} +\kappa_1, & (\mathcal{C}^1\text{ at }x^+_\ell)\\
\Dq^+_x(x_\ell^{+\ast})+\nu^+_1\theta^+_1e^{\theta^+_1x_\ell^{+\ast}}+\nu^+_2\theta^+_2e^{\theta^+_2 x_\ell^{+\ast}} =  \kappa_1. & (\mathcal{C}^1\text{ at } x_\ell^{+\ast})
\end{cases}
\end{align}
Now, we can perform the iterations $y_h^{0} \to (x^{+,(0)}_\ell, x^{+,\ast(0)}_\ell) \to y_h^{1} \to (x^{+,(1)}_\ell, x^{+,\ast (1)}_\ell) \ldots$.

We find the following fixed--point of the best--response functions of the producer and the consumer:
\begin{align}\label{eq:eqm-ii}
\mathcal{C}^{\rm{II},\ast}_p=\begin{bmatrix}
1.9, &  3.6, & -, & +\infty \\
2.4, &  4.5, & 4.5, & 6.1 \\
\end{bmatrix}, \qquad
\mathcal{C}^{\rm{II},\ast}_c=[-\infty, 4.3],
\end{align}

The system starts in the expansion regime, and once the price reaches level $X^\ast_t =4.3$, the consumer switches to contraction and the systems remains in that state for ever. After that, she relies on the producer to impulse $(X^\ast_t)$ up/down when prices get too low/too high but never reverts to the Expansion regime. Thus, in the long-run $(X^\ast_t)$ is simply a Brownian motion with negative drift $\mu_-$ that has two impulse boundaries $x^-_\ell = 2.4, x^-_h =6.1$.

Compared to the double--switch equilibrium of the previous section, the above market equilibrium in \eqref{eq:eqm-ii} has two important differences. First, as $t \to \infty$ we have that $\mu^\ast_t \to \mu_-$ so that in the long--run the market will be in the contraction regime and the consumer becomes inactive. Second, because the producer eventually ``takes over'', she will intervene much more frequently (see center panel of Figure~\ref{fig:SamplePathT}),  benefiting himself and reducing consumer value.

\subsection{Type III -- Preemptive}
\label{ssec:eqt3}

The producer may have an interest to preempt the switch, say, from the expansion regime to the contraction regime to avoid decline in the consumption of the commodity he produces. In this case, the equilibrium is a fixed point of the best--response function of the producer described in Section~\ref{ssec:pbr-pre} and the best--response function of the consumer described in Section~\ref{ssec:cbr-one}. We look for an equilibrium where the consumer would like to switch at $y_h$ in the expansion regime, but where the producer makes $y_h$ his own intervention threshold to impulse the price down.

Using the same protocol as in Type~I and Type~II equilibrium research, we find the following threshold strategy for the producer and the consumer:
\begin{align}\label{eq:typeIII}
\mathcal{C}^{\rm{III},\ast}_p:=\begin{bmatrix}
1.7, &  3.1, & 3.1, & 4.3 \\
-, &  -, & -, & - \\
\end{bmatrix}, \qquad
\mathcal{C}^{\rm{III},\ast}_c:=[-, 4.3].
\end{align}

In the preemptive equilibrium, the price fluctuates in a narrower range  than the other two equilibria. Here, $(X^\ast_t)$ oscillates between $x^+_\ell = 1.7$ and $x^+_h = y_h =4.3$.

\subsection{Equilibrium Non-Uniqueness}

There are at least three potential equilibria. A natural question is thus whether one of them is preferable to the others. Figure~\ref{fig:eq123} shows the value functions of the producer and of the consumer in the two market regimes (expansion and contraction) and for the different equilibria from type~I to type~III. We observe that the producer would prefer in both regimes to live in a type~I equilibrium. The function $v_1^{\pm}$ dominates all the other ones (note that there is no $v_3^-$ because in equilibrium type~III contraction never happens).  However, the consumer would rather be in the preemptive equilibrium type~III: her value function $w^+_3$ dominates the other two. Intuitively, we may think that the switching costs she  saves by letting the producer do all the work of maintaining the price around its long term average value compensate for the inconvenience of having prices that are higher than preferred.

\begin{figure}[t!]
	\centering
		\includegraphics[width=0.85\textwidth]{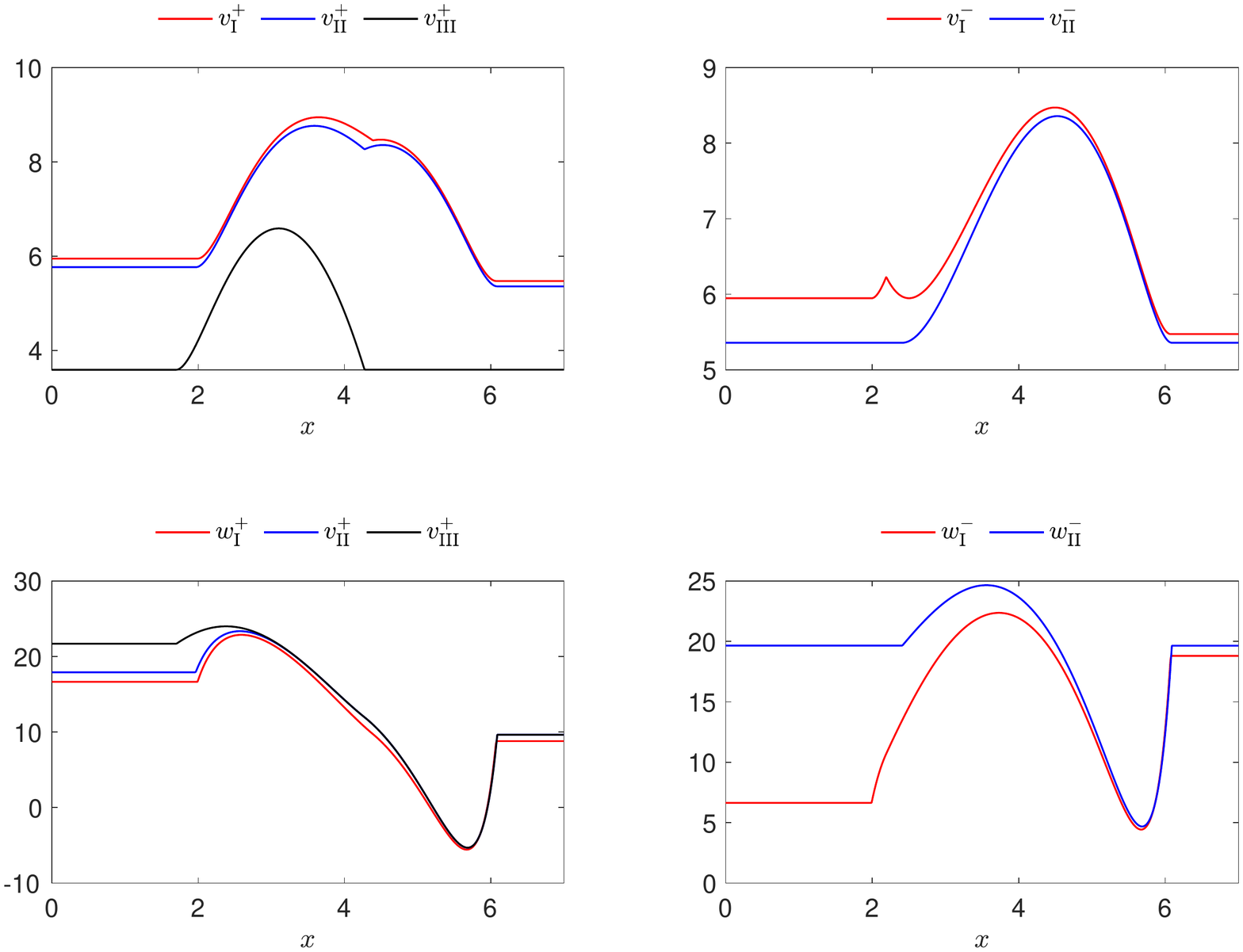}
	\caption{Producer and consumer value functions in the different equilibria.}
	\label{fig:eq123} 
\end{figure}

\subsection{Impact of Market Volatility}
As one example of comparative statistics that are possible in our model, we investigate the impact of volatility parameter $\sigma$ of $X$ on the equilibrium profits and behavior of $X^\ast$. In Table \ref{tab:long-runX} we list statistics of $\phi^\ast$ for a range of market volatilities $\sigma$. We also quantify the profitability of the two players through their average percentage of optimality (APOO)
\begin{align*}
\mathrm{APOO} := \frac{\int_{\mathcal{D}} \pi_r (x)\phi^\ast( dx )}{\pi_r (\bar{X}_r)},
\end{align*}
which is the ratio between average profit rate $\pi_r(X^\ast_t)$ in equilibrium and the maximum profit that could be hypothetically obtained at the first--best level $\pi_r(\bar{X}_r)$, $r \in \{c, p\}$.

In all types of equilibria, both players are worse off in terms of expected profit rate as $\sigma$ increase. This occurs even though in type I and in type II equilibria the average price $\E[ X^*]$ increases. However, that gain is dominated by the losses due to higher $\mathrm{Var}(X^\ast)$ which implies that prices tend to be further from their preferred levels $\bar{X}_r$ decreasing $\mathbb{E}_{\phi^\ast}\left[ \pi_r \right]$.

\begin{table}[thb]
	\centering
	\begin{tabular}{l | l | rrrrc} \hline\hline
&$\sigma$ & $\mathbb{E}_{\phi^\ast}[X^\ast] $ & $\mathrm{Var}_{\phi^\ast}(X^\ast)$ & $\mathbb{E}_{\phi^\ast}\big[\pi_p\big]$  & $\mathbb{E}_{\phi^\ast}\big[\pi_c\big]$ & Switch (per yr)\\
 \hline Type~I  &0.25 & 3.52 &  0.73 &  0.81 (80\%) &  2.4 (80\%) & 0.021\\ 
 		      &0.3 & 3.62 &  0.80 &  0.80 (79\%) &  2.2 (73\%) & 0.021\\ 
 		      &0.4 & 3.77 &  0.94 &  0.76 (75\%) &  1.90 (62\%) & 0.020  \\
 \hline Type~II  &0.25 & 3.73 &  0.68 &  0.87 (85\%) &  2.3 (74\%) & 0.020\\ 
 		       &0.3 & 3.76 &  0.74 &  0.85 (84\%) &  2.2 (71\%) & 0.020\\ 
 		       &0.4 & 3.81 &  0.85 &  0.81 (80\%) &  1.95 (64\%) & 0.020  \\
 \hline Type~III  &0.25 & 3.41 &  0.45 &  0.86 (85\%) &  2.7 (90\%) & 0.0\\ 
 		       &0.3 & 3.35 &  0.51 &  0.83 (82\%) &  2.7 (90\%) & 0.0\\ 
 		       &0.4 & 3.28 &  0.61 &  0.78 (77\%) &  2.7 (88\%) & 0.0  \\
\hline\hline
	\end{tabular}
\caption{Long--run mean and variance of $X^\ast$, long--run profit rates, and frequency of regime switches as market volatility $\sigma$ changes (APOO in parentheses).}
\label{tab:long-runX}
\end{table}

 \subsection{Effect of consumer's switching cost}

A key parameter that controls which equilibrium type we face is the consumer's switching cost $h_0$. Starting from the double--switch situation,
as $h_0$ increases ($>0.6$), the consumer is less incentivised to switch from $\mu^+$ to $\mu^-$ and we enter the single--switch scenario of Section~\ref{ssec:cbr-one}. Consequently, she receives the No--Switch payoff $\w_0^+(x)$ when $\mu_t=\mu^+$ and solving for her best--response boils down to solve for $y_h$ only. Once $h_0$ gets very large, her best--response is simply the No--Switch response $\w_0^{\pm}$. Conversely, as $h_0 \downarrow 0$ her actions become free. In that situation, we can reduce the producer problem to a single, piecewise VI with a free boundary $\tilde{X}_c$:
\begin{align}
\sup \Big\{ - \beta v(x) + \mu_{-} v_x + \frac{1}{2} \sigma^2 v_{xx}  + \pi_p (x) \; ; \;
                                            \sup_{\xi} \big\{ v(x+\xi) - v(x) - K_p (\xi) \big\} \Big\} & = 0 \qquad x > \tilde{X}_c, \\
\sup \Big\{ - \beta v(x) + \mu_{+} v_x + \frac{1}{2} \sigma^2 v_{xx}  + \pi_p (x)\; ; \;
                                            \sup_{\xi} \big\{ v(x+\xi) - v(x) - K_p (\xi) \big\} \Big\} & = 0 \qquad x < \tilde{X}_c,
\end{align}
with the ${\cal C}^0$ regularity at $\tilde{X}_c$: $\lim_{x \uparrow \tilde{X}_c}  v(x) = \lim_{x \downarrow \tilde{X}_c} v(x)$.

\medskip

Fig.~\ref{fig:h0} shows that for low $h_0$, $x^{+\ast}_h < y_\ell$ and $x^+_h$ is greater but close to $y_h$. Thus, when consumption switches from expansion to contraction, it is very likely that the price touches $x_h^+$ soon thereafter and is impulsed back to $x ^{+\ast}_h$ and thus, the regime rapidly switches back to expansion again. When the switching cost increases, this solution disappears.

\begin{figure}[bht!]
	\centering
	\includegraphics[width=0.35\textwidth]{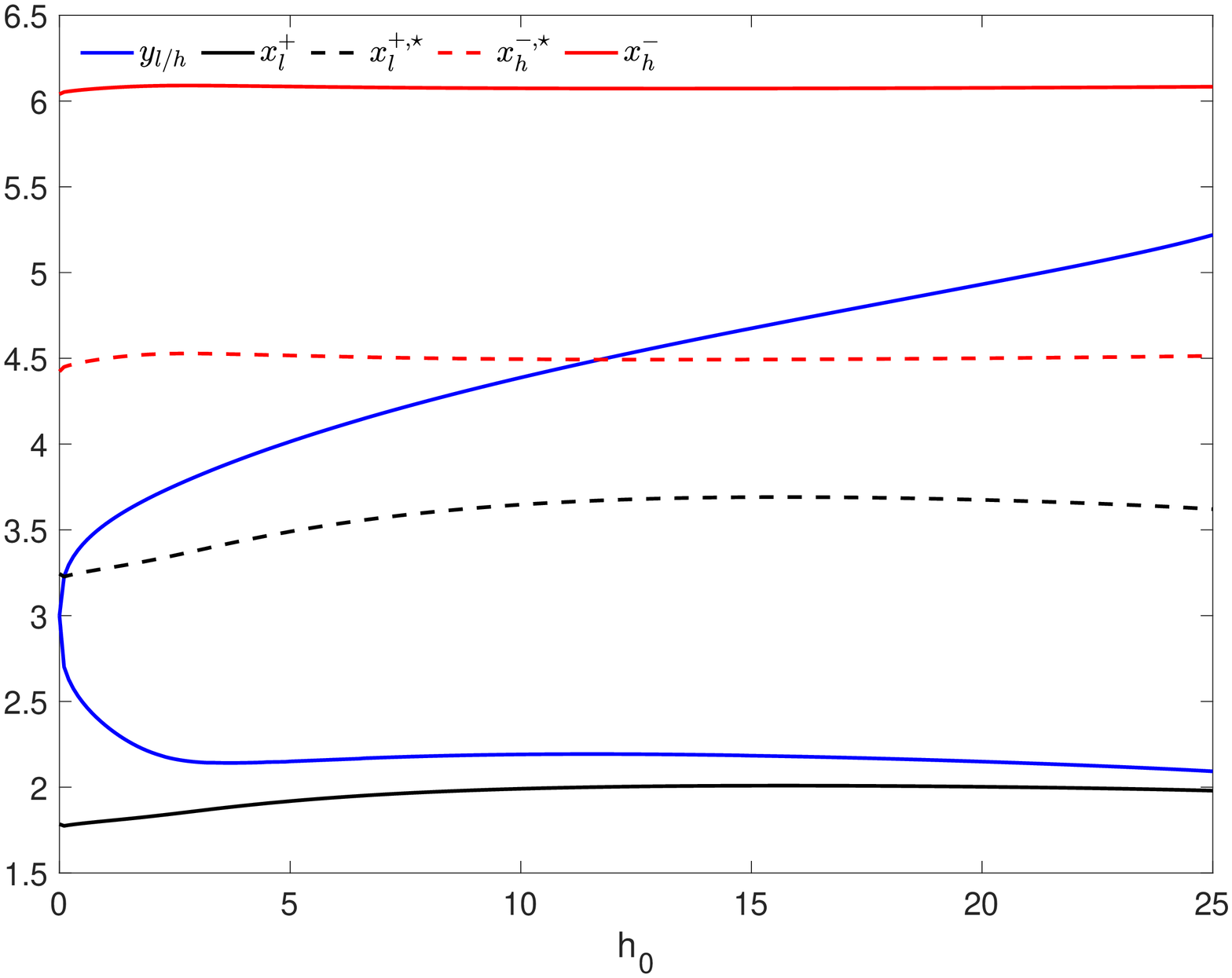}
	\caption{Equilibrium thresholds as a function of consumer switching cost $h_0$ given Table~\ref{tab:params-2} parameter values. We show the consumer thresholds $y_{\ell}, y_h$, the producer thresholds $x_\ell^+, x_h^-$ and respective impulse target levels $x_\ell^{\ast,+}, x_h^{\ast,-}$.}
	\label{fig:h0}
\end{figure}

\begin{remark}
{\rm It is possible for the impulse amounts to be so large as to lead to a double simultaneous control: producer's impulse instantaneously followed by the consumer switching. In this setting, the producer effectively forces the consumer to switch the regime by impulsing $X^\ast$ hard enough. This situation corresponds to $x^{-\ast}_h < y_\ell$, so that the impulse in the contraction regime moves $X^\ast$ into the respective switching region $(-\infty, y_\ell)$, and as a result the consumer immediately switches to the expansion regime. This situation occurs if, for instance, the drifts are  $\mu^-  = 0.01,\mu^+ = 0.1$, so that the consumer is not able to ever efficiently lower prices. Consequently, the producer is forced to fully control price reduction. We observe in the above situation the equilibrium thresholds of $x^{-\ast}_h = 3.04 < y_\ell = 3.69$. $\Box$}
\end{remark}

\section{Case study: diversification effect of vertical integration}\label{sec:oil-study}

The industrial organization of upstream and downstream segments is linked to anti-trust regulations. From a regulatory perspective, vertical integration could be used to increase market power and foreclose competitors. For example, see De Fontenay and Gans (2005) \cite{DG05} who develop a game theoretic model for the foreclosure effect of vertical concentration and Hasting and Gilbert (2005)~\cite{HG05} for the related empirical facts in the context of the US retail gasoline market. At the same time, consumers can benefit  from vertical integration of commodity producers; we refer to related analysis of electricity markets from a market equilibrium point of view (A\"id et al. (2011) \cite{ACPT11}) and an empirical point of view (Mansur (2007) \cite{M07}). Another virtue of vertical integration is its potential to reduce the long-term exposure of the firm to commodity price fluctuations. See Helfat and Teece (1987) \cite{HT87} for an empirical estimation of the hedge procured by vertical integration in the oil business.

In this section, we accordingly study whether or not downstream or upstream firms have an interest in being vertically integrated. To this end we consider a small firm that has no market power regarding the commodity price $X$ and focus on the case of the market equilibrium type~I (generic case). We then investigate whether the firm can benefit from a diversification effect by having activity both in the downstream consumer side and the upstream conversion side.

To make the case study concrete, we consider a simplified version of the crude oil and gasoline markets, the latter a shorthand for refined products, calibrated to the ballpark of the 2019 state of the world. 
Currently, world oil consumption is about 100~Mb/d (millions of barrels per day), normalized to 1 "barrel" per day. We take as a nominal initial price $X_0$ $=$ $50$~USD/b and a nominal volatility of crude $\sigma$ $=$ 10~USD/b. To calibrate our model, we consider that crude oil producers have a preferred range of prices that goes from $x^1_p = 30$~USD/b to $x^2_p =  100$~USD/b and that the average cost of oil extraction is $c_p = 30$~USD/b. This leads to a demand function $D_p(x) = 1 - 0.01 x$, which captures the low sensitivity of the demand of crude to prices. The crude is transformed into gasoline with a small amount of losses 5\%, so that the conversion factor is $\alpha =0.95$.

We set the transfer function of crude oil price to average price of gasoline to  $P(x) = 10 + 1.1x$, where $P(x)$ is also expressed in USD/b. There is evidence that the (pre-tax) price of gasoline is a linear function of the crude. For instance,  using monthly data of the Energy Information Agency of the US Department of Energy on refined products prices from January 1983 to November 2019~\footnote{Data available at \url{http://www.eia.gov/dnav/pet/pet_pri_refoth_dcu_nus_m.htm}.}, we regressed the US Total gasoline Retail sales by refineries $\hat P$ to the monthly crude oil price $\hat X$ and found a linear relation $$\hat P_m = 1.2 \hat X_m + 14 + \epsilon_m$$ with a regression $R^2 = 95\%$. Considering that the basket of refined products includes not just gasoline (even if it accounts for the largest share), we simplified the relation. Note that the condition  $p_1 \ge \alpha$ for having a downstream convex profit function holds.  Furthermore, refinery costs $c_c$ are highly variable between 4 to 10 USD/b. We take the higher value of $c_c=10$. Finally, we consider that the demand function for refined products, $D_c (P)=d'_0 - d'_1 P$ is such that $d'_0=5$~b/d of crude equivalent refined products and $d'_1=0.05$. With these parameters, the preferred range of crude prices for the consumer is between $x^1_c = 11$ and $x^2_c = 82$~USD/b. We consider fixed action costs both for the production firm and the downstream firm. We consider that the producer and the downstream firm lose two years of profit at optimal price to change state making $\kappa_0 = 2 \pi_p(\bar X_p)$  and $h_0 = 2 \pi_c(\bar X_c)$. Finally, we take $\mu_\pm = \pm 0.15$  per year, which implies that it takes 10 years for the crude price to increase by 1.5~USD.

\begin{table}[thb!]
        \centering
        \begin{tabular}{l r l l}
        \hline
        		& Value & Interpretation & Units\\ \hline\hline
        $\beta$ & $0.1$   	& Discount rate & \%/year \\
        $X_0$ & $50$   	& Initial oil price & USD/b \\
        $d_0$ & $1$ 		& Demand function for oil:  intercept & Mb/d \\
        $d_1$ & $0.01$     	& Demand function for oil: slope & Mb/d/(USD/b) \\
        $d'_0$ & $5$ 		& Demand function for gasoline: intercept & Mb/d \\
        $d'_1$ & $0.05$    	& Demand function for gasoline: slope & Mb/d/(USD/b) \\
        $\alpha$ & $0.95$	& Transformation rate & dimensionless  \\
        $p_0$   & $10$       	& Crude -- gasoline price function: intercept & USD/b  \\
        $p_1$   & $1.1$       	& Crude -- gasoline transfer price function: slope & USD/b/(USD/b) \\
        $c_p$   & $30$       	& Oil production cost & USD/b  \\
        $c_c$   & $5$       	& Refining cost & USD/b  \\
        $\mu_\pm$   & $\pm \, 0.15$       & Annualized crude drift parameters & USD/b \\
        $\sigma$   & $10$       & Annualized crude volatility & USD/b \\
        $h_0$   & $2 \pi_c(\bar X_c) = 29 $       & Consumption switching cost & USD  \\
        $\kappa_0$   & $2 \pi_p(\bar X_p) = 24.5$       & Production switching cost: fixed & USD  \\
        $\kappa_1$   & $0$       & Production switching cost: proportional & USD/b  \\
	 \hline\hline
        \end{tabular}
        \caption{Nominal values for model parameters for the crude oil case study.}
        \label{tab:params}
\end{table}

The resulting equilibrium type~I producer impulse strategy $\mathcal{C}^{\rm I,\ast}_p$ and consumer's switching strategy $\mathcal{C} ^{\rm I,\ast}_c$  associated to the calibration summarized in Table~\ref{tab:params} are given by
\begin{align}
\mathcal{C}^{\rm I,\ast}_p=\begin{bmatrix}
26, &  62, & -, & +\infty \\
-\infty, &  -, & 69, & 104 \\
\end{bmatrix}, \qquad
\mathcal{C}^{\rm I,\ast}_c=\begin{bmatrix}
22.5, &  87 \\
\end{bmatrix}.
\end{align}
Thus, in equilibrium, the crude price $X^\ast$ fluctuates between $22.5$ and $87$ USD/b, with potential excursions up to $104$ or down to $26$~USD/b at which point producers intervene.

\begin{figure}[thb!!]
	\centering
	\includegraphics[width=0.3\textwidth]{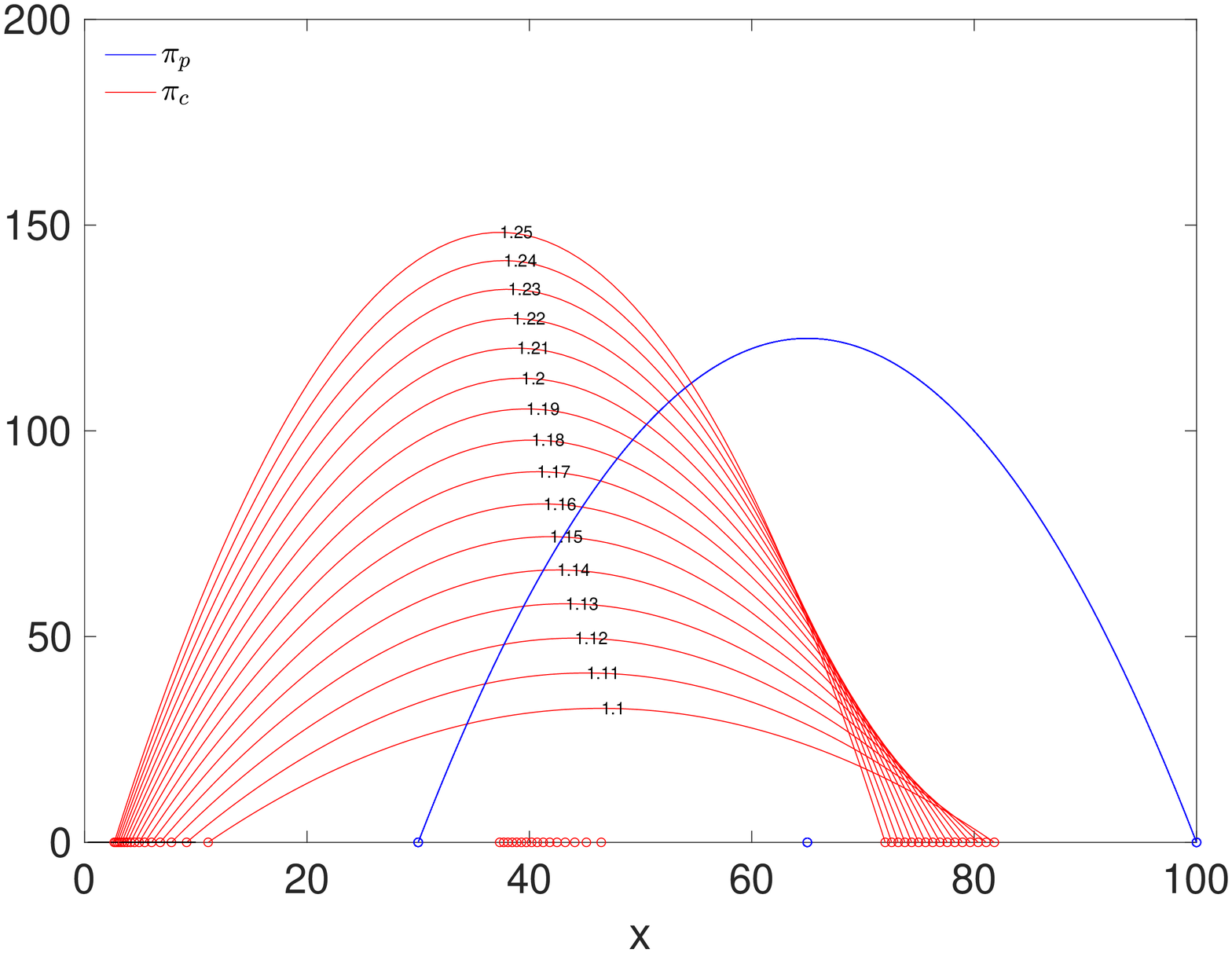}
	\includegraphics[width=0.3\textwidth]{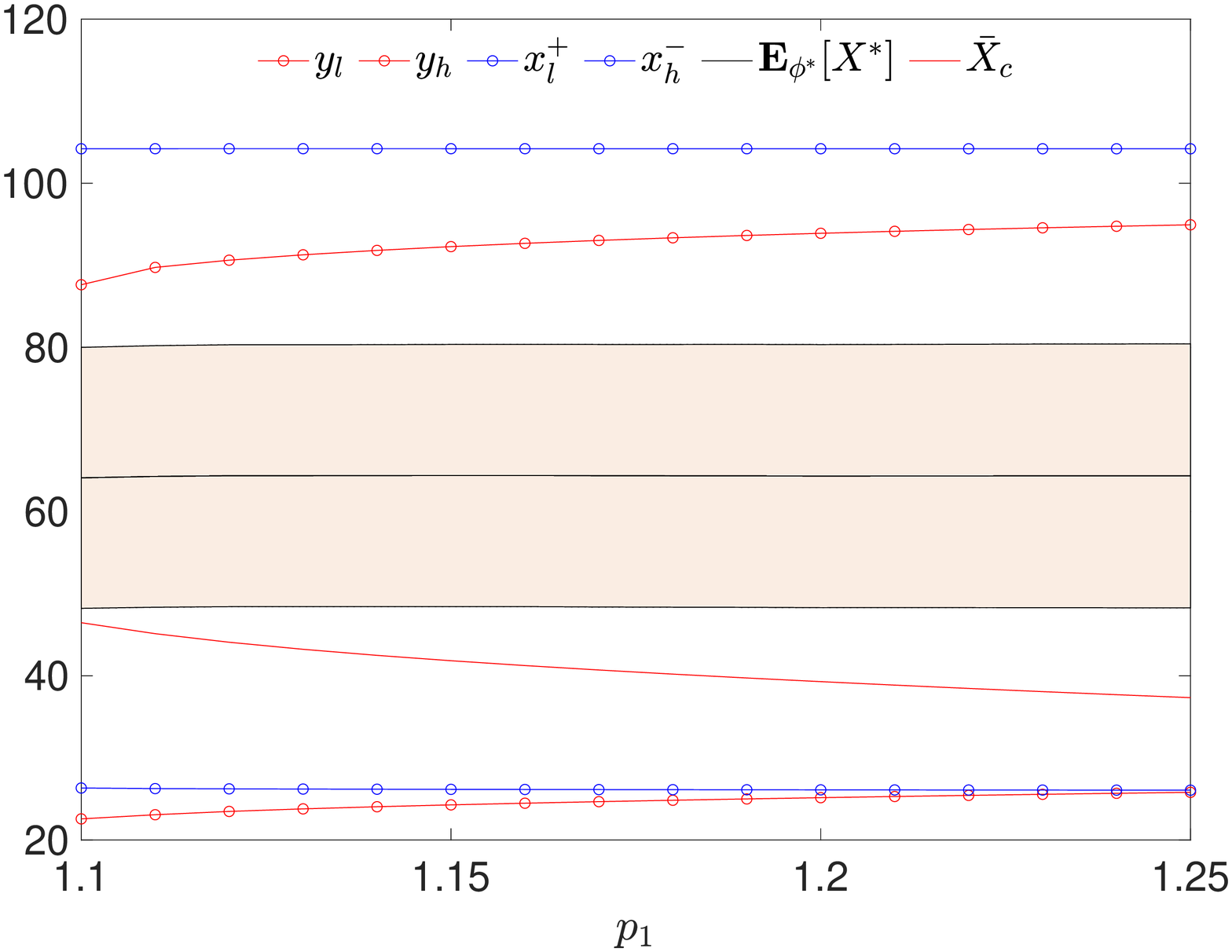}
	\includegraphics[width=0.3\textwidth]{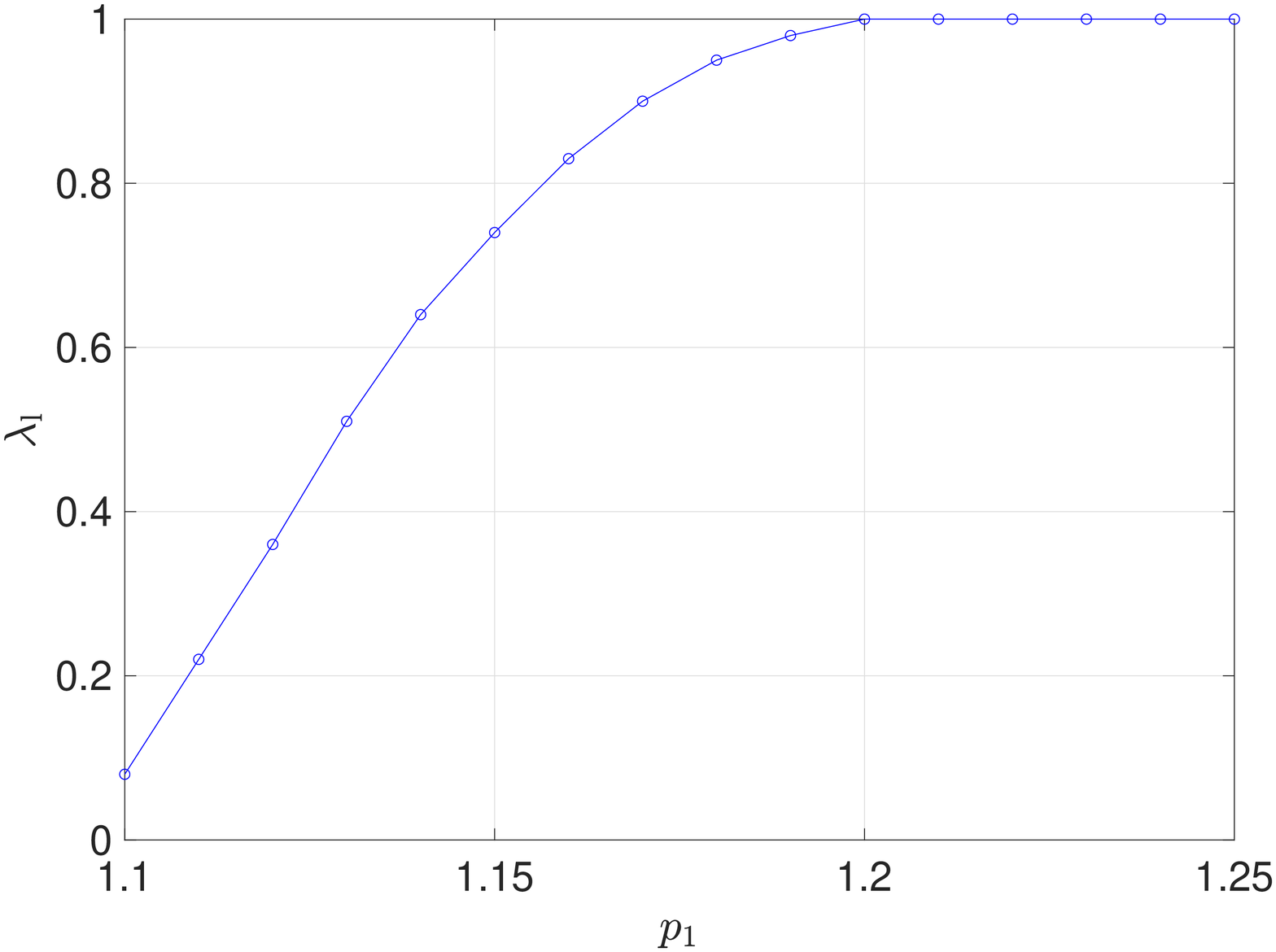}
	\caption{Left: Profit rate function of the consumer $\pi_c(x)$ (red) as the pass-through parameter $p_1$ is varied, as well as the fixed producer profit rate $\pi_p(x)$ (blue). Middle: Respective equilibrium  strategy thresholds $y_\ell, y_h$  (red) and $x_\ell^+, x_h^-$ (blue) as a function of $p_1$. We also plot $\bar{X}_c$ and $\E_{\phi^\ast}[ X^\ast]$, shading the typical commodity price range $[\E_{\phi^\ast}[ X^\ast] \pm \sigma_{\phi^\ast}( X^\ast)]$.  Right: Risk-minimizing integration level $\lambda^*$  as a function of $p_1$. }
	\label{fig:vieqt1pi}
\end{figure}
\begin{figure}[thb!]
	\centering
	\includegraphics[width=0.85\textwidth]{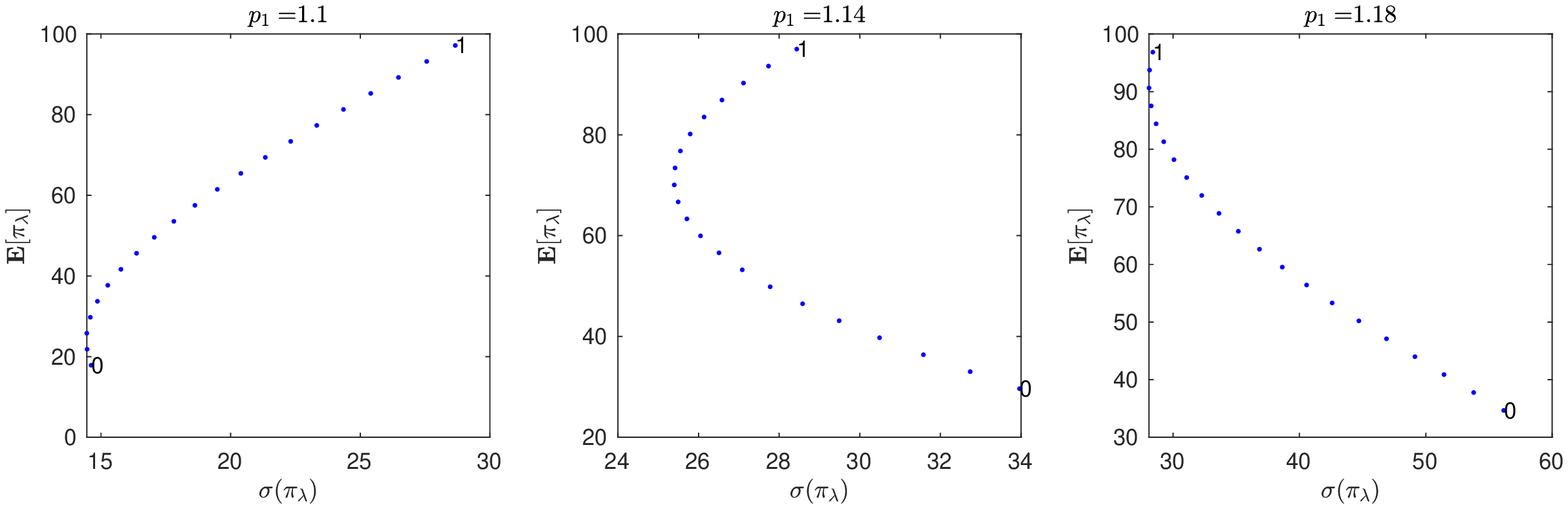}
	\caption{The curve $\lambda \mapsto (\sigma(\pi_\lambda),\E\big[ \pi_\lambda\big])$ as a function of the pass-through parameter $p_1$. }
	\label{fig:vieqt1p1}
\end{figure}

Now let us consider a small firm engaging in a fraction $\lambda \in (0,1)$ of activity in the downstream sector and $1-\lambda$ in the upstream sector. Her profit rate is thus $\pi_\lambda := \lambda \pi_c + \big( 1-\lambda \big) \pi_p$. The firm is vertically integrated when $0 < \lambda <1$.  Denote by $\sigma(\pi_\lambda)$ the standard deviation of her profit rate $\pi_\lambda(\cdot)$ integrated against the stationary distribution $\phi^\ast$ of $X^\ast$, and by $\E\big[\pi_\lambda\big] = \int \pi_\lambda(x) \phi^\ast(dx)$ the respective expected profit rate. To fix ideas and because the analysis is symmetric, we are interested in situations where a pure downstream firm ($\lambda=0$) would be better off having part of her activity in the upstream sector. This will take place when the upstream activity provides a higher expected profit rate and/or a lower risk as measured by $\sigma(\pi_\lambda)$. Figure~\ref{fig:vieqt1p1} presents the risk--return curves $\lambda \mapsto (\sigma(\pi_\lambda),\E\big[\pi_\lambda\big])$ as the pass-through parameter $p_1$ increases from the nominal value of $1.1$ to $1.18$. We observe that for low values of $p_1$ diversification gains are limited: expected profit rate goes up but risk also increases. For moderate $p_1$ a pure downstream firm unambiguously benefits from some upstream activity: she can achieve the same level of risk with a higher expected profit. For high $p_1$ the upstream sector dominates completely with lower risk and higher average profit. Figure~\ref{fig:vieqt1pi} (Right) shows the critical integration level $\lambda^*$ that minimizes the risk $\sigma(\pi_\lambda)$ and captures the ``variance--minimal'' business model.

We observe that for high enough pass-through values, being a producer $(\lambda=1)$ dominates any other combination of activity. This phenomenon happens even though the maximum profit rate of the downstream firm $\pi_c(\bar X_c)$ increases and gets higher than the producer's maximum profit rate function $\pi_p(\bar X_p)$ as shown in the left panel of Figure~\ref{fig:vieqt1pi}. As shown by the evolution of equilibrium price range in Figure~\ref{fig:vieqt1pi} (Middle), as $p_1$ increases, the equilibrium is getting more and more detrimental to the downstream firm. The shaded salmon area represents the interval $[ \E_{\phi^\ast}[X^\ast] - \sigma_{\phi^\ast}(X^\ast), \E_{\phi^\ast}[X^\ast] + \sigma_{\phi^\ast}(X^\ast)]$ where commodity prices tend to reside. The average commodity price remains stable around $65$~USD/b, and its the standard deviation is not affected much by $p_1$ either, while $\bar X_c$ is steadily decreasing. Thus, since the expected profit rate of the integrated firm is a function of the expected price and its standard deviation, it does not change much. But its variance grows as a function of $p_1$ and thus increases significantly. To conclude, in our model we do observe a diversification effect obtained by mixing upstream and downstream activities, however the integration gains  depend closely on the pass-through parameter $p_1$ which serves as a  transmission channel of the volatility of the commodity price to the retail price.

\section{Conclusion}\label{sec:conclude}

We showed how a simple model of competition between upstream and downstream representative firms having different pace of intervention can lead to a rich variety of equilibria, potentially non--unique. The fact that the upstream firm can impact the price more rapidly than the downstream firm gives the producer a significant advantage, enabling him to lock the consumer in the producer's preferred range of prices. Further, in the case of the crude oil market and its refinery products, we stressed how the pass-through parameter $p_1$ plays a key role for the diversification effect induced by vertical integration. Vertical integration is beneficial for low values of $p_1$ while for higher values, production dominates downstream activity both in terms of expected profit rate and profit standard deviation.

\section{Proofs}\label{sec:proofs}

\subsection{Proof of Proposition~\ref{prop:cons-inactive}}
\begin{proof}
The proof is standard, nonetheless we give some detail for the reader's convenience. To ease the notation, let us consider only the case $\mu = \mu_+$, the other case being identical. Let $w^+ _0 (x) = \Dd^{+}(x)+u^{+}(x)$, where the parameters $ (\lambda^+ _{1,0}, \lambda^+ _{2,0}) \in \mathbb R$ solve the system \eqref{inaction-eq1}-\eqref{inaction-eq2}. By construction the function $w_0 ^+$ is of class $\mathcal C^2$ everywhere. Hence we can apply It\^o's formula to $e^{-\beta s} w_0 ^+ (X_s)$ on the time interval $[0, t \wedge \zeta_n)$, yielding
\begin{align*} e^{-\beta (t\wedge \zeta_n) } w_0 ^+ (X_{t\wedge \zeta_n}) & = w_0 ^+(x) + \int_{0+} ^{t \wedge \zeta_n} e^{-\beta s} \left[ w_0 ^{+ \prime}(X_{s-}) (\mu_+ ds + \sigma dW_s -dN_s) -\beta w_0 ^+ (X_{s-})ds \right] \\
& \quad + \frac{\sigma^2}{2} \int_{0+} ^{t \wedge \zeta_n} e^{-\beta s} w_0 ^{+ \prime \prime} (X_{s-}) ds + \sum_{0<s\le t \wedge \zeta_n} e^{-\beta s} \left[ \Delta w_0 ^+ (X_s) + w_{0} ^{+ \prime} (X_{s-})\Delta N_s \right], \end{align*}
where $(\zeta_n)_{n \ge 1}$ is a localizing sequence of stopping times along which the local martingale part above is in fact a true martingale. We use the notation $^\prime$ and $^{\prime \prime}$ for, respectively, the first and second derivative in $x$.
Taking expectations on both sides, using the fact that $w_0 ^+$ solves the ordinary differential equation \eqref{inaction-PDE} and letting $n \to \infty$, we obtain
\[ \mathbb E\left [ e^{-\beta t } w_0 ^+ (X_t) \right] = w_0 ^+ (x) - \mathbb E \left[\int_{0+} ^t e^{-\beta s} \pi_c (X_s) ds + \sum_{0<s\le t} \Delta w_0 ^+ (X_s)\right] . \]
Now, notice that on the jumps of $X$ we have $\Delta w_0 ^+ (X_s) = w_0 ^+ (x^{+ \ast} _r) - w_0 ^+ (x^{+} _r)$, which is zero by the boundary conditions \eqref{eq:bd_noswitch}, hence the jump part in the equation above vanishes. Moreover, being $X_t \in [x_\ell ^+ , x_h ^+]$ for all $ t \ge 0$, we have by dominated convergence that $\mathbb E\left [ e^{-\beta t } w_0 ^+ (X_t) \right] \to 0$ as $ t \to \infty$. Therefore, letting $t \to \infty$ we can conclude that $w_0 ^+ (x) = J_c ^+ (x; N, \mu^+)$ for all $x \in [x_\ell ^+ , x_h ^+]$.
\end{proof}

Fig.~\ref{fig:BR_NoSwiI} illustrates the fact that a threshold switching strategy might not be optimal in all potential situations by considering the shape of $w^\pm_0(x)$. In the right panel, we have comonotonicity between $w^+$ and $w^-$: the consumer is incentivised to switch to $\mu^+$ when $X_t$ is low and to $\mu^-$ when $X_t$ is high. In that situation, we expect that a threshold--type strategy is a best response. In contrast, on the left panel two other cases are illustrated. First, we see that it is possible that $w^+(\cdot) \ll w^-(\cdot)$, in other words the consumer has a strong preference to one regime over the other. In that case, the expansion regime could be absorbing, i.e.~it is \emph{optimal} to never switch to $\mu_-$. In the plot this would happen if $h_0$ is low (dashed line), whereby $w^-(x) > w^+(x) - h_0$ and it is optimal to switch to $\mu_-$ at any $x$ (therefore $\mu_+$ would never be observed in the resulting game evolution). At the same time, we see that if $h_0$ is moderate (the solid line), then the region where $w^-_0(x) > w^+_0(x) - h_0$ is \emph{disconnected}, so it is likely that a two--threshold switching strategy is an optimal response.

\begin{figure}[ht]
	\centering
	\subfigure[][$\mathcal{C}_p=\begin{bmatrix}
	1.5 &3.0 & 2.8 & 4.0\\
	1.2 & 2.5  & 2.6  & 4.2
	\end{bmatrix}$ ]{
		\label{fig:BR_NoSwiI} 
		\includegraphics[width=0.35\textwidth]{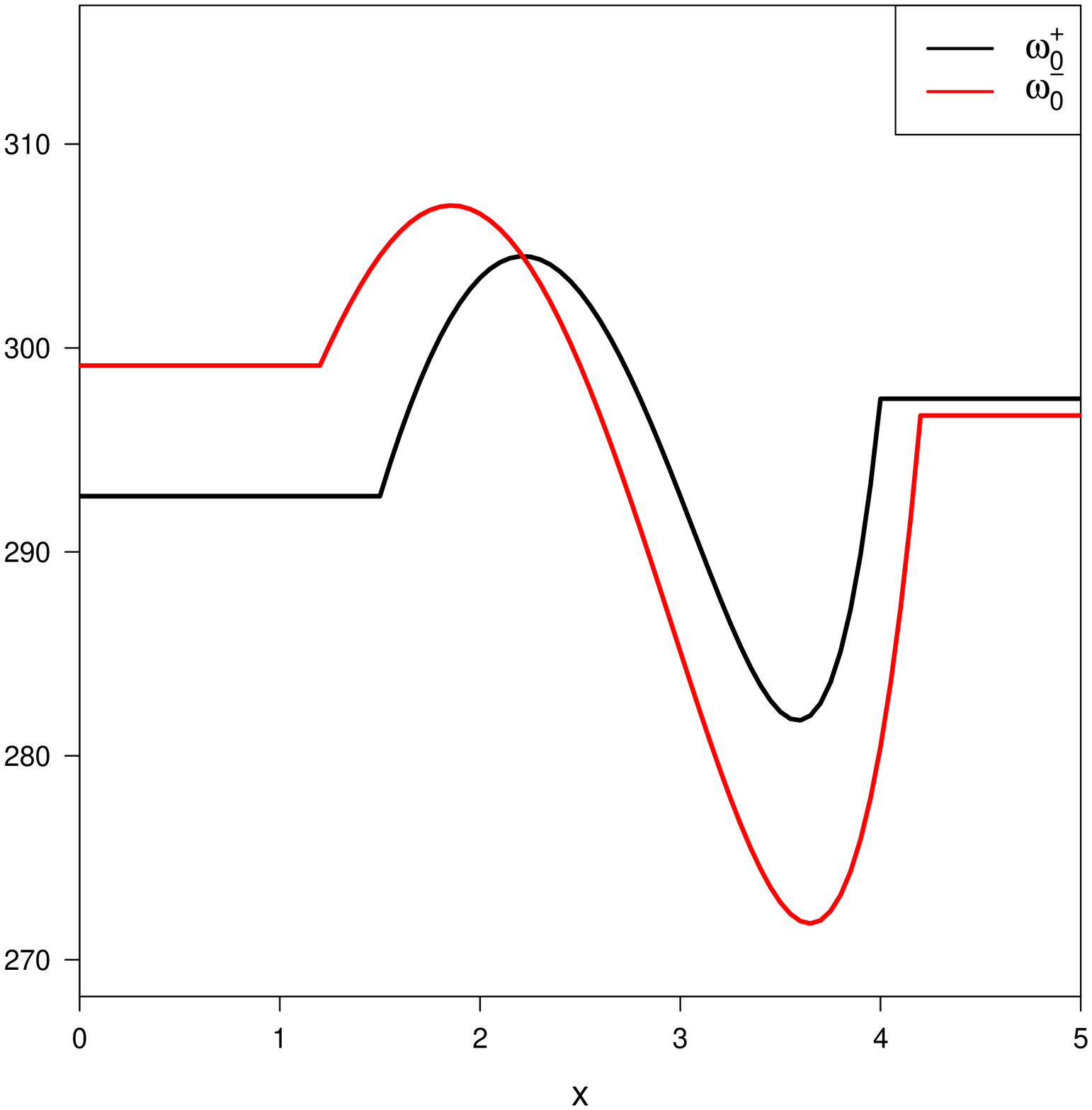}}
	\subfigure[][$\mathcal{C}_p=\begin{bmatrix}
	1.2 &3.0 & 2.8 &  3.5\\
	1.0 & 2.8 &2.4 & 3.2
	\end{bmatrix}$ ]{
		\label{fig:BR_NoSwiII}
		\includegraphics[width=0.35\textwidth]{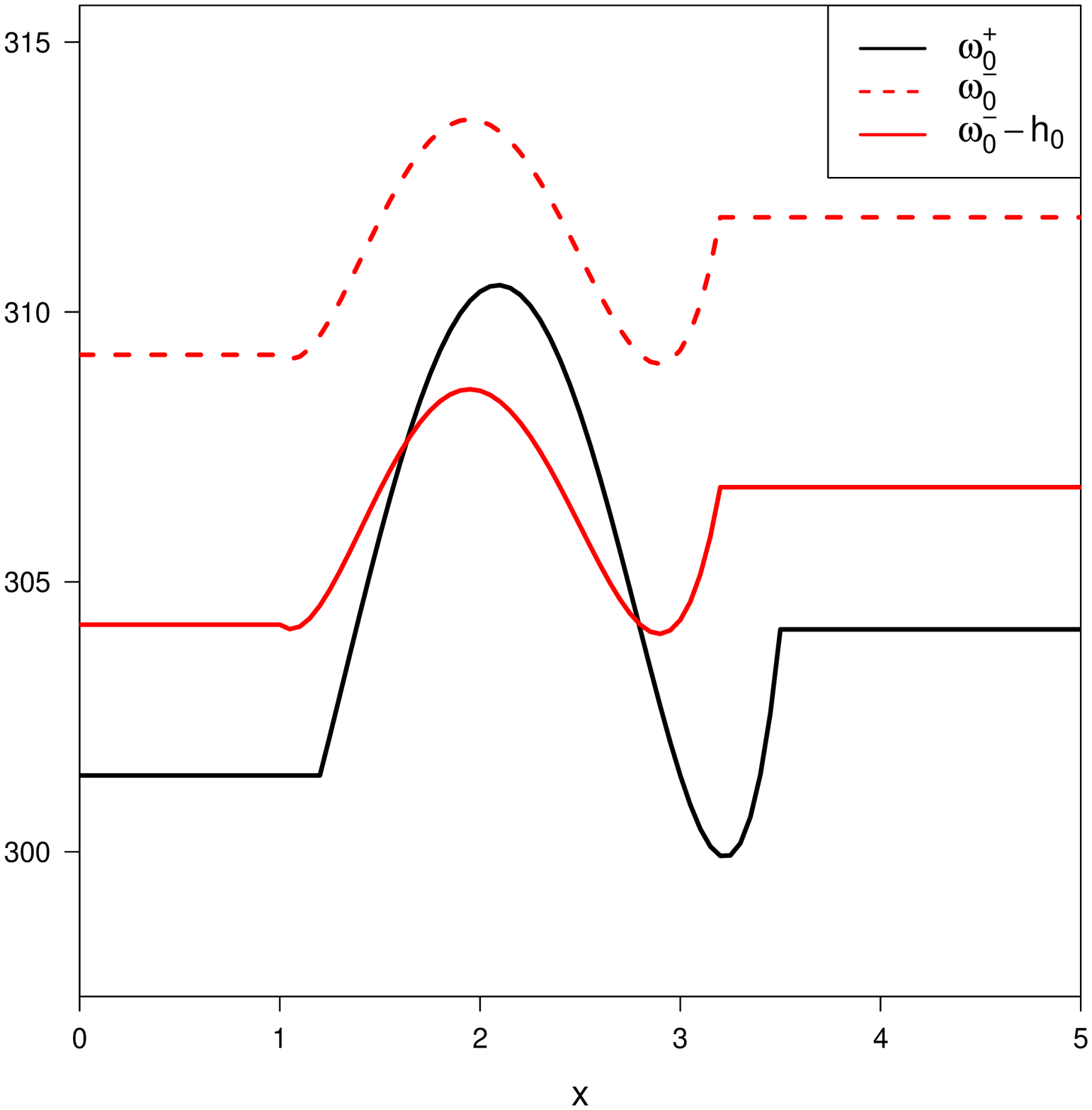}}
	\caption{No--Switch payoffs $\omega^{\pm}_0(x)$ of the consumer given the producer's strategy $\mathcal{C}_p$. }
	\label{fig:BR_NoSwi} 
\end{figure}

\subsection{Proof of Proposition~\ref{prop:cons-best-response}}
\begin{proof}[Proof of Proposition~\ref{prop:cons-best-response}]
By construction, the functions $w^\pm (x)$ in \eqref{eq:Con_DS_payoff} solve the system of VIs in \eqref{eq:QVI-cons1}-\eqref{eq:QVI-cons2} and satisfy $w^+ \in \C^2 ((x_\ell ^+ , x_h ^+) \setminus \{y_\ell\}) \cap \C^1 ((x_\ell ^+, x_h ^+)) \cap \C^0 (\mathbb R)$ and $w^- \in \C^2 ((x_\ell ^-, x_h ^-) \setminus \{y_h\}) \cap \C^1 ((x_\ell ^- , x_h ^-)) \cap \C^0 (\mathbb R)$. Let $N$ denote the pure jump component in $X$'s dynamics associated to the producer's strategy with thresholds $(x_\ell ^\pm, x_\ell ^{\pm \ast}; x_h ^\pm, x_h ^{\pm \ast})$. The proof is structured in two steps.\medskip

\noindent\emph{-- Step 1: optimality.} The following verification argument proves that such functions coincide with the best--response payoffs of the consumer and that the switching times $\hat \sigma_i$ as in the statement are optimal provided they are admissible. First, by an approximation procedure as in the first part of the proof in \cite[Theorem 3.3]{ABCCV20}, we can assume without loss of generality that $w^+ \in \C^2 ((x_\ell ^+, x_h ^+)) \cap \C^0 (\mathbb R)$. Let $\mu_{0-} = \mu_+$. Consider two consecutive switching times of any consumer admissible strategy, say $\sigma_{2i}$ and $\sigma_{2i +1}$, for $i \ge 0$ with the convention $\sigma_0 =0$, and recall that over $[\sigma_{2i}, \sigma_{2i + 1})$ the state process $X$ has drift $\mu_+$.
Applying It\^o's formula to $e^{-\beta t} w^+ (X_t)$ over the interval $[\sigma_{2i} \wedge T , \sigma_{2i +1}\wedge T)$, for some finite $T >0$, we obtain
\begin{align*}
e^{-\beta (\sigma_{2i +1}\wedge T)} w^+ (X_{\sigma_{2i +1}\wedge T}) & = e^{-\beta (\sigma_{2i}\wedge T)} w^+ (X_{\sigma_{2i}\wedge T})\\
& \quad + \int_{\sigma_{2i} \wedge T} ^{\sigma_{2i +1} \wedge T} e^{-\beta s} \left\{ w^{+}_x (X_s ) dX_s + \frac{\sigma^2}{2}\w_{xx} ^{+} (X_s ) ds -\beta X_s ds \right\} \\
& \quad + \sum_{\sigma_{2i} \wedge T < u \le \sigma_{2i+1} \wedge T} e^{-\beta s}\left\{ \Delta w^+ (X_s) + w_x ^{+} (X_s) \Delta N_s\right\}.
\end{align*}
Using the dynamics $dX_s = \mu_+ ds + \sigma dW_s -dN_s$ between the two switching times above, localizing the martingale part through a suitable sequence of stopping times $\zeta_n$ and taking expectation on both sides, we obtain
\begin{align*}
\mathbb E & \left[ e^{-\beta \zeta^{2i+1} _{n,T}} w^+ (X_{\zeta^{2i+1}_{n,T}})\right]  = \mathbb E \left[ e^{-\beta \zeta^{2i} _{n,T}} w^+ (X_{\zeta^{2i} _{n,T}})\right]\\
&  + \mathbb E \left[ \int_{\zeta^{2i} _{n,T}} ^{\zeta^{2i+1} _{n,T}} e^{-\beta s} \left\{ w_x ^{+} (X_s )\mu_+ + \frac{\sigma^2}{2} w_{xx} ^{+} (X_s ) -\beta X_s \right\} ds \right]
 + \mathbb E \left[ \sum_{\zeta^{2i} _{n,T} \le s < \zeta^{2i+1} _{n,T}} e^{-\beta s} \Delta w^+ (X_s)\right],
\end{align*}
where we set $\zeta_{n,T}^{k} := \sigma_{k}\wedge \zeta_n \wedge T$.
Notice first that the third summand above vanishes since between $\sigma_{2i}$ and $\sigma_{2i +1}$, the state process $X$ can jump only due to the impulses of the producer, hence at any of such jumps the $\C^0$-pasting condition at $x^+_h$ yields
\[ \Delta w^+ (X_s ) = (w^+ (X_s ) - w^+ (X_{s-} ))\mathbf 1_{(\Delta X_s \neq 0)} = (w^+ (x_h ^\ast) -w^+ (x_h ^+))\mathbf 1_{(\Delta X_s \neq 0)} = 0.\]
Regarding the second summand, we use the variational inequality \eqref{eq:QVI-cons1} so that we can write
\begin{align*}
\mathbb E \left[ e^{-\beta \zeta^{2i+1} _{n,T}} w^+ (X_{\zeta^{2i+1} _{n,T}})\right] \le \mathbb E \left[ e^{-\beta \zeta^{2i} _{n,T}} w^+ (X_{\zeta^{2i} _{n,T}})\right] - \mathbb E \left[ \int_{\zeta^{2i} _{n,T}} ^{\zeta^{2i+1} _{n,T}} e^{-\beta s} \pi_c (X_s ) ds \right].
\end{align*}
Now, letting $n \to \infty$ we obtain by dominated convergence that
\[ \mathbb E  \left[ e^{-\beta \sigma_{2i+1} \wedge T} w^+ (X_{\sigma_{2i+1}\wedge T})\right] \le \mathbb E \left[ e^{-\beta \sigma_{2i} \wedge T} w^+ (X_{\sigma_{2i}\wedge T})\right] - \mathbb E \left[ \int_{\sigma_{2i}} ^{\sigma_{2i+1}} e^{-\beta s} \pi_c (X_{s \wedge T} ) ds \right], \quad i \ge 0.\]
Analogously, we can get the same inequality between the switching times $\sigma_{2i -1}$ and $\sigma_{2i}$ for $i \ge 1$ with $w^-$ replacing $w^+$, so summing them all up we have
\begin{align*}  - \mathbb E \left[ \int_0 ^{(\sup_i \sigma_i) \wedge T} \pi_c (X_s) ds\right] & \ge \sum_{i \ge 0} \mathbb E \left[ e^{-\beta \sigma_{2i+1} \wedge T} w^+ (X_{\sigma_{2i+1}\wedge T}) - e^{-\beta \sigma_{2i} \wedge T} w^+ (X_{\sigma_{2i}\wedge T}) \right] \\
& \quad + \sum_{i \ge 1} \mathbb E \left[ e^{-\beta \sigma_{2i} \wedge T} w^- (X_{\sigma_{2i}\wedge T}) - e^{-\beta \sigma_{2i-1} \wedge T} w^- (X_{\sigma_{2i-1}\wedge T}) \right] .
\end{align*}
Note that by admissibility $\sum_{i \ge 1} e^{-\beta \sigma_i} \in L^2 (\mathbb P)$, which implies  $\sup_{i \ge 1} \sigma_i = +\infty$ almost surely. Then, using the $\C^0$-pasting conditions in \eqref{smooth-pasting} and letting $T \to \infty$, we finally obtain
\[ \mathbb E \left[ \int_0 ^{+\infty} \pi_c (X_s) ds\right] + \sum_{i \ge 1} \mathbb E \left[ e^{-\beta \sigma_i} h_0 \right] \le w^+ (x) ,\]
for any admissible consumer strategy $(\sigma_i)$. Applying the same arguments to the sequence $(\hat \sigma_i)$ we would get equalities instead of inequalities everywhere. The proof for the case $\mu_{0-}=\mu_-$ is analogous and therefore is omitted.\medskip

\noindent\emph{-- Step 2: admissibility.} To conclude we show that the switching times $\hat \sigma_i$ are admissible, i.e.~they belong to the set $\mathcal A_c$. To do so, notice first that $(\hat \sigma_i)$ is a sequence of $[0,\infty)$-valued stopping times. Hence, it remains to show that a.s.~$\hat \sigma_i < \hat \sigma_{i+1}$ for all $i\ge 0$, and $\sum_{i \ge 1} e^{-\beta \hat \sigma_i} \in L^2 (\mathbb P)$. The former follows from $y_\ell < y_h$. For the latter, we can proceed as in the proof of \cite[Prop. 4.7]{ABCCV20}, whose main idea is to write each $\sigma_i$ as a sum of independent exit times for some (scaled) Brownian motion with possibly different drifts and initial conditions. First, let us denote $(\tau^\prime _k)_{k \ge 1}$ the increasing
sequence of stopping times exhausting the intervention times of both players. Therefore we have
\begin{align*}
\mathbb E\left[ \left( \sum_{i \ge 1} e^{-\beta \sigma_i}\right)^2\right] & \le \mathbb E\left[ \left( \sum_{k \ge 1} e^{-\beta \tau^\prime _k}\right)^2\right] \le \lim_{m \to\infty} 2 \mathbb E\left[ \sum_{1\le k \le  r \le m} e^{- \beta (\tau^\prime _r + \tau^\prime _k)}\right] \\
& \le \lim_{m \to\infty} 2 \mathbb E\left[ \sum_{1\le k \le m} e^{- 2\beta \tau^\prime _k }\right] = 2 \mathbb E\left[ \sum_{k \ge 1} e^{- 2\beta \tau^\prime _k }\right],
\end{align*}
hence it suffices to prove that $\sum_{k \ge 1} e^{- 2\beta \tau^\prime _k } \in L^1 (\mathbb P)$. Now, notice that $\tau^\prime _k$, $k \ge 1$, can be represented as  $\sum_{r=1}^k \zeta_r$, where $\zeta_r$ is a sequence of independent random variables distributed as the exit time, say $\zeta^{z,\mu}$, of one of the processes $z + \mu t + \sigma W_t$ with
\[ (z, \mu) \in \mathcal Z^\pm :=  \left\{ (y_h, \mu_+),(y_{\ell}, \mu_- ), (x_{h}^{+ \ast} , \mu_+), ( x_\ell ^{- \ast} , \mu_-) \right\} ,  \]
from the respective intervals
\[ (-\infty, y_h ), \quad (y_\ell, +\infty), \quad (-\infty, x_h ^+), \quad (x_\ell ^- , +\infty). \]
Due to the independence of the sequence $\zeta_r$ we have
\begin{align*}
\mathbb E\left[ \sum_{k \ge 1} e^{- 2\beta \tau^\prime _k }\right]  = \sum_{k \ge 1} \prod_{r=1}^k \mathbb E\left[ e^{- 2\beta \zeta_r }\right] \le  \sum_{k \ge 1} \left(\mathbb E\left[ e^{- 2\beta \min_{(z,\mu) \in \mathcal Z^\pm}\zeta^{z,\mu} }\right]\right)^k ,
\end{align*}
which is a convergent geometric series, due to $\beta >0$ and the fact that $\zeta^{z,\mu} > 0$ almost surely for all $(z,\mu) \in \mathcal Z^\pm$. This shows that sequence of switching times $\hat \sigma_i$ is an admissible consumer's strategy and concludes the proof.
\end{proof}

\subsection{Proofs of Propositions~\ref{prop:prod-best-response1} and~\ref{prop:prod-best-response2}}
\begin{proof}[Proof of Proposition~\ref{prop:prod-best-response1}]
Let $v: \{\mu_- , \mu_+ \} \times \mathbb R \to \mathbb R$ be the function defined as $v(\mu_\pm , x) = v^\pm (x)$, with $v^\pm$ as in \eqref{eq:DS_OneSide}. By construction, the functions $(v^+, v^-)$ solve the system of VIs in \eqref{eq:QVI-prod-nonpree} and moreover $v^\pm \in \C^2 ((x_\ell ^+ , x_h ^-) \setminus \{y_\ell, y_h\}) \cap \C^0 (\mathbb R)$, hence not necessarily $\C^1$ at the points $y_\ell, y_h$. Recall that $\mu_t = \mu_+ \sum_{i=0}^\infty \mathbf 1_{\{ \sigma_{2i} \le t < \sigma_{2i+1} \} } + \mu_-  \sum_{i=1}^\infty \mathbf 1_{\{ \sigma_{2i-1} \le t < \sigma_{2i} \} }$, $t \ge 0$, where without loss of generality we can assume $\sigma_i$ is the $i$-th switching instance taken by the consumer in the case $\mu_{0-}=\mu_+$ (remember the convention $\sigma_0 =0$). The other case $\mu_{0-}=\mu_-$ can be treated in a similar way, it is therefore omitted. We split the rest of the proof in two steps.\medskip

\noindent\emph{-- Step 1: optimality.} The following verification argument proves that such functions coincide with the best--response payoffs of the producer and that the impulse strategy as in the statement is optimal provided it is admissible. First, by an approximation procedure as in the first part of the proof in \cite[Theorem 3.3]{ABCCV20}, we can assume without loss of generality that $v^\pm \in \C^2 ((x_\ell ^+, x_h ^-)) \cap \C^0 (\mathbb R)$. Consider any producer admissible strategy $(\tau_i, \xi_i)_{i \ge 1}$ as in the first part of Definition \ref{def:adm}.
Applying It\^o's formula to $e^{-\beta t} v (\mu_t, X_t)$ over the interval $[ \sigma_{2i} \wedge T , \sigma_{2i +1}\wedge T)$, for some finite $T >0$, we obtain
\begin{align*}
e^{-\beta (\sigma_{2i +1}\wedge T)} v (\mu_{\sigma_{2i +1}\wedge T}, X_{\sigma_{2i +1}\wedge T}) & = e^{-\beta (\sigma_{2i +1}\wedge T)} v^+ (X_{\sigma_{2i +1}\wedge T}) \\
& = e^{-\beta (\sigma_{2i}\wedge T)} v^+ (X_{\sigma_{2i}\wedge T})\\
& \quad + \int_{\sigma_{2i} \wedge T} ^{\sigma_{2i +1} \wedge T} e^{-\beta s} \left\{ v^{+}_x (X_s ) dX_s + \frac{\sigma^2}{2} v^+ _{xx} (X_s ) ds -\beta v^+(X_s) ds \right\} \\
& \quad + \sum_{\sigma_{2i} \wedge T < s \le \sigma_{2i+1} \wedge T} e^{-\beta s}\left\{ \Delta v^+ (X_s) + v_x ^{+} (X_s) \Delta N_s \right\},
\end{align*}
where the first equality comes from the fact that over $[ \sigma_{2i} \wedge T , \sigma_{2i +1}\wedge T)$, the drift equals $\mu_+$ (remember that $\mu_{0-}=\mu_+$). Using the dynamics $dX_s = \mu_+ ds + \sigma dW_s -dN_s$, with $N_t := \sum_{i \ge 1} \xi_i \mathbf 1_{\{ \tau_i \le t \} }$, between the two switching times above, localizing the martingale part through a suitable sequence of stopping times $\zeta_n$ and taking expectation on both sides, we obtain
\begin{align*}
\mathbb E & \left[ e^{-\beta \zeta^{2i+1} _{n,T}} v^+ (X_{\zeta^{2i+1}_{n,T}})\right]  = \mathbb E \left[ e^{-\beta \zeta^{2i} _{n,T}} v^+ (X_{\zeta^{2i} _{n,T}})\right]\\
& + \mathbb E \left[ \int_{\zeta^{2i} _{n,T}} ^{\zeta^{2i+1} _{n,T}} e^{-\beta s} \left\{ v_x ^{+} (X_s )\mu_+ + \frac{\sigma^2}{2} v_{xx} ^{+} (X_s ) -\beta v^+(X_s) \right\} ds \right]
 + \mathbb E \left[ \sum_{\zeta^{2i} _{n,T} \le s < \zeta^{2i+1} _{n,T}} e^{-\beta s} \Delta v^+ (X_s)\right],
\end{align*}
where we set $\zeta_{n,T}^{k} := \sigma_{k}\wedge \zeta_n \wedge T$.
For the third summand above, notice that between $\sigma_{2i}$ and $\sigma_{2i +1}$, due to the non-local term in the variational inequality \eqref{eq:QVI-prod-nonpree}, the state process $X$ can jump only due to the impulses of the producer and at any of such jumps we have
\begin{align*} \Delta v^+ (X_{\tau_i} ) & \le -K_p (\xi _i), \quad i \ge 0,
\end{align*}
implying
\[ \mathbb E \left[ \sum_{\zeta^{2i} _{n,T} \le s < \zeta^{2i+1} _{n,T}} e^{-\beta s} \Delta v^+ (X_s)\right] \le \mathbb E \left[ \sum_{j: \zeta^{2i} _{n,T} \le \tau_j < \zeta^{2i+1} _{n,T}} e^{-\beta s} K_p (\xi_j )\right]. \]
Regarding the second summand, we use the variational inequality \eqref{eq:QVI-prod-nonpree} so that we can write
\begin{align*}
\mathbb E \left[ e^{-\beta \zeta^{2i+1} _{n,T}} v^+ (X_{\zeta^{2i+1} _{n,T}})\right] \le & \, \mathbb E \left[ e^{-\beta \zeta^{2i} _{n,T}} v^+ (X_{\zeta^{2i} _{n,T}})\right] - \mathbb E \left[ \int_{\zeta^{2i} _{n,T}} ^{\zeta^{2i+1} _{n,T}} e^{-\beta s} \pi_p (X_s ) ds \right]\\
& + \mathbb E \left[ \sum_{j: \zeta^{2i} _{n,T} \le \tau_j < \zeta^{2i+1} _{n,T}} e^{-\beta s} K_p (\xi_j )\right].
\end{align*}
Now, due to $X_t \in [ x_\ell ^+ , x_h ^-]$ for all $t \ge 0$, letting $n \to \infty$ we obtain by dominated convergence that
\begin{align*} \mathbb E  \left[ e^{-\beta \sigma_{2i+1} \wedge T} v^+ (X_{\sigma_{2i+1}\wedge T})\right] \le \, &  \mathbb E \left[ e^{-\beta \sigma_{2i} \wedge T} v^+ (X_{\sigma_{2i}\wedge T})\right] - \mathbb E \left[ \int_{\sigma_{2i}} ^{\sigma_{2i+1}} e^{-\beta s} \pi_p (X_{s \wedge T} ) ds \right] \\
& + \mathbb E \left[ \sum_{j: \sigma_{2i} \le \tau_j < \sigma_{2i+1}} e^{-\beta s} K_p (\xi _j )\right], \quad i \ge 0.\end{align*}
Analogously, we can get the same inequality between the switching times $\sigma_{2i -1}$ and $\sigma_{2i}$ for $i \ge 1$ with $v^-$ replacing $v^+$, so summing them all up we have
\begin{align}  - \mathbb E \left[ \int_0 ^{(\sup_i \sigma_i) \wedge T} \pi_p (X_s) ds\right] & \ge \sum_{i \ge 0} \mathbb E \left[ e^{-\beta \sigma_{2i+1} \wedge T} v^+ (X_{\sigma_{2i+1}\wedge T}) - e^{-\beta \sigma_{2i} \wedge T} v^+ (X_{\sigma_{2i}\wedge T}) \right] \nonumber \\
& \quad + \sum_{i \ge 1} \mathbb E \left[ e^{-\beta \sigma_{2i} \wedge T} v^- (X_{\sigma_{2i}\wedge T}) - e^{-\beta \sigma_{2i-1} \wedge T} v^- (X_{\sigma_{2i-1}\wedge T}) \right]. \label{eq:opt-prop3}
\end{align}
Note that by admissibility $\sum_{i \ge 1} e^{-\beta \sigma_i} \in L^2 (\mathbb P)$, which implies  $\sup_{i \ge 0} \sigma_i = +\infty$ almost surely. Then, using the $\C^0$-pasting conditions in \eqref{eq:prod-sys} and letting $T \to \infty$, we finally obtain
\[ \mathbb E \left[ \int_0 ^{+\infty} \pi_p (X_s) ds\right] + \sum_{i \ge 1} \mathbb E \left[ e^{-\beta \tau_i} K_p (\xi_i) \right] \le v^+ (x) ,\]
for any admissible producer's strategy $(\tau_i , \xi_i)_{i \ge 1}$. Applying the same arguments to the impulse strategy $(\tau^* _i , \xi^* _i)_{i \ge 1}$ as in the statement we would get equalities instead of inequalities everywhere. Notice that the second order conditions \eqref{eq:SOC} guarantee the optimality of the impulses $\xi_i ^*$. 
\medskip

\noindent\emph{-- Step 2: admissibility.} To conclude the proof, we need to show that the impulse strategy $(\tau_i ^* , \xi_i ^*)_{i \ge 1}$ is admissible as in the first part of Definition \ref{def:adm}. Property 1 is granted by the dynamics of the state variable $X$ and the fact that producer's thresholds satisfy $x_\ell ^+ < x_\ell ^{+*}, x_h ^{-*} < x_h ^-$.

Property 2 is obviously satisfied by definition of the optimal impulses $\xi_i ^*$ as in the statement. Hence, we are left with showing property 3, i.e. $\sum_{i \ge 1} e^{-\beta \tau^* _i } \xi^* _i \in L^2 (\mathbb P)$. We can proceed once more as in the proof of \cite[Prop. 4.7]{ABCCV20} and in the second part of Proposition \ref{prop:cons-best-response}'s proof. We provide all details for reader's convenience. First, let us denote $(\tau^\prime _k)_{k \ge 1}$ the increasing 
sequence of stopping times exhausting the intervention times of both players. Since the optimal impulses $(\xi_i ^*)_{i \ge 1}$ are uniformly bounded by some positive constant, say $\kappa$, we have
\begin{align*}
\mathbb E\left[ \left( \sum_{i \ge 1} \xi^*_i e^{-\beta \tau_i}\right)^2\right] & \le \kappa^2 \mathbb E\left[ \left( \sum_{k \ge 1} e^{-\beta \tau^\prime _k}\right)^2\right] \le \lim_{m \to\infty} 2 \mathbb E\left[ \sum_{1\le k \le  r \le m} e^{- \beta (\tau^\prime _r + \tau^\prime _k)}\right] \\
& \le \lim_{m \to\infty} 2\kappa^2 \mathbb E\left[ \sum_{1\le k \le m} e^{- 2\beta \tau^\prime _k }\right] = 2\kappa^2  \mathbb E\left[ \sum_{k \ge 1} e^{- 2\beta \tau^\prime _k }\right],
\end{align*}
hence it suffices to prove that $\sum_{k \ge 1} e^{- 2\beta \tau^\prime _k } \in L^1 (\mathbb P)$. Now, notice that $\tau^\prime _k$, $k \ge 1$, can be represented as  $\sum_{r=1}^k \zeta_r$, where $\zeta_r$ is a sequence of independent random variables distributed as the exit time, say $\zeta^{z,\mu}$, of one of the processes $z + \mu t + \sigma W_t$ with
\[ (z, \mu) \in \mathcal Z^\pm := \{ (y_h, \mu_+),(y_{\ell}, \mu_- ), (x_{h}^{- \ast} , \mu_-), ( x_\ell ^{+ \ast} , \mu_+)\} ,  \]
from the respective intervals
\[ (-\infty, y_h ), \quad (y_\ell, +\infty), \quad (-\infty, x_h ^-), \quad (x_\ell ^+ , +\infty). \]
Due to the independence of the sequence $\zeta_r$ we have
\begin{align*}
\mathbb E\left[ \sum_{k \ge 1} e^{- 2\beta \tau^\prime _k }\right]  = \sum_{k \ge 1} \prod_{r=1}^k \mathbb E\left[ e^{- 2\beta \zeta_r }\right] \le  \sum_{k \ge 1} \left(\mathbb E\left[ e^{- 2\beta \min_{(z,\mu) \in \mathcal Z^\pm}\zeta^{z,\mu} }\right]\right)^k ,
\end{align*}
which is a convergent geometric series, due to $\beta >0$ and the fact that $\zeta^{z,\mu} > 0$ almost surely for all $(z,\mu) \in \mathcal Z^\pm$. This shows that $(\tau^*_i , \xi_i ^*)_{i \ge 1}$ is an admissible producer's impulse strategy and concludes the proof.
\end{proof}

\begin{proof}[Proof of Proposition~\ref{prop:prod-best-response2}] Here, notice that $x_h ^+ = y_h$ so that, given producer's priority in case of simultaneous interventions (cf. Remark \ref{rmk:priority}), the drift is always equal to $\mu_+$ (recall that we are in the case $\mu_{0-} = \mu_+$). Hence, this proof can be performed as the one of Proposition~\ref{prop:prod-best-response1}, by ignoring the intervals where the drift is $\mu_-$ so that the second half in the RHS of inequality \eqref{eq:opt-prop3} is zero. The admissibility is proved in the same way. The details are therefore omitted.
\end{proof}

\subsection{Equilibrium Dynamics Computation}
Let $(a, b)$ be an arbitrary interval and $x \in (a,b)$ be an interior starting location. We define $\delta_+(x;a,b)$ to be the first passage time associated to the interval $(a, b)$ of a Brownian Motion with drift $\mu_+$ starting from $x$ and $P_+(x; a, b)$ to be the probability that this BM hits $a$ before $b$ (similarly for $\delta_-(x;a,b)$ and $P_-(x; a, b)$ associated to drift $\mu_-$). These quantities admit explicit expressions, see \cite{Borodin}.

The expected time $\tau_- := \inf \{ t : \mu_t = \mu_-\}$ for $\mu_t$ to switch from $\mu_+$ to $\mu_-$ within a double--switch and one--sided impulse equilibrium is then
\begin{align}\label{eq:time-to-switch}
\mathbb{E}\big[ \tau_-\big]=\mathbb{E}\big[\delta_+(x_0; x^+_\ell, y_h)\big]+\frac{P_+(x_0; x^+_\ell , y_h)}{\mathbf{P}_{I^+_h, S_-}}\mathbb{E}_+\big[\delta(x^{+\ast}_\ell; x^+ _\ell, y_h)\big],
\end{align}
where $\mathbf{P}$ is the transition matrix of $M^\ast_n$. Above, the first term denotes the time to either reach $x_\ell ^+$ (producer impulses up) or $y_h$ (switch to contraction); the second term counts the additional time if $x_\ell^+$ is reached first multiplied by the respective probability $P_+(x_0; x^+_\ell , y_h)$.
Let $\vec{\zeta}$ be the resulting vector of expected sojourn times. Then the long-run proportion of time that  $X^\ast$ carries  a positive drift ($\mu_+$) is
\begin{align}
\rho_+ =\frac{\Pi_{S_+}\zeta_{S_+}+\Pi_{I^+_\ell}\zeta_{I^+_\ell}+\Pi_{I^+_h}\zeta_{I^+_h}}{\vec{\Pi}\cdot\vec{\zeta}^\dagger},
\end{align}
and similarly the long-run proportion associated to a negative drift ($\mu_-$) is $\rho_-=1-\rho_+$.

\bibliographystyle{plain}

\end{document}